\documentclass[envcountsame]{llncs}
\usepackage[latin1]{inputenc}
\usepackage{amsmath}
\usepackage{amsfonts}
\usepackage{amssymb}
\usepackage{enumerate}
\usepackage{graphicx}
\usepackage{xcolor}
\usepackage{array}
\definecolor{violet}{rgb}{0.5,0,0.5}
\usepackage{mathrsfs}
\usepackage{bbm}
\usepackage{fourier-orns}
\usepackage{arydshln}
\usepackage{pgf,tikz,pgflibraryarrows}
\usepackage{pgflibrarysnakes,pgflibraryshapes}
\usepackage{xcolor}
\usepackage[disable]{todonotes}
\usepackage{stmaryrd}
\usepackage{xspace}
\usepackage{tabularx}
\usepackage{algorithm}
\usepackage{algpseudocode}
\usepackage{paralist}
\usepackage{url}
         
\definecolor{vertCitron}{rgb}{0.6,0.8,0.2}

\definecolor{fuchsia}{rgb}{0.82,0,0.37}

\definecolor{turquoise}{rgb}{0,0.40,0.55}

\definecolor{orange}{rgb}{0.93,0.43,0.06}
\newcommand{\oran}{\textcolor{orange}}
\definecolor{ciel}{rgb}{0,100,120}

\newcommand{\blan}{\textcolor{white}}

\newcommand{\sub}[1]{\ensuremath{Sub\left(#1\right)}}
\newcommand{\sem}[1]{\ensuremath{\left\llbracket#1\right\rrbracket}}
\newcommand{\dest}[1]{\ensuremath{Dest\left(#1\right)}}
\newcommand{\Tree}{tree-like\xspace}
\newcommand{\ATA}{OCATA\xspace}
\newcommand{\TATA}{TOCATA\xspace}

\newcommand{\nbclocks}[1]{\ensuremath{\left\|#1\right\|}}
\newcommand{\configs}[1]{\ensuremath{\mathrm{Config}\left(#1\right)}}
\newcommand{\Aa}{\ensuremath{\mathcal{A}}}
\newcommand{\Bb}{\ensuremath{\mathcal{B}}}
\newcommand{\Cc}{\ensuremath{\mathcal{C}}}
\newcommand{\Ee}{\ensuremath{\mathcal{E}}}
\newcommand{\Ii}{\ensuremath{\mathcal{I}}}
\newcommand{\Ss}{\ensuremath{\mathcal{S}}}
\newcommand{\Tt}{\ensuremath{\mathcal{T}}}

\newcommand{\Zz}{\ensuremath{\mathcal{Z}}}
\newcommand{\appf}[1]{\ensuremath{\mathit{APP}_{#1}}}
\newcommand{\timestep}[1]{\ensuremath{\overset{#1}{\rightsquigarrow}}}
\newcommand{\appfunc}[1]{\ensuremath{f_{#1}^\star}}
\newcommand{\mername}{\ensuremath{{\sf Merge}}\xspace}
\newcommand{\mer}[1]{\ensuremath{\mername\left(#1\right)}}
\newcommand{\re}[1]{\ensuremath{{\sf reset}\left(#1\right)}}
\newcommand{\gu}[1]{\ensuremath{{\sf guard}\left(#1\right)}}
\newcommand{\de}[1]{\ensuremath{{\sf loop}\left(#1\right)}}
\newcommand{\loc}{{\sf loc}}

\newcommand{\SubLTL}[1]{\ensuremath{Sub_{[0,+\infty)} \left(#1\right)}\xspace}

\newcommand{\R}{\mathbb{R}}
\newcommand{\N}{\mathbb{N}}

\def\ninf{\N^{+\infty}}
\def\dest{\ensuremath{{\sf dest}}}
\def\succ{\ensuremath{{\sf Succ}}}
\def\tts#1{\ensuremath{{\sf TTS}\left(#1\right)}}
\def\mhts#1{\ensuremath{{\sf MHTS}\left(#1\right)}}
\def\conf#1{\ensuremath{{\sf conf}\left(#1\right)}}

\def\cmax{\ensuremath{\text{c}_{\text{max}}}}

\usepackage[pdftex,colorlinks=true,bookmarks=true,bookmarksopen=true,bookmarksopenlevel=2,bookmarksdepth=2,citecolor=fuchsia]{hyperref}

\title{On MITL and alternating timed automata over infinite
  words\vspace*{-.2cm}}

\author{Thomas Brihaye${}^1$ \and Morgane
  Estiévenart${}^1$\thanks{This author has been supported by a FRIA
    scholarship.}  \and Gilles Geeraerts${}^2$}

\institute{${}^1$ UMons, Belgium, ${}^2$ U.L.B., Belgium\vspace*{-0.6cm}}

\begin{document}

\maketitle

\thispagestyle{plain}

\begin{abstract}
  \emph{One clock alternating timed automata} (\ATA) have been
  introduced as natural extension of (one clock) timed automata to
  express the semantics of MTL \cite{OW05}. In this paper, we consider
  the application of \ATA to the problems of model-checking and
  satisfiability for MITL (a syntactic fragment of MTL), interpreted
  over infinite words.  Our approach is based on the \emph{interval
    semantics} (recently introduced in \cite{BEG13} in the case of
  finite words) extended to infinite words. We propose region-based
  and zone-based algorithms, based on this semantics, for MITL
  model-checking and satisfiability. We report on the performance of 
  a prototype tool implementing those algorithms.
\end{abstract}

\section{Introduction}
\emph{Model-checking} \cite{CGP01} is today one of the most prominent
and successful techniques to establish \emph{automatically} the
correctness of a computer system. The system designer provides a
\emph{model-checker} with a \emph{model} of the system and a
\emph{formal property} that the system must respect. The model-checker
either proofs that the system respects the property, or outputs an
error trace that can be used for debugging. For their implementation,
many model-checkers rely on the so-called \emph{automata-based
  approach}, where the behaviours of the system and the set of
\emph{bad behaviours} are represented by the languages $L(B)$ and
$L(A_{\neg\varphi})$ of Büchi automata $B$ and $A_{\neg \varphi}$
respectively. Then, the model-checker performs automata-based
manipulations to check whether $L(B)\cap
L(A_{\neg\varphi})=\emptyset$.  While Büchi automata are adequate for
modeling systems, properties are more easily expressed by means of
logical sentences. The linear temporal logic (LTL for short) is
arguably one of the most studied logic to express such
requirements. Algorithms to turn an LTL formula $\varphi$ into a Büchi
automaton $A_\varphi$ recognising the same language are well-known,
thereby enabling its use in model-checkers.

Yet, this classical theory is not adequate for reasoning about
\emph{real-time} properties of systems, because Büchi automata and LTL
can only express \emph{sequence of events}, but have no notion of
(time) \emph{distance} between those events. Introduced by Alur and
Dill in 1994 \cite{AD94}, \emph{timed automata} (an extension of Büchi
automata with clocks, i.e. real variables that evolve synchronously)
are today the best accepted model for those real-time
systems. Symmetrically, several logics have been introduced to specify
real-time properties of systems. Among them, MITL (a syntactic
fragment of MTL \cite{K90}) is particularly appealing, because it
combines expressiveness \cite{AFH96} and tractability (MTL is mostly
undecidable \cite{AFH96}, while model-checking and satisfiability are
\textsc{ExpSpace}-c in MITL).
A comprehensive and efficient automata-based framework to support MITL
model-checking (and other problems such as satisfiability) is thus
highly desirable. In a recent work \cite{BEG13} we made a first step
towards this goal in the restricted case of \emph{finite words
  semantics}. We rely on \emph{one-clock alternating timed automata}
(\ATA for short), in order to avoid the direct, yet involved,
translation from MITL to timed automata first introduced in
\cite{AFH96}. The translation from MITL to \ATA\ -- which has been
introduced by Ouaknine and Worrell in the general case of MTL
\cite{OW05} -- is straightforward. However, the main difficulty with
alternating timed automata is that they cannot, in general, be
converted into an equivalent timed automaton, even in the one-clock
case. Indeed, a run of an alternating automaton can be understood as
several copies of the same automaton running in parallel on the same
word. Unfortunately, the clock values of all the copies are not always
synchronised, and one cannot bound, a priori, the number of different
clock values that one must track along the run. Hence, contrary to the
untimed word case, subset construction techniques cannot be directly
applied to turn an \ATA into a timed automaton (with finitely many
clocks).

Our solution \cite{BEG13} amounts to considering an alternative
semantics for \ATA, that we call the \emph{interval semantics}, where
clock valuations are not punctual values but intervals with real
endpoints. One of the features of this semantic is that several clock
values can be grouped into intervals, thanks to a so-called
\emph{approximation function}. For instance, consider a configuration
of an \ATA with three copies of the automata currently in the same
location $\ell$, and with clock values $0.42$, $1.2$ and $5.7$
respectively. It can be approximated by a single interval $[0.42,
5.7]$, meaning: `there are two copies witch clock values $0.42$ and
$5.7$, and there are \emph{potentially} several copies with clock
values in the interval $[0.42, 5.7]$'. This technique allows to reduce
the number of variables needed to track the clock values of the
\ATA. While this grouping yields an under-approximation of the
accepted language, we have showed \cite{BEG13} that, in the case of
finite words, and for an \ATA $A_\varphi$ obtained from an MITL
formula $\varphi$, one can always define an approximation function
s.t. the language of $A_\varphi$ is preserved and the number of intervals
along all runs is bounded by a constant $M(\varphi)$ depending on the
formula.  Using classical subset construction and tracking the
endpoints of each interval by means of a pair of clocks, we can then
translate the \ATA into a Büchi TA accepting the same language.

In the present work, we continue this line of research and demonstrate
that our techniques carry on to the infinite words case. Achieving
this result is not straightforward because \ATA on infinite words have
not been studied as deeply as in the finite words case, probably
because infinite words language emptiness of \ATA is decidable only on
restricted subclasses \cite{OW05,PW12}. Hence, to reach our goal, we
make several technical contributions regarding infinite words \ATA,
that might be of interest outside this work. First, in
Section~\ref{sec:an-interv-semant} we adapt the interval semantics of
\cite{BEG13} to the infinite words case. Then, in
Section~\ref{sec:from-mitl-timed}, we introduce \emph{tree-like} \ATA
(\TATA for short), a subclass of \ATA that exhibit some structure akin
to a tree (in the same spirit as the Weak and Very Weak Alternating Automata
\cite{KV01, GO01}).  For every MTL formula
$\varphi$, we show that the \ATA $A_\varphi$ obtained by the Ouaknine
and Worrell construction \cite{OW05} recognises the language of
$\varphi$ (a property that had never been established in the case of
infinite words, as far as we know\footnote{Even in \cite{OW05} where
  the authors consider a fragment of MTL over infinite words, but
  consider only safety properties that are reduced to questions on
  finite words.}), is in fact a \TATA. This shows in particular that
\TATA are semantically different from the `weak \ATA' introduced
in~\cite{PW12} (where `weak' refers to weak accepting conditions), 
and whose emptiness problem is decidable. We prove
specific properties of \TATA that are important in our constructions
(for instance, \TATA on infinite words can be easily complemented),
and we adapt the classical Miyano and Hayashi construction \cite{MH84}
to obtain a procedure to translate any \TATA $A_\varphi$ obtained from
an MITL formula into an equivalent timed Büchi automaton
$\Bb_\varphi$.  Equipped with these theoretical results, we propose in
Section~\ref{sec:experimental-results} algorithms to solve the
\emph{satisfiability} and \emph{model-checking} problems of MITL. We
define region-based and zone-based \cite{ADOQW08} versions of our
algorithms. Our algorithms work \emph{on-the-fly} in the sense that
they work directly on the structure of the \ATA $\Aa_\varphi$ (whose
size is linear in the size of $\varphi$), and avoid building
$\Bb_\varphi$ beforehand (which is, in the worst case, exponential in
the size of $\varphi$).  Finally in
Section~\ref{sec:experimental-results}, we present prototype tools
implementing those algorithms. To the best of our knowledge, these are
the first tools solving those problems for the full MITL. We report on
and compare their performance against a benchmark of MITL formulas
whose sizes are parametrised. While still preliminary, the results are
encouraging.

\section{Preliminaries}
\paragraph{Basic notions.} Let $\R$, $\R^+$ and $\N$ denote the sets of
real, non-negative real and natural numbers respectively. We call
\textbf{\textit{interval}} a convex subset of $\R$. We rely on the
classical notation $\langle a,b\rangle$ for intervals, where $\langle$
is $($ or $[$, $\rangle$ is $)$ or $]$, $a\in\R$ and
$b\in\R\cup\{+\infty\}$. For an interval $I=\langle a,b\rangle$, we
let $\inf(I)=a$ be the \emph{infimum} of $I$, $\sup(I)=b$ be its
\emph{supremum} ($a$ and $b$ are called the \emph{endpoints} of $I$)
and $\vert I \vert = \sup(I) - \inf(I)$ be its \emph{length}.  We note
$\mathcal{I}(\R)$ the set of all intervals.  We note
$\mathcal{I}(\R^{+})$ (resp. $\mathcal{I}(\ninf)$) the set of all
intervals whose endpoints are in $\R^{+}$ (resp. in $\N \cup \lbrace
+\infty \rbrace$).  Let $I \in \mathcal{I}(\R)$ and $t \in \R$, we
note $I+t$ for $\lbrace i+t \in\R \mid i \in I \rbrace$.  Let $I$ and
$J$ be two intervals, we let $I < J$ iff $\forall i \in I, \forall j
\in J : i < j$.

Let $\Sigma$ be a finite alphabet. An \emph{infinite word} on a set $S$ is an
infinite sequence $s=s_1 s_2s_3 \ldots$ of elements in $S$.  An infinite time sequence $\bar{\tau} =
\tau_{1} \tau_{2} \tau_{3} \ldots$ is an infinite word on $\R^+$
s.t. $\forall i \in \N, \tau_{i} \leq \tau_{i+1}$. An \emph{infinite timed word} over $\Sigma$ is a pair $\theta = (\bar{\sigma},\bar{\tau})$ where $\bar{\sigma}$ is an infinite word over $\Sigma$, $\bar{\tau}$ an infinite time sequence. We also note $\theta$ as
$(\sigma_{1},\tau_{1}) (\sigma_{2},\tau_{2}) (\sigma_{3},\tau_{3})
\ldots$. We denote by $T\Sigma^\omega$ the set of all infinite timed words. A \emph{timed language} is a (possibly infinite) set of infinite timed words.

\paragraph{Metric Interval Time Logic.} Given a finite alphabet
$\Sigma$, the formulas of MITL are defined by the following grammar,
where $\sigma \in \Sigma$, $I \in \mathcal{I}(\ninf)$ is non-singular:
\begin{center}
$\varphi$ := $\; \top \quad \vert \quad \sigma \quad \vert \quad \varphi_{1} \wedge \varphi_{2} \quad \vert \quad \neg \varphi \quad \vert \quad \varphi_{1} U_{I} \varphi_{2}$.
\end{center}
We rely on the following usual shortcuts $\lozenge_{I} \varphi$ stands
for $\top U_{I} \varphi$, $\qed_{I} \varphi$ for $\neg \lozenge_{I}
\neg \varphi$, $\varphi_{1} \tilde{U}_{I} \varphi_{2}$ for $\neg (
\neg \varphi_{1} U_{I} \neg \varphi_{2})$, $\qed \varphi$ for
$\qed_{[0,\infty)} \varphi$ and $\lozenge \varphi$ for
$\lozenge_{[0,\infty)} \varphi$.

Given an MITL formula $\varphi$, we note $Sub(\varphi)$ the set of all
subformulas of $\varphi$, i.e.: $\sub{\varphi}=\{\varphi\}$ when
$\varphi\in\{\top\}\cup\Sigma$,
$\sub{\neg\varphi}=\{\neg\varphi\}\cup\sub{\varphi}$ and
$\sub{\varphi}=\{\varphi\}\cup\sub{\varphi_1}\cup\sub{\varphi_2}$ when
$\varphi=\varphi_1U_I\varphi_2$ or $\varphi=\varphi_1\wedge\varphi_2$. We
let $|\varphi|$ denote the \emph{size of $\varphi$}, defined as the
\emph{number} of $U$ or $\tilde{U}$ modalities it contains.

\begin{definition}[Semantics of MITL] Given an infinite timed word $\theta =
  (\bar{\sigma},\bar{\tau})$ over $\Sigma$, a position $i\in \N_{0}$ and an 
  MITL formula $\varphi$, we say that $\theta$
  satisfies $\varphi$ from position $i$, written $(\theta, i) \models
  \varphi$ iff the following holds:

  \noindent\begin{tabularx}{\textwidth}{ll}
    $(\theta,i) \models \top $ &
    $(\theta,i) \models \varphi_{1}
    \wedge \varphi_{2}$ iff  $(\theta,i) \models \varphi_{1}$ and
    $(\theta,i) \models \varphi_2$ \\
    $(\theta,i) \models \sigma$ iff $\sigma_{i} = \sigma$\phantom{aaaaaaaa} &
    $(\theta,i) \models \neg \varphi$ iff $(\theta,i) \nvDash \varphi$ \\
    \multicolumn{2}{X}{%
      $(\theta,i) \models \varphi _{1}U_{I} \varphi _{2}$ iff
      $\exists j \geq i$: $(\theta,i)
      \models \varphi _{2}$, $\tau_{j}-\tau_{i} \in I$ $\wedge$ $\forall
      i \leq k < j$: $(\theta,k) \models \varphi_{1}$
   }
  \end{tabularx}
  We say that \textbf{\textit{$\theta$ satisfies $\varphi$}}, written
  $\theta \models \varphi$, iff $(\theta,1) \models \varphi$.  We note
  $\sem{\varphi}$ the timed language $\{\theta \mid \theta \models \varphi\}$.
\end{definition}
Observe that, for every MITL formula $\varphi$, $\sem{\varphi}$ is a timed
language and that we can transform any MITL formula in an equivalent
MITL formula in \emph{negative normal form} (in which negation can
only be present on letters $\sigma \in \Sigma$) using the operators:
$\wedge, \vee, \neg, U_{I}$ and $\tilde{U}_{I}$.

\begin{example}
  We can express the fact that `\textit{every occurrence of $p$ is
    followed by an occurrence of $q$ between 2 and 3 time units
    later}' by: $\Box (p \Rightarrow \lozenge_{[2,3]} q)$. Its
  negation, $\neg \big(\Box (p \Rightarrow \lozenge_{[2,3]} q )
  \big)$, is equivalent to the following negative normal form formula:
  $\top U_{[0,+\infty)} ( p \wedge \perp \tilde{U}_{[2,3]} \neg q )$.
\end{example}

\paragraph{Alternating timed automata.} One-clock alternating timed
automata (\ATA for short) have been introduced by Ouaknine and Worrell
to define the language of MTL formulas \cite{OW07}. We will rely on
\ATA to build our automata based framework for MITL. Let $\Gamma(L)$
be a set of formulas of the form $\top$, or $\perp$, or $\gamma_{1}
\vee \gamma_{2}$ or $\gamma_{1}$ $\wedge$ $\gamma_{2}$ or $\ell$ or
$x \bowtie c$ or $x.\gamma$, with
$c \in \N$, ${\bowtie} \in \lbrace <, \leq, >, \geq \rbrace$, $\ell
\in L$.  We call $x \bowtie c$ a \emph{clock constraint}.  Then, a
\emph{one-clock alternating timed automaton} (\ATA) \cite{OW07} is a
tuple $\mathcal{A} = (\Sigma, L, \ell_{0}, F, \delta)$ where $\Sigma$
is a finite alphabet, $L$ is a finite set of locations, $\ell_{0}$ is
the initial location, $F \subseteq L$ is a set of accepting locations,
$\delta : L \times \Sigma \rightarrow \Gamma(L)$ is the transition
function. Intuitively, disjunctions in $\delta(\ell)$ model
non-determinism, conjunctions model the creation of several copies of
the automata running in parallel (that must all accept for the word to
be accepted) and $x.\gamma$ means that $x$ (the clock of the \ATA) is
reset when taking the transition.

We assume that, for all $\gamma_1$, $\gamma_2$ in $\Gamma(L)$:
$x.(\gamma_{1} \vee \gamma_{2}) = x.\gamma_{1} \vee x.\gamma_{2}$,
$x.(\gamma_{1} \wedge \gamma_{2}) = x.\gamma_{1} \wedge x.\gamma_{2}$,
$x.x.\gamma = x.\gamma$, $x.(x \bowtie c) = 0 \bowtie c$, $x.\top =
\top$ and $x.\perp = \perp$. Thus, we can write any formula of
$\Gamma(L)$ in disjunctive normal form, and, from now on, we assume
that $\delta (\ell,\sigma)$ is written in disjunctive normal
form. That is, for all $\ell$, $\sigma$, we have $\delta (\ell,\sigma)
= \underset{j}{\bigvee} \underset{k}{\bigwedge} A_{j,k}$, where each
term $A_{j,k}$ is of the form $\ell$, $x.\ell$, $x \bowtie c$ or $0
\bowtie c$, with $\ell \in L$ and $c \in \N$.  We call \emph{arc} of
the \ATA $\mathcal{A}$ a triple $(\ell, \sigma, \bigwedge_k A_{j,k})$
s.t. $\bigwedge_k A_{j,k}$ is a disjunct in $\delta (\ell,\sigma)$.
For an arc $a = (\ell, \sigma, \bigwedge_k A_{j,k})$, we note 
$Start(a) = \ell$, $Dest(a) = \{ \ell' \; \vert \; \ell' \text{ is 
a location present in } \bigwedge_k A_{j,k}\}$ and we say that $a$ is
labeled by $\sigma$.

\begin{example}
\label{ex1}
Fig.~\ref{ExGilles} (left) shows an \ATA $\Aa_{\varphi}=\big\{\Sigma,
\{\ell_\Box,\ell_\lozenge\}, \ell_\Box, \{\ell_\Box\}, \delta\big\}$,
over the alphabet $\Sigma = \lbrace a, b \rbrace$, and with transition
function: $\delta (\ell_\Box, a) = \ell_\Box \wedge x.\ell_\lozenge$,
$\delta (\ell_\Box, b) = \ell_\Box$, $\delta (\ell_\lozenge, a) =
\ell_\lozenge$ and $\delta (\ell_\lozenge, b) = \ell_\lozenge \vee ( x
\geq 1 \wedge x \leq 2 )$.
We depict a conjunctive transition such as $\delta (\ell_\Box, a) =
\ell_\Box \wedge x.\ell_\lozenge$ by an arrow splitting in two
branches connected to $\ell_\Box$ and $\ell_\lozenge$ (they might have
different resets: the reset of clock $x$ is depicted by $x:=0$).
Intuitively, when reading an $a$ from $\ell_\Box$ with clock value
$v$, the automaton starts \emph{two copies of itself}, the former in
location $\ell_\Box$, with clock value $v$, the latter in location
$\ell_\lozenge$ with clock value $0$.  Both copies should accept the
suffix for the word to be accepted. The edge labeled by $b,x\in[1,2]$
from $\ell_\lozenge$ has no target location: it depicts the fact that,
when the automaton has a copy in location $\ell_\lozenge$ with a clock
valuation in $[1,2]$, the copy accepts all further suffixes and can
thus be \emph{removed} from the automaton.
\end{example}

\section{The intervals semantics for \ATA on infinite words\label{sec:an-interv-semant}}

In this section, we adapt to infinite timed words the intervals
semantics introduced in \cite{BEG13}. In this semantics,
\emph{configurations} are sets of \emph{states} $(\ell, I)$, where
$\ell$ is a location of the \ATA and $I$ is an \emph{interval} (while
in the standard semantics states are pairs $(\ell, v)$, where $v$ is
the valuation of the clock). Intuitively, a state $(\ell, I)$ is an
abstraction of a set of states of the form $(\ell,v)$ with $v\in I$.

Formally, a \emph{state} of an \ATA $\mathcal{A} = (\Sigma, L,
\ell_{0}, F, \delta)$ is a pair $(\ell,I)$ where $\ell \in L$ and $I
\in \mathcal{I}(\R^{+})$.  We note $S = L \times \mathcal{I}(\R^{+})$
the state space of $\mathcal{A}$.
When $I = [v,v]$ (sometimes denoted $I=\{v\}$), we shorten $(\ell,I)$
by $(\ell,v)$. A \emph{configuration} of an \ATA $\mathcal{A}$ is a
(possibly empty) finite set of states of $\mathcal{A}$ in which all
intervals associated with a same location are disjoint.  In the rest
of the paper, we sometimes see a configuration $C$ as a function from
$L$ to $2^{\Ii(\R^+)}$ s.t. for all $\ell\in L$: $C(\ell)=\{I\mid
(\ell,I)\in C\}$. We note $\configs{\Aa}$ the set of all configurations of
$\mathcal{A}$.  The \emph{initial configuration} of $\mathcal{A}$ is
$\lbrace(\ell_{0},0)\rbrace$.  For a configuration $C$ and a delay $t
\in \R^{+}$, we note $C+t$ the configuration $\lbrace (\ell,I+t) \mid
(\ell,I) \in C \rbrace$. From now on, we assume that, for all
configurations $C$ and all locations $\ell$: when writing $C(\ell)$ as
$\{I_1,\ldots, I_m\}$ we have $ I_{i} < I_{i+1}$ for all $1 \leq i <
m$. Let $E$ be a finite set of intervals from
$\mathcal{I}(\R^{+})$. We let $\nbclocks{E}=|\{[a,a]\in E\}|+2\times
|\{I\in E\mid \inf(I)\neq\sup(I)\}|$ denote the number of
\emph{individual clocks} we need to encode all the information present
in $E$, using one clock to track singular intervals, and two clocks to
retain $\inf(I)$ and $\sup(I)$ respectively for non-singular intervals
$I$. For a configuration $C$, we let $\nbclocks{C}=\sum_{\ell\in L}
\nbclocks{C(\ell)}$.

\paragraph{Interval semantics \cite{BEG13}.} Let $M\in\configs{\Aa}$ be a configuration of an
\ATA \Aa, and $I \in \mathcal{I}(\R^{+})$. We define the satisfaction
relation "$\models_{I}$" on $\Gamma(L)$~as:
$$
\begin{array}{ll}
  M \models_{I} \top\\
  M \models_{I} \gamma_{1} \wedge \gamma_{2} &\text{ iff } M \models_{I} \gamma_{1} \text{ and } M \models_{I} \gamma_{2}\\
  M \models_{I} \gamma_{1} \vee \gamma_{2} &\text{ iff } M \models_{I} \gamma_{1} \text{ or } M \models_{I} \gamma_{2}
\end{array}
\qquad
\begin{array}{ll}
  M \models_{I} \ell  &\text{ iff } (\ell,I) \in M\\
  M \models_{I} x \bowtie c &\text{ iff } \forall x \in I, x \bowtie c\\
  M \models_{I} x.\gamma &\text{ iff } M \models_{[0,0]} \gamma
\end{array}
$$
We say that a configuration $M$ is a \textit{minimal model} of the formula $\gamma
\in \Gamma(L)$ wrt the interval $I \in
\mathcal{I}(\R^{+}) \text{ iff } M \models_{I} \gamma$ and there is no
$M' \subsetneq M$ such that ${M' \models_{I} \gamma}$.
Intuitively, for $\ell \in L, \sigma \in \Sigma$ and $I \in
\mathcal{I}(\R^{+})$, a minimal model of $\delta(\ell,\sigma)$ wrt
$I$ represents a configuration the automaton can reach from
state $(\ell,I)$ by reading $\sigma$. Observe that the definition of
$M \models_{I} x \bowtie c$ only allows to take a transition
$\delta(\ell,\sigma)$ from state $(\ell,I)$ if all the values in $I$
satisfy the clock constraint $x \bowtie c$ of $\delta(\ell,\sigma)$.
We denote $\succ((\ell,I),\sigma) = \{ M\mid M \text{ is a minimal
  model of } \delta(\ell,\sigma) \text{ wrt } I \}$. We lift the
definition of $\succ$ to configurations $C$ as follows:
$\succ(C,\sigma)$ is the set of all configurations $C'$ of the form
$\cup_{s\in C} M_{s}$, where, for all $s\in C$:
$M_{s}\in\succ(s,\sigma)$. That is, each $C'\in\succ(C,\sigma)$ is
obtained by chosing one minimal model $M_s$ in $\succ(s,\sigma)$ for
each $s\in C$, and taking the union of all those $M_s$.

\begin{example}
  Let us consider again the \ATA of Fig.~\ref{ExGilles} (left), and
  let us compute the minimal models of $\delta(\ell_\lozenge,
  b)=\ell_\lozenge \vee ( x \geq 1 \wedge x \leq 2 )$ wrt to
  $[1.5,2]$. A minimal model of $\ell_\lozenge $ wrt $[1.5,2]$ is
  $M_1=\{(\ell_\lozenge, [1.5,2]) \}$. A minimal model of $( x \geq 1
  \wedge x \leq 2 )$ is $M_2=\emptyset$ since all values in $[1.5,2]$
  satisfy $( x \geq 1 \wedge x \leq 2 )$. As $M_2 \subseteq M_1$,
  $M_2$ is the unique minimal model of $\delta(\ell_\lozenge, b)$ wrt $[1.5,2]$:
  $\succ((\ell_\lozenge,[1.5,2]),b)=\{ M_2 \}$.
\end{example}

\paragraph{Approximation functions} Let us now recall the notion of
\emph{approximation functions} that associate with each configuration
$C$, a set of configurations that \emph{approximates} $C$ \emph{and
  contains less states} than $C$.  Formally, for an \ATA $\Aa$, an
\emph{approximation function} is a function $f:\configs{\Aa}\mapsto
2^{\configs{\Aa}}$ s.t. for all configurations $C$, for all $C'\in
f(C)$, for all locations $\ell\in L$: $(i)$
\begin{inparaenum}[(i)]
\item   $\nbclocks{C'(\ell)}\leq\nbclocks{C(\ell)}$;
\item for all $I\in
  C(\ell)$, there exists $J\in C'(\ell)$ s.t. $I \subseteq J$; and
\item for all $J \in C'(\ell)$, there are $I_1, I_2 \in C(\ell)$
  s.t. $\inf(J) = \inf(I_1)$ and $\sup(J) = \sup(I_2)$.
\end{inparaenum}
We note $\appf{\Aa}$ the set of approximation functions for~$\Aa$. We
lift all approximation functions $f$ to sets $\Cc$ of configurations
in the usual way: $f(\Cc)=\cup_{C\in \Cc}f(C)$.
In the rest of the paper we will rely mainly on approximation
functions that enable to \emph{bound} the number of clock copies in
all configurations along all runs of an \ATA $\Aa$.  Let $k \in \N$,
we say that $f \in\appf{\Aa}$ is a \emph{$k$-bounded approximation
  function} iff for all $ C \in \configs{\Aa}$, for all $C' \in f(C)$:
$\nbclocks {C'} \leq k$.

\paragraph{$f$-Runs of \ATA} We can now define formally the notion of
\emph{run} of an \ATA in the interval semantics. This notion will be
parametrised by an approximation function $f$, that will be used to
reduce the number of states present in each configuration along the
run. Each new configuration in the run is thus obtained in three
steps: letting time elapse, performing a discrete step, and applying
the approximation function.  Formally, let $\mathcal{A}$ be an \ATA of
state space $S$, $f\in\appf{\Aa}$ be an approximation function and
$\theta = (\sigma_{1},\tau_{1}) (\sigma_{2},\tau_{2}) \ldots
(\sigma_{i},\tau_{i})\ldots$ be an infinite timed word. Let us note
$t_{i} = \tau_{i} - \tau_{i-1}$ for all $i \geq 1$, assuming $\tau_{0}
= 0$. An \emph{$f$-run of $\Aa$ on $\theta$} is an infinite sequence
$C_0,C_1,\ldots,C_i,\ldots$ of configurations s.t.:
$C_0=\{(\ell_0,0)\}$ and for all $i\geq 1$: $C_{i}\in
f(\succ(C_{i-1}+t_i,\sigma_i))$. Observe that for all pairs of
configurations $C$, $C'$ s.t. $C'\in f(\succ(C+t,\sigma))$ for some
$f$, $t$ and $\sigma$, each $s\in C$ can be associated with a unique
set $\dest(C,C',s)\subseteq C'$ containing all the `successors' of $s$
in $C'$ and obtained as follows. Let
$\overline{C}\in\succ(C+t,\sigma)$ be s.t. $C'\in
f(\overline{C})$. Thus, by definition,
$\overline{C}=\cup_{\overline{s}\in C} M_{\overline{s}}$, where each
$M_{\overline{s}}\in\succ(\overline{s},\sigma)$ is the minimal model
that has been chosen for $\overline{s}$ when computing
$\succ(C,\sigma)$. Then, $\dest(C,C',s)=\{(\ell',J)\in C'\mid
(\ell',I)\in M_s\textrm{ and } I\subseteq J\}$. Remark that
$\dest(C,C',s)$ is well-defined because intervals are assumed to be
disjoint in configurations. The function $\dest$ allows to define a
DAG representation of runs, as is usual with alternating automata. We
regard a run $\pi=C_0,C_1,\ldots,C_i,\ldots$ as a rooted DAG
$G_\pi=(V,\rightarrow)$, whose vertices $V$ correspond to the states
of the \ATA (vertices at depth $i$ correspond to $C_i$), and whose set
of edges $\rightarrow$ expresses the \ATA transitions. Formally,
$V=\cup_{i\geq 0} V_i$, where for all $i\geq 0$: $V_i=\{(s,i)\mid s\in
C_i\}$ is the set of all vertices of depth $i$. The root of $G_\pi$ is
$((\ell_0,0),0)$. Finally, $(s_1,i_1)\rightarrow (s_2,i_2)$ iff
$i_2=i_1+1$ and $s_2\in\dest(C_{i-1}, C_i,s_1)$. From now on, we will
mainly rely on the DAG characterisation of $f$-runs.

\begin{example}
  Fig.~\ref{reprun} displays three DAG representation of run prefixes
  of $\Aa_\varphi$ (Fig.~\ref{ExGilles}), on the word
  $(a,0.1)(a,0.2)(a,1.9)(b,2)(b,3)\ldots$ (grey boxes highlight the
  successive configurations). $\pi$ only is an $Id$-run and shows why
  the number of clock copies cannot be bounded in general: if
  $\Aa_{\varphi}$ reads $n$ $a$'s between instants $0$ and $1$, $n$
  copies of the clock are created in location $\ell_{\lozenge}$.
\end{example}

Please find below an equivalent definition of an $f$-run in which the time elapsing and discrete transition (i.e. the acting of reading a letter) are distinguished. The definition below is inspired from the classical definition of run of an \ATA \cite{OW07} and is more convenient to manipulate in the proves of our results.
\begin{definition}
\label{DefTS}
  Let $\mathcal{A}$ be an \ATA and let $f\in\appf{\Aa}$ be an
  approximation function.  The \emph{$f$-semantics of $\Aa$} is the
  transition system $\mathcal{T}_{\mathcal{A}, f} = (\configs{\Aa},
  \rightsquigarrow,\break \longrightarrow_{f})$ on configurations of
  $\mathcal{A}$ defined as follows:
  \begin{itemize}
  \item the transition relation $\rightsquigarrow$ takes care of the
    elapsing of time: $\forall t\in \R^{+}, C
    \overset{t}{\rightsquigarrow} C' \text{ iff } C' = C+t$. We let
    ${\rightsquigarrow} = \underset{t \in \R^{+}}{\bigcup}
    \overset{t}{\rightsquigarrow}$.
  \item the transition relation $\longrightarrow$ takes care of
    discrete transitions between locations and of the approximation:
    $C \overset{\sigma}{\longrightarrow} C'$ iff 
    $C' \in f(Succ(C,\sigma))$. 
    We let $\longrightarrow_{f} = \underset{\sigma \in
      \Sigma}{\bigcup} \overset{\sigma}{\longrightarrow_{f}}$.
  \end{itemize}
\end{definition}
Let $\theta = (\bar{\sigma},\bar{\tau})$
be an infinite timed word and let $f\in\appf{\Aa}$ be
an approximation function.  Let us note $t_{i} = \tau_{i} -
\tau_{i-1}$ for all $i \geq 1$, assuming
$\tau_{0} = 0$.  An \emph{$f$-run} of $\Aa$ (of state space $S$) on $\theta$ is also an infinite
sequence of discrete and continuous transitions in $\Tt_{\Aa,f}$ that
is labelled by $\theta$, i.e. a sequence of the form: $ C_{0}
\overset{t_{1}}{\rightsquigarrow} C_{1}
\overset{\sigma_{1}}{\longrightarrow_{f}} C_{2}
\overset{t_{2}}{\rightsquigarrow} C_{3}
\overset{\sigma_{2}}{\longrightarrow_{f}}
\dots \overset{t_{i}}{\rightsquigarrow} C_{2i-1}
\overset{\sigma_{i}}{\longrightarrow_{f}} C_{2i} \dots$.\\
In the rest of the paper,
we (sometimes) use the abbreviation $C_{i} \overset{t,
  \sigma}{\longrightarrow_{f}} C_{i+2}$ for $C_{i}
\overset{t}{\rightsquigarrow} C_{i+1} = C_{i}+t
\overset{\sigma}{\longrightarrow_{f}} C_{i+2}$: the sequence of configurations $C_0, C_2, C_4, \dots,$ $C_{2i}, \dots$ corresponds to the sequence of configurations of depth $0, 1, 2, \dots, i, \dots$ (respectively) in the DAG defined previously, i.e. $V_0, V_1, V_2, \dots, V_i, \dots$.  On Fig.~\ref{reprun}, the $Id$-run $\pi$ of the automaton $\Aa_\varphi$ of Fig.~\ref{ExGilles} is represented as a DAG and as a sequence of $\overset{\sigma}{\longrightarrow_{Id}}$ transitions (above) and as a DAG (below).  This example should convince the reader of the equivalence of these definitions.

\paragraph{$f$-language of \ATA} We can now define the accepted
language of an \ATA, parameterised by an approximation function $f$. A
\emph{branch} of an $f$-run $G$ is a (finite or) infinite path in
$G_\pi$. We note $Bran^\omega(G)$ the set of all \emph{infinite}
branches of $G_\pi$ and, for a branch $\beta$, we note $Infty(\beta)$
the set of locations occurring infinitely often along $\beta$.  An
$f$-run is \emph{accepting} iff $\forall \beta \in Bran^\omega(G)$,
$Infty(\beta) \cap F \neq \emptyset$ (i.e. we consider B\"uchi 
acceptance condition)\todo{T: Parenthese OK? On ne disait nulle part 
explicitement que notre condition d acceptation etait Buchi. Vu la 
footnote sur co-Buchi, ca me semblait bien de l'ajouter.}.
We say that an infinite timed
word $\theta$ is $f$-accepted by $\mathcal{A}$ iff there exists an
accepting $f$-run of $\mathcal{A}$ on $\theta$.  We note
$L^{\omega}_{f}(\mathcal{A})$ the language of all infinite timed words
$f$-accepted by $\mathcal{A}$.  We close the section by observing that
a standard semantics for \ATA (where clock valuations are punctual
values instead of intervals) is a particular case of the interval
semantics, obtained by using the approximation function $Id$
s.t. $Id(C) = \{C\}$ for all $C$. We denote by $L^\omega(\Aa)$ the
language $L^\omega_{Id}(\Aa)$.  Then, the following proposition shows
the impact of approximation functions on the accepted language of the
\ATA: they can only lead to \emph{under-approximations} of
$L^\omega(\Aa)$.

\begin{proposition}
\label{inclu}
For all \ATA $\Aa$, for all $f \in \appf{\Aa}$: $L^{\omega}_{f}(\Aa)
\subseteq L^{\omega}(\Aa)$.
\end{proposition}
\begin{proof}[Idea]
  In $Id$-runs, all clock values are punctual, while in $f$-runs,
  clock values can be non-punctual intervals. Consider a set
  $(\ell,v_1),\ldots, (\ell,v_n)$ of states in location $\ell$ and
  with punctual values $v_1\leq\ldots \leq v_n$, and consider its
  approximation $s=(\ell [v_1,v_n])$. Then, if a $\sigma$-labeled
  transition is firable from $s$, it is also firable from all
  $(\ell,v_i)$. The converse is not true: there might be a set of
  $\sigma$-labeled transitions that are firable from each
  $(\ell,v_i)$, but no $\sigma$-labeled transition firable from $s$,
  because \emph{all clock values} in $I$ must satisfy the transition
  guard.\qed
\end{proof}

\section{\TATA: a class of \ATA for MITL\label{sec:from-mitl-timed}}
In this section, we introduce the class of \emph{tree-like} \ATA
(\TATA for short), and show that, when applying, to an MITL formula
$\varphi$, the construction defined by Ouaknine and Worrell
\cite{OW07} in the setting of MTL interpreted on \emph{finite words},
one obtains a \TATA that accepts the \emph{infinite words} language of
$\varphi$. To prove this result, we rely on the specific properties of
\TATA (in particular, we show that their acceptance condition can be
made simpler than in the general case). Then, we show that there is a
family of \emph{bounded approximation functions} $f^\star_\varphi$, s.t.,
for every MITL formula $\varphi$,
$L^{\omega}_{\appfunc{\varphi}}(\Aa_\varphi)=L^{\omega}(\Aa_\varphi)$. This
result will be crucial to the definition of our on-the-fly
model-checking algorithm in Section~\ref{sec:mitl-model-checking}. We
also exploit it to define a natural procedure that builds, for
all MITL formula $\varphi$, a \emph{Büchi timed automaton $\Bb_\varphi$}
accepting $\sem{\varphi}$.

\paragraph{From MITL to \ATA.} We begin by recalling the syntactic
translation from MTL (a superset of MITL) to \ATA, as defined by
Ouaknine and Worrell \cite{OW07}. Observe that it has been defined in
the setting of \emph{finite words}, hence we will need to prove that
it is still correct in the infinite words setting. Let $\varphi$ be an
MITL formula (in negative normal form). We let $\mathcal{A}_{\varphi}
= (\Sigma, L, \ell_{0}, F, \delta)$ where: $L$ is the set containing
the initial copy of $\varphi$, noted `$\varphi_{init}$', and all the
formulas of $Sub(\varphi)$ whose outermost connective is `$U$' or
`$\tilde{U}$'; $\ell_{0} = \varphi_{init}$; $F$ is the set of the
elements of $L$ of the form $\varphi_{1} \tilde{U}_{I}
\varphi_{2}$. Finally $\delta $ is defined by induction on the
structure of $\varphi$:
\begin{itemize}
\item $\delta(\varphi_{init}, \sigma) = x.\delta(\varphi, \sigma)$
\item $\delta(\varphi_{1} \vee \varphi_{2}, \sigma) =
  \delta(\varphi_{1}, \sigma) \vee \delta(\varphi_{2}, \sigma)$;  $\delta(\varphi_{1} \wedge \varphi_{2}, \sigma) =
  \delta(\varphi_{1}, \sigma) \wedge \delta(\varphi_{2}, \sigma)$
\item $\delta(\varphi_{1}U_{I} \varphi_{2}, \sigma) =
  (x.\delta(\varphi_{2},\sigma) \wedge x \in I) \vee
  (x.\delta(\varphi_{1},\sigma) \wedge \varphi_{1}U_{I} \varphi_{2}
  \wedge x\leq sup(I))$
\item $\delta(\varphi_{1} \tilde{U}_{I} \varphi_{2}, \sigma) =
  (x.\delta(\varphi_{2},\sigma) \vee x\notin I) \wedge
  (x.\delta(\varphi_{1},\sigma) \vee \varphi_{1} \tilde{U}_{I}
  \varphi_{2} \vee x > sup(I))$
\item $\forall \sigma_1,\sigma_2\in \Sigma$: $\delta(\sigma_{1},
  \sigma_{2}) = \left\{
    \begin{array}{ll}
      \text{true} &\text{if } \sigma_{1} \text{=} \sigma_{2} \\
      \text{false}&\text{if } \sigma_{1} \neq \sigma_{2}
    \end{array}
  \right.$ and $\delta(\neg \sigma_{1}, \sigma_{2}) = \left\{
    \begin{array}{ll}
      \text{false} &\text{if } \sigma_{1} \text{=} \sigma_{2} \\
      \text{true} &\text{if } \sigma_{1} \neq \sigma_{2}
    \end{array}
  \right.$
\item $\forall \sigma\in\Sigma$: $\delta(\top,\sigma)=\top$ and
  $\delta(\bot,\sigma)=\bot$.
\end{itemize}
\begin{example}
  As an example, consider again the \ATA $\Aa_\varphi$ in
  Fig.~\ref{ExGilles}. It accepts exactly $\sem{\qed ( a \Rightarrow
    \lozenge_{[1,2]} b)}$. It has been obtained by means of the above
  construction (trivial guards such as $x\leq +\infty$, and the state
  $\varphi_{init}$ have been omitted).
\end{example}

\paragraph{Tree-like \ATA} Let us now define a strict subclass of \ATA
that captures all the infinite words language of MTL formulas, but
whose acceptance condition can be made simpler. An \ATA $\mathcal{A} =
(\Sigma, L, \ell_{0}, F, \delta)$ is a \TATA iff there exists a
partition $L_1, L_2,\ldots, L_m$ of $L$ and a partial order
$\preccurlyeq$ on the sets $L_1, L_2, \dots, L_m$ s.t.:
\begin{inparaenum}[(i)]
\item each $L_i$ contains either only accepting states or no accepting
  states and
\item the partial order $\preccurlyeq$ is compatible with the
  transition relation and yields the `tree-like' structure of the
  automaton.
\end{inparaenum}
Here is the formal definition of a \TATA.

\begin{definition} An \ATA $\mathcal{A} =
(\Sigma, L, \ell_{0}, F, \delta)$ is a \TATA iff there exists a
partition $L_1, L_2,\ldots, L_m$ of $L$ satisfying:
\begin{itemize}
\item each $L_i$ contains either only accepting states or no accepting
  states: $\forall 1 \leq i \leq m$ either $L_i \subseteq F$ or $L_i \cap F = \emptyset$, and
\item there is a partial order $\preccurlyeq$ on the sets $L_1, L_2, \dots, L_m$ compatible with the
  transition relation and that yields the `tree-like' structure of the
  automaton in the following sense: $\preccurlyeq$ is s.t. $L_j \preccurlyeq L_i$ iff
 $\exists \sigma \in \Sigma$, $\ell \in L_i$ and $\ell' \in L_j$ such that $\ell'$ is present in $\delta(\ell,\sigma)$.
\end{itemize}
\end{definition}
In particular, \ATA built from \emph{MTL} formulas, such as
$\Aa_\varphi$ in Fig.~\ref{ExGilles}, are \TATA. Since MTL is a
superset of MITL, this proposition is true in particular for MITL
formulas:
\begin{proposition}
  For every {\em MTL} formula $\varphi$, $\Aa_\varphi$ is a \TATA.
\end{proposition}
\begin{proof}
  Let $L=\lbrace \ell_1, \ell_2, \dots, \ell_m \rbrace$ be the
  locations of $\Aa_\varphi$. We consider the partition $\lbrace
  \ell_1 \rbrace, \lbrace \ell_2 \rbrace, \dots, \lbrace \ell_m
  \rbrace$ of $L$ and the order $\preccurlyeq$ s.t.  $\{\ell_j\}
  \preccurlyeq \{\ell_i\}$ iff $\ell_j$ is a subformula of
  $\ell_i$. It is easy to check that they satisfy the definition of
  \TATA.\qed
\end{proof}

\paragraph{Properties of \TATA} Let us now discuss two peculiar
properties of \TATA that are not enjoyed by \ATA. The first one is
concerned with the acceptance condition. In the general case, a run of
an \ATA is accepting iff all its branches visit accepting states
infinitely often. Thanks to the partition characterising a \TATA, this
condition can be made simpler: a run is now accepting iff each branch
eventually visits accepting states \emph{only}, because it reaches a
partition of the locations that are all accepting.

\begin{proposition}
\label{simpleAC}
An $Id$-run $G_\pi$ of a \TATA $\Aa = (\Sigma, L, \ell_{0}, F, \delta)$ is \emph{accepting} iff $\forall \beta = \beta_{0} \beta_{1} \ldots \beta_{i} \ldots \in Bran^\omega(G_\pi)$, $\exists n_\beta \in
\N$ s.t. $\forall i > n_{\beta}$: $\beta_i=((\ell,v),i)$ implies $\ell\in F$.
\end{proposition}
\begin{proof}
First, remark that the "if" case is trivial.  In the following, we prove the "only if" case. As $\Aa$ is a \Tree \ATA, there exists a partition of $L$ into disjoint subsets $L_1, L_2, \dots, L_m$ satisfying:
\begin{enumerate}
\item $\forall 1 \leq i \leq m$ either $L_i \subseteq F$ or $L_i \cap F = \emptyset$, and
\item there is a partial order $\preccurlyeq$ on the sets $L_1, L_2, \dots, L_m$ such that $L_j \preccurlyeq L_i$ iff
 $\exists \sigma \in \Sigma$, $\ell \in L_i$ and $\ell' \in L_j$ such that $\ell'$ is present in $\delta(\ell,\sigma)$.
\end{enumerate}
Let $\pi$ be an accepting $Id$-run of $\mathcal{A}$ and $G_\pi$ its associated DAG.  Let $\beta = \beta_{0} \beta_{1} \ldots \beta_{i} \ldots$ be a (finite or infinite) branch of $G_\pi$, we must prove that either $\beta$ is finite, either $\exists n_\beta \in \N$ s.t. $\forall i > n_{\beta}$: $\beta_i=((\ell,v),i)$ implies $\ell\in F$.  As $\pi$ is accepting, either $\beta$ is finite, either there exists a smallest $n_0$ such that $\beta_{n_0}$ has an accepting location.  If $\beta$ is finite, we are done.  So, let us suppose that $\beta$ is infinite. In this case, there exists a smallest $n_0$ such that $\beta_{n_0}$ has an accepting location, say $\ell \in L_i$ for a certain $1 \leq i \leq m$. 1. implies that $L_i \subseteq F$. Thanks to 2., $\beta_{n_0 +1}$ is found:
\begin{enumerate}
\item[a.] taking an arc looping on $L_i$, or
\item[b.] taking an arc going to $L_j$, for a certain $1 \leq j \leq m$ with $L_j \preccurlyeq L_i$ and $L_i \neq L_j$.
\end{enumerate}
Remark that case a. can be repeated infinitely many times, while case b. can happen at most $m-1$ times.  So, there exists $n^\star \geq n_0$ and $1 \leq i^\star \leq m$ such that $\forall i > n^\star$: $\beta_i=((\ell,v),i)$ implies $\ell\in L_{i^\star}$.  As $\pi$ is an accepting $Id$-run, 1. implies that $L_{i^\star} \subseteq F$ and $n^\star$ is the research $n_\beta$: $\forall i > n^\star$, $\beta_i=((\ell,v),i)$ implies $\ell\in F$.\qed
\end{proof}

The second property of interest for us is that \TATA can be easily
complemented. One can simply swap accepting and non-accepting
locations, and `dualise' the transition relation, \emph{without
  changing the acceptance condition}\footnote{In general, applying
  this construction yields an \ATA \emph{with co-Büchi acceptance
    condition} for the complement of the language.} (as in the case of
\ATA on finite words \cite{OW07}).  Formally, the dual of a formula $\varphi \in \Gamma(L)$
is the formula $\overline{\varphi}$ defined inductively as
follows. $\forall \ell \in L$, $\overline{\ell} = \ell$ ;
$\overline{false} = true$ and $\overline{true} = false$ ;
$\overline{\varphi_1 \vee \varphi_2} = \overline{\varphi_1} \wedge \overline{\varphi_2}$ ;
$\overline{\varphi_1 \wedge \varphi_2} = \overline{\varphi_1} \vee \overline{\varphi_2}$ ;
$\overline{x.\varphi} = x.\overline{\varphi}$ ; the dual of a clock
constraint is its negation (for example: $\overline{x \leq c } = x >
c$). Then, for all \TATA $\Aa = (\Sigma, L, \ell_0, F, \delta)$, we
let $\Aa^C = (\Sigma, L, \ell_0, L \setminus F, \overline{\delta})$
where $\overline{\delta}(\ell, \sigma) =
\overline{\delta(\ell,\sigma)}$. Thanks to Proposition~\ref{simpleAC},
we will prove that $\Aa^C$ accepts the complement of $\Aa$'s language.

Before proving this, we make several useful
observations about the transition relation of an \ATA. Let $\delta$ be
the transition function of some \ATA, let $\ell$ be a location, let
$\sigma$ be a letter, and assume $\delta(\ell,\sigma)=\bigvee_k a_k$,
where each $a_k$ is an \emph{arc}, i.e. a conjunction of atoms of the
form: $\ell'$, $x.\ell'$, $x\bowtie c$, $0\bowtie c$, $\top$ or
$\bot$. Then, we observe that \emph{each minimal model} of
$\delta(\ell,\sigma)$ wrt some interval $I$ corresponds to
\emph{firing one of the arcs $a_k$ from $(\ell,I)$}. That is, each
minimal model can be obtained by choosing an arc $a_k$ from
$\delta(\ell,\sigma)$, and applying the following procedure. Assume
$a_k = \ell_1\wedge\cdots\wedge \ell_n\wedge
x.(\ell_{n+1}\wedge\cdots\wedge \ell_{m})\wedge \varphi$, where
$\varphi$ is a conjunction of clock constraints. Then, $a_k$ is
firable from a minimal model $(\ell,\sigma)$ iff $I\models \varphi$
(otherwise, no minimal model can be obtained from $a_k$). In this
case, the minimal model is $\{(\ell_i,I) \vert 1\leq i\leq n\} \cup
\{(\ell_i,[0,0])\vert n+1\leq i\leq m\}$.\\
From now on, we consider that $\delta(\ell,\sigma)$ is always written in disjunctive normal form, i.e. $\delta(\ell,\sigma) = \underset{k \in K}{\bigvee} \bigwedge A_{k}$, where the terms $A_{k}$ might be $\ell, x.\ell, x \bowtie c, 0 \bowtie c, \top$ or $\bot$, for $\ell \in L$ and $c \in \N$.  Then, each minimal model M of $\delta(\ell,\sigma)$ wrt $I$ has the form $A_{k}[I]$, for some $k \in K$, where $A_{k}[I] = \lbrace (\ell,J) \; \vert \; \ell \in A_{k}\rbrace \cup \lbrace (\ell,\lbrace 0 \rbrace) \; \vert \; x.\ell \in A_{k}\rbrace$ and $I$ satisfies all the clocks constraints of $A_{k}$.\\

\begin{proposition}
\label{Prop2}
For all \TATA $\Aa$, $L^\omega(\Aa^C) = T\Sigma^{\omega} \setminus
L^\omega(\Aa)$.
\end{proposition}
\begin{proof}
~\\ \textbf{Claim.} Let $\theta$ be an infinite timed word.  Let $\pi$ be an 
$Id$-run of $\Aa$ on $\theta$ and $\pi'$ be a 
 $Id$-run of $\Aa^C$ on $\theta$. Then, $\pi$ and $\pi'$ have a common branch.\\
\textbf{Proof of the claim.} Let $\pi$ be a 
 $Id$-run of $\Aa$, denoted $C_{0} \overset{t_{1}}{\rightsquigarrow} C_{1} \overset{\sigma_{1}}{\longrightarrow} C_{2} \overset{t_{2}}{\rightsquigarrow} C_{3} \overset{\sigma_{2}}{\longrightarrow} \dots \overset{t_{n}}{\rightsquigarrow} C_{2n-1} \overset{\sigma_{n}}{\longrightarrow} C_{2n} \dots$ and $\pi'$ be a 
  $Id$-run of $\Aa^C$, denoted $ C'_{0} \overset{t_{1}}{\rightsquigarrow} C'_{1}
\overset{\sigma_{1}}{\longrightarrow} C'_{2} \overset{t_{2}}{\rightsquigarrow} C'_{3} \overset{\sigma_{2}}{\longrightarrow} \dots \overset{t_{n}}{\rightsquigarrow} C'_{2n-1} \overset{\sigma_{n}}{\longrightarrow} C'_{2n} \dots$.  We will recursively construct a common branch of $\pi$ and $\pi'$.  We will denote its elements by $e_0 e_1 e_2 ...$ (so, each $e_i \in C_i \cap C'_i$ and is a certain state $(\ell, v)$).\\
\underline{Basis:} $C_{0} = C'_{0} = \lbrace (\ell_0, 0) \rbrace$, so, we take $e_0 = (\ell_0, 0)$.\\
\underline{Induction:} Suppose we constructed a common beginning of branch of $\pi$ and $\pi'$: $e_0 e_1 e_2 ... e_{2n-3} e_{2n-2}$.  Let us show that we can extend this common beginning of branch, constructing $e_{2n-1}$ and $e_{2n}$.\\
We know that $e_{2n-2} \in C_{2n-2} \cap C'_{2n-2}$.  Suppose $e_{2n-2} = (\ell, v)$.  Then, transitions $C_{2n-2} \overset{t_{n}}{\rightsquigarrow} C_{2n-1}$ and $C'_{2n-2} \overset{t_{n}}{\rightsquigarrow} C'_{2n-1}$ both send $e_{2n-2} = (\ell, v)$ on $(\ell, v + t_n)$.  So we can construct $e_{2n-1} = (\ell, v + t_n) \in C_{2n-1} \cap C'_{2n-1}$.\\
Now, suppose that $C_{2n-1} = \lbrace (\ell_i, v_i)_{i \in I} \rbrace$ and let $i^\star \in I$ such that $e_{2n-1} = (\ell_{i^\star}, v_{i^\star})$. $C_{2n}$ is obtained thanks to $\delta(\ell_i, \sigma_{n})$ (for all $i \in I$), which can be written as (in disjunctive normal form) $\underset{j \in J}{\bigvee}$ $\underset{k \in K_j}{\bigwedge} A_{jk}$ where each $A_{jk}$ is a term of type $\ell$, $x.\ell$ or $x \bowtie c$.  For each $j \in J$, we can then define $A_{j}[v] = \lbrace (\ell,v) \; \vert \; \exists k \in K_j : \ell \in A_{jk} \rbrace \cup \lbrace (\ell,0) \; \vert \; \exists k \in K_j : x.\ell \in A_{jk} \rbrace$: each $A_j[v_{i}]$, for a $v_{i}$ satisfying the clock constraints present in $\underset{k \in K_j}{\bigcup} A_{jk}$, is a minimal model of $\delta(\ell_i, \sigma_{n})$ wrt $v_{i}$ (and they are the only ones !).  On one hand, there exists a $j^\star \in J$ such that $A_{j^\star}[v_{i^\star}] \subseteq C_{2n-1}$ and $v_{i}$ satisfies the clock constraints present in $\underset{k \in K_{j^\star}}{\bigcup} A_{j^\star k}$: for each element $(\ell,v)$ of $A_{j^\star}[v_{i^\star}]$, there is a transition in $\pi$ linking $e_{2n-1} = (\ell_{i^\star}, v_{i^\star})$ to $(\ell,v)$.  On the other hand, $\overline{\delta}(\ell_i, \sigma_{n}) = \underset{j \in J}{\bigwedge}$ $\underset{k \in K_j}{\bigvee} \overline{A_{jk}}$: for an $A_{jk}$ of type $\ell$ or $x.\ell$, $\overline{A_{jk}} = A_{jk}$, and for $A_{jk} = x \bowtie c$, $\overline{A_{jk}}$ is the negation of $x \bowtie c$.  As $v_{i^\star}$ satisfies the clock constraints present in $\underset{k \in K_{j^\star}}{\bigcup} A_{j^\star k}$, it does not satisfy the clock constraints present in $\underset{k \in K_{j^\star}}{\bigcup} \overline{A_{j^\star k}}$. Letting $C'_{2n-1} = \lbrace (\ell_{i'}, v_{i'})_{i' \in I'} \rbrace$, each minimal model of $\overline{\delta}(\ell_i, \sigma_{n})$ wrt $v_{i'}$ must be such that, for each $j \in J$:
\begin{itemize}
\item either $v_{i'}$ satisfies a certain clock constraint present in $\underset{k \in K_j}{\bigcup} A_{jk}$,
\item or this minimal model contains an element of $A_{j}[v_{i'}]$.
\end{itemize}
In particular, the minimal model of $\overline{\delta}(\ell_{i^\star}, \sigma_{n})$ wrt $v_{i^\star}$ that is used in $\pi'$ to construct $C_{2n}$ contains an element of $A_{j^\star}[v_{i^\star}]$ (because it does not satisfy any clock constraint in $\underset{k \in K_{j^\star}}{\bigcup} A_{j^\star k}$
) and there is a transition in $\pi'$ linking $e_{2n-1} = (\ell_{i^\star}, v_{i^\star})$ to this element.  It is this element we choose to be $e_{2n}$. \qed
~\\
\noindent We will prove the proposition thanks to this claim.  We first prove that $L^\omega(\Aa) \cap L^\omega(\Aa^C) = \emptyset$, and then that $L^\omega(\Aa) \cup L^\omega(\Aa^C) = T\Sigma^\omega$.\\
\underline{$\mathbf{L^\omega(\Aa) \cap L^\omega(\Aa^C) = \emptyset.}$} Suppose by contradiction that there exists a timed word $\theta$ such that $\theta \in L^\omega(\Aa)$ and $\theta \in L^\omega(\Aa^C)$.  Then, there exists 
 accepting $Id$-runs $\pi$ and $\pi'$ (resp.) of $\Aa$ and $\Aa^C$ on $\theta$.  By the previous claim, $\pi$ and $\pi'$ have a common branch $e_0 e_1 e_2 ...$.  As $\pi$ is an accepting $Id$-run of $\Aa$, $\exists n \in \N$, $\forall m > n$, the location of $e_{m}$ is in $F$. But as $\pi'$ is an accepting $Id$-run of $\Aa^C$, $\exists n' \in \N$, $\forall m > n$, the location of $e_{m}$ is in $L \setminus F$: this is impossible and so $\theta \notin L^\omega(\Aa)$ or $\theta \notin L^\omega(\Aa^C)$.\\
\underline{$\mathbf{L^\omega(\Aa) \cup L^\omega(\Aa^C) = T\Sigma^\omega.}$} Let $\theta$ be a timed word such that $\theta \notin L^\omega(\Aa)$, we must prove that $\theta \in L^\omega(\Aa^C)$.  We will inductively construct an accepting $Id$-run of $\Aa^C$ on $\theta$.  To actually construct an \textit{accepting} $Id$-run, we will maintain an additional property of the beginning of $Id$-run we extend at each inductive step, say $\pi_{2n}$ denoted $C_{0} \overset{t_{1}}{\rightsquigarrow} C_{1} \overset{\sigma_{1}}{\longrightarrow} C_{2} \overset{t_{2}}{\rightsquigarrow} C_{3} \overset{\sigma_{2}}{\longrightarrow} \dots \overset{t_{n}}{\rightsquigarrow} C_{2n-1} \overset{\sigma_{n}}{\longrightarrow} C_{2n}$: "for each branch $\beta$ of $\pi_{2n}$, there exists a beginning of $Id$-run of $\Aa$ on $\theta$ such that:
\begin{itemize}
\item either $\beta$ is a branch of $\pi'$ and $\pi'$ can not be prolonged into a complete $Id$-run of $\Aa$ (it is "blocking"), 
\item or $\pi'$ can be prolonged into a complete $Id$-run $\pi_c$ of $\Aa$ such that $\beta$ is the beginning of a branch of $\pi_c$ on which each location of $F$ only occurs a finite number of times."
\end{itemize}
We note $\mathbf{(^\star 2n)}$ this property.\\
\underline{Basis.} We construct $C_0 = \lbrace (\ell_0, 0) \rbrace$.  As there is no accepting $Id$-run of $\Aa$ on $\theta$, property $\mathbf{(^\star 0)}$ is trivially verified.\\
\underline{Induction.} Suppose we constructed a beginning of $Id$-run $\pi_{2n}$ of $\Aa^C$ on $\theta$: $C_{0} \overset{t_{1}}{\rightsquigarrow} C_{1} \overset{\sigma_{1}}{\longrightarrow} C_{2} \overset{t_{2}}{\rightsquigarrow} C_{3} \overset{\sigma_{2}}{\longrightarrow} \dots \overset{t_{n}}{\rightsquigarrow} C_{2n-1} \overset{\sigma_{n}}{\longrightarrow} C_{2n}$ such that property $\mathbf{(^\star 2n)}$ is verified.  We will extend $\pi_{2n}$ to obtain a beginning of $Id$-run $C_{0} \overset{t_{1}}{\rightsquigarrow} C_{1} \overset{\sigma_{1}}{\longrightarrow} C_{2} \overset{t_{2}}{\rightsquigarrow} C_{3} \overset{\sigma_{2}}{\longrightarrow} \dots \overset{t_{n}}{\rightsquigarrow} C_{2n-1} \overset{\sigma_{n}}{\longrightarrow} C_{2n} \overset{t_{n+1}}{\rightsquigarrow} C_{2n+1} \overset{\sigma_{n+1}}{\longrightarrow} C_{2n+2}$ such that property $\mathbf{(^\star 2n+2)}$ is verified.\\
First, we construct $C_{2n+1} = C_{2n} + t_{n+1}$, so that $\pi_{2n+1}$ is $C_{0} \overset{t_{1}}{\rightsquigarrow} C_{1} \overset{\sigma_{1}}{\longrightarrow} C_{2} \overset{t_{2}}{\rightsquigarrow} C_{3} \overset{\sigma_{2}}{\longrightarrow} \dots \overset{t_{n}}{\rightsquigarrow} C_{2n-1} \overset{\sigma_{n}}{\longrightarrow} C_{2n} \overset{t_{n+1}}{\rightsquigarrow} C_{2n+1}$.  $\mathbf{(^\star 2n+1)}$ is trivially verified thanks to $\mathbf{(^\star 2n)}$.\\
Now, suppose that $C_{2n+1} = \lbrace (\ell_i,v_i)_{i \in I} \rbrace$.  We know that ( $\mathbf{(^\star 2n+1)}$ ) for each branch $\beta$ of $\pi_{2n}$, there exists a beginning of $Id$-run $\pi'$ of $\Aa$ on $\theta$ such that:
\begin{itemize}
\item either $\beta$ is a branch of $\pi'$ and $\pi'$ can not be prolonged into a complete $Id$-run of $\Aa$, 
\item or $\pi'$ can be prolonged into a complete $Id$-run $\pi_c$ of $\Aa$ such that $\beta$ is the beginning of a branch of $\pi_c$ on which each location of $F$ only occurs a finite number of times.
\end{itemize}
In particular, for each $(\ell_i,v_i) \in C_{2n+1}$, there exists a beginning of $Id$-run $\pi_{i}$ of $\Aa$ on $\theta$ such that $(\ell_i,v_i)$ is the (2n+1)th element of a branch of $\pi_{i}$. $\pi_{i}$ must then evolve towards a minimal model of $\delta(\ell_i,\sigma_{n+1})$ wrt $v_i$. Let us note $\delta(\ell_i,\sigma_{n+1})$ as $\underset{j \in J}{\bigvee}$ $\underset{k \in K_j}{\bigwedge} A_{jk}$, as previously.  As there is no accepting $Id$-run of $\Aa$ on $\theta$, whatever is the minimal model $A_j[v_i]$ we decide to take as successor of $(\ell_i,v_i)$ in a beginning of $Id$-run of $\Aa$ on $\theta$, and no matter how it is then prolonged, it will contain a branch on which each location of $F$ only occurs a finite number of times or will be "blocking".   Let us note $(\ell_{i}^{j}, v_{i}^{j})_{i \in I, j \in J}$ the successors of $(\ell_i,v_i)$ on the branches on which each location of $F$ only occurs a finite number of timed or are "blocking"\footnote{Remark that a branch of a $Id$-run of $\Aa$ is blocking reading $\sigma_{n+1}$ iff $\underset{k \in K_j}{\bigcup} A_{jk}$ contains a clock constraint $x \bowtie c$ not satisfied in $v_i$ or there is no transition from $\ell_i$ reading $\sigma_{n+1}$.  In the first case $\overline{A_{jk}}$ contains the negation of $x \bowtie c$, which is verified by $v_i$, and in the second case $\delta(\ell_i,\sigma_{n+1})$ = false and so $\overline{\delta}(\ell_i,\sigma_{n+1})$ = true. In both cases (in particular) $\underset{k \in K_j}{\bigvee} \overline{A_{jk}}$ will be satisfied in $\overline{\delta}(\ell_i,\sigma_{n+1})$ even if no successor is attributed to $(\ell_i,v_i)$ (the branch ending in $(\ell_i,v_i)$ in $pi_{2n}$ is finite but not blocking).}, considering a certain $Id$-run of $\Aa$ on $\theta$ going from $(\ell_i,v_i)$ to the minimal model $A_j[v_i]$. We construct $C_{2n+2} = \lbrace (\ell_{i}^{j},v_{i}^{j})_{i \in I, j \in J} \rbrace$ (it is actually the union, for $i \in I$, of minimal models of $\overline{\delta}(\ell_i,\sigma_{n+1}) = \underset{j \in J}{\bigwedge}$ $\underset{k \in K_j}{\bigvee} \overline{A_{jk}}$ wrt $v_i$).  By construction, $\mathbf{(^\star 2n+2)}$ is verified.\\
Thanks to this induction, we can construct an infinite $Id$-run $\pi$, $C_{0} \overset{t_{1}}{\rightsquigarrow} C_{1} \overset{\sigma_{1}}{\longrightarrow} C_{2} \overset{t_{2}}{\rightsquigarrow} C_{3} \overset{\sigma_{2}}{\longrightarrow} \dots \overset{t_{n}}{\rightsquigarrow} C_{2n-1} \overset{\sigma_{n}}{\longrightarrow} C_{2n} \dots$, such that for each branch $\beta$ of $\pi$, there exists a $Id$-run or a beginning of $Id$-run $\pi'$ of $\Aa$ on $\theta$ such that:
\begin{itemize}
\item either $\beta$ is a branch of $\pi'$ and $\pi'$ can not be prolonged into a complete $Id$-run of $\Aa$ (it is "blocking"): it means that $\beta$ is a finite and non-blocking branch of $\pi$, 
\item or $\pi'$ is a $Id$-run of $\Aa$ such that $\beta$ is a branch of this $Id$-run on which each location of $F$ only occurs a finite number of times.
\end{itemize}
So, all the infinite branches of $\pi$ visit $L \setminus F$ infinitely often and $\pi$ is an accepting $Id$-run of $\Aa^C$ on $\theta$.\qed
\end{proof}

\paragraph{\TATA and MITL} Equipped with those results we can now
expose the two main results of this section. First, the translation
from MTL to \ATA introduced in \cite{OW07} carries on to infinite
words (to the best of our knowledge this had not been proved before
and does not seem completely trivial since our proof requires the
machinery of \TATA developed in this paper). Second, for every MITL
formula $\varphi$, we can devise an
$M(\varphi)$-bounded\footnote{$M(\varphi) \; \leq \; \vert \varphi \vert \times
  \underset{I \in \mathcal{I}_\varphi}{\max} \left( 4 \times
  \left\lceil \frac{\inf(I)}{\vert I \vert} \right\rceil +2, 2 \times
  \left\lceil \frac{\sup(I)}{\vert I \vert} \right\rceil +2 \right)$,
  where $\mathcal{I}_\varphi$ is the set of all the intervals that occur
  in $\varphi$.} approximation function $\appfunc{\varphi}$ to bound
the number of intervals needed along all runs of the intervals
semantics of the \TATA $\Aa_\varphi$, while retaining the semantics of
$\varphi$. Notice that this second property fails when applied to
formulae $\varphi$ of MTL. \todo{T: ai ajouté derniere phrase. OK?}

\begin{theorem}
\label{equalSem}
For every MTL formula $\varphi$: $L^{\omega}(\Aa_\varphi)=\sem{\varphi}$.
\end{theorem}
\begin{proof}
This has been proved in the finite words case in
  \cite[Prop.~6.4]{OW07}. This proof relies crucially on the
  fact that \ATA can be complemented in this case. Thanks to
  Proposition~\ref{Prop2}, we can adapt the proof of~\cite{OW07}.\qed
\end{proof}

We will now show there exists a family of bounded approximation functions $\appfunc{\varphi}$ s.t. for every MITL formula $\varphi$, $L^{\omega}_{\appfunc{\varphi}}(\mathcal{A}_{\varphi}) = L^{\omega}(\mathcal{A}_{\varphi})$.  We first recall from \cite{BEG13} the exact value of the upper bound $M(\varphi)$ on the number of clock copies (intervals) we need to consider in the configurations to recognise an MITL formula $\varphi$. Then, we recall the definition, for an MITL formula $\varphi$, of the approximation function $\appfunc{\varphi}$.\\
Let $\varphi$ be an MITL formula
in negative normal form.  We define $M(\varphi)$, thanks to
$M^{\infty}(\varphi)$ and $M^{1}(\varphi)$ defined as follows
\begin{itemize}
\item if $\varphi = \sigma$ or $\varphi=\neg \sigma$ (with $\sigma\in
  \Sigma$), then $M(\varphi) = 2$ and $M^{\infty}(\varphi) = M^{1}(\varphi) =
  0$.
\item if $\varphi = \varphi_{1} \wedge \varphi_{2}$, then $M(\varphi) = \max \{2,
  M^{1}(\varphi_{1}) + M^{1}(\varphi_{2})\}$, $M^{\infty}(\varphi) =
  M^{\infty}(\varphi_{1}) + M^{\infty}(\varphi_{2})$ and $M^{1}(\varphi)
  = M^{1}(\varphi_{1}) + M^{1}(\varphi_{2})$.
\item if $\varphi = \varphi_{1} \vee \varphi_{2}$, then $M(\varphi) = \max
  \{2, M^{1}(\varphi_{1}), M^{1}(\varphi_{2})\}$, $M^{\infty}(\varphi) =\break
  \max \{ M^{\infty}(\varphi_{1}), M^{\infty}(\varphi_{2})\}$ and
  $M^{1}(\varphi) = \max \{ M^{1}(\varphi_{1}), M^{1}(\varphi_{2})\}$.
\item if $\varphi = \varphi_{1}U_{I} \varphi_{2}$, then $M(\varphi) = \max
  \{ 2, M^{\infty}(\varphi_{1}) +M^{1}(\varphi_{2}) +1 \}$,
  $M^{\infty}(\varphi) = \left(4\times\left\lceil \frac{inf(I)}{\vert I
        \vert} \right\rceil +2\right) +M^{\infty}(\varphi_{1})
  +M^{\infty}(\varphi_{2})$ and $M^{1}(\varphi) = M^{\infty}(\varphi_{1})
  +M^{1}(\varphi_{2}) +1$.
\item if $\varphi = \varphi_{1} \tilde{U}_{I} \varphi_{2}$, then $M(\varphi)
  = \max \{2, M_{1}(\varphi_{1}) +M_{\infty}(\varphi_{2}) +1\}$,
  $M_{\infty}(\varphi) = \left(2\times \left\lceil \frac{sup(I)}{\vert I
        \vert} \right\rceil +2\right) +M_{\infty}(\varphi_{1})
  +M_{\infty}(\varphi_{2})$ and $M_{1}(\varphi) = M_{1}(\varphi_{1})
  +M_{\infty}(\varphi_{2}) +1$.
\end{itemize}

Throughout this description, we assume an \ATA $\Aa$ with
set of locations $L$. Let
$S=\{(\ell,I_0),(\ell,I_1),\ldots,(\ell,I_m)\}$ be a set of
states of $\Aa$, all in the same location $\ell$, with $I_0
< I_1 < \cdots < I_m$. Then, we let $\mer{S}=\{(\ell, [0,sup(I_1)]),
(\ell, I_2),\ldots, (\ell,I_m)\}$ if $I_0=[0,0]$ and $\mer{S}=S$
otherwise, i.e., $\mer{S}$ is obtained from $S$ by \emph{grouping}
$I_0$ and $I_1$ iff $I_0=[0,0]$, otherwise $\mer{S}$ does not modify
$S$. Observe that, in the former case, if $I_{1}$ is not a singleton,
then $\nbclocks{\mer{S}}=\nbclocks{S}-1$. Now, we can lift the
definition of \mername to configurations. Let $C$ be a configuration
of $\Aa$ and let $k\in \N$. We let:
\begin{align*}
  \mer{C,k} &= \big\{C'\mid \nbclocks{C'}\leq k\text{ and } \forall
  \ell\in L: C'(\ell)\in\{\mer{C(\ell)}, C(\ell)\}\big\}
\end{align*}
Observe that $\mer{C,k}$ is a (possibly empty) \emph{set of
  configurations}, where each configuration $(i)$ has at most $k$
clock copies, and $(ii)$ can be obtained by applying (if possible) or
not the $\mername$ function to each $C(\ell)$. Let us now define a
family of $k$-bounded approximation functions, based on
$\mername$. Let $k\geq 2\times |L|$ be a bound and let $C$ be a
configuration, assuming that $C(\ell)=\{I^\ell_1,\ldots,
I^\ell_{m_\ell}\}$ for all $\ell\in L$. Then:
\begin{align*}
  F^k(C) &=
  \left\{
    \begin{array}{ll}
     \mer{C,k} &\text{If } \mer{C,k}\neq\emptyset\\
     \big\{(\ell, [inf(I^\ell_1),sup(I^\ell_{m_\ell})])\mid \ell\in L\big\} &\text{otherwise}
    \end{array}
  \right.
\end{align*}
Roughly speaking, the $F^k(C)$ function tries to obtain configurations
$C'$ that approximate $C$ and s.t. $\nbclocks{C'}\leq k$, using the
$\mername$ function. If it fails to, i.e., when $\mer{C,k}=\emptyset$,
$F^k(C)$ returns a single configuration, obtained from $C$ be grouping
all the intervals in each location. The latter case occurs in the
definition of $F^k$ for the sake of completeness. When the \ATA $\Aa$
has been obtained from an MITL formula $\varphi$, and for $k$ big enough
(see hereunder) each $\theta\in\sem{\varphi}$ will be recognised by at
least one $F^k$-run of $\Aa$ that traverses only configurations
obtained thanks to $\mername$.  We can now finally define
$\appfunc{\varphi}$ for every MITL formula $\varphi$, by letting
$\appfunc{\varphi}=F^K$, where $K=\max\{2\times |L|,M(\varphi)\}$.\\
We insist on the importance of Proposition \ref{simpleAC} and Theorem \ref{equalSem} that enable to adapt with a large facility the proof of Theorem 13 from \cite{BEG13} to infinite words, providing the following theorem:

\begin{theorem}
  \label{ThmBorne}
  For every MITL formula $\varphi$, there exists an $M(\varphi)$-bounded
    approximation function $\appfunc{\varphi}$ such that $L^{\omega}_{\appfunc{\varphi}}(\mathcal{A}_{\varphi}) =
  L^{\omega}(\mathcal{A}_{\varphi})=\sem{\varphi}$.
\end{theorem}

\begin{figure}[t]
\centering
\begin{tikzpicture}[scale=.9]

\begin{scope}
\draw (0,0) node [circle, double, draw, inner sep=3pt] (A) {$\ell_{\Box}$};
\draw (3,0) node [circle, draw, inner sep=3pt] (B) {$\ell_\lozenge$};

\draw [-latex'] (-1,0) -- (A) ;
\draw [-latex'] (A) .. controls +(230:1.2cm) and +(310:1.2cm) .. (A) node [pos=0.5,below] {$b$};
\draw [-latex'] (B) .. controls +(230:1.2cm) and +(310:1.2cm) .. (B) node [pos=0.5,below] {$a,b$};
\draw [-latex'] (B) -- (5.2,0) node [pos=.5,above] {$\begin{array}{c}b\\ x \in [1,2]\end{array}$};
\draw[-*] (A) -- (1.5,0) node [pos=.3,below] {$a$};
\draw[-latex'] (1.5,0) -- (B) node [pos=.5,below] {$x := 0$};
\draw [-latex'] (1.4,0) .. controls +(-220:0.8cm) .. (A);
\end{scope}

\begin{scope}[shift={(5.7,0)},scale=.9]
\fill (0,0) circle (0.05cm);
\draw (0,-0.3) node {\scriptsize{0}};
\fill (2,0) circle (0.05cm);
\draw (2,-0.3) node {\scriptsize{1}};
\draw (4,-0.3) node {\scriptsize{2}};
\draw (6,-0.3) node {\scriptsize{3}};

\draw[-latex'] (0,0) -- (7,0);
\draw (7,-0.3) node {\tiny{time}};

\fill (0.6,0) circle (0.03cm);
\fill (0.8,0) circle (0.03cm);
\fill (1,0) circle (0.03cm);
\fill (1.2,0) circle (0.03cm);
\fill (1.4,0) circle (0.03cm);
\fill (1.6,0) circle (0.03cm);
\fill (1.8,0) circle (0.03cm);

\fill (2.2,0) circle (0.03cm);
\fill (2.4,0) circle (0.03cm);
\fill (2.6,0) circle (0.03cm);
\fill (2.8,0) circle (0.03cm);
\fill (3,0) circle (0.03cm);
\fill (3.2,0) circle (0.03cm);
\fill (3.4,0) circle (0.03cm);
\fill (3.6,0) circle (0.03cm);

\fill (4.2,0) circle (0.03cm);
\fill (4.4,0) circle (0.03cm);
\fill (4.6,0) circle (0.03cm);
\fill (4.8,0) circle (0.03cm);
\fill (5,0) circle (0.03cm);
\fill (5.2,0) circle (0.03cm);
\fill (5.4,0) circle (0.03cm);
\fill (5.6,0) circle (0.03cm);
\fill (5.8,0) circle (0.03cm);

\fill (6.2,0) circle (0.03cm);
\fill (6.4,0) circle (0.03cm);
\fill (6.6,0) circle (0.03cm);
\fill (6.8,0) circle (0.03cm);

\fill (0.2,0) circle (0.05cm);
\draw (0.2,0.3) node {$a$};
\fill (0.4,0) circle (0.05cm);
\draw (0.4,0.3) node {$a$};
\fill (3.8,0) circle (0.05cm);
\draw (3.8,0.3) node {$a$};
\fill (4,0) circle (0.05cm);
\draw (4,0.3) node {$b$};
\fill (6,0) circle (0.05cm);
\draw (6,0.3) node {$b$};

\draw[-] (0.05,0.5) -- (0.53,0.5);
\draw[-] (0.05,0.3) -- (0.05,0.5);
\draw[-] (0.53,0.3) -- (0.53,0.5);
\draw[-latex'] (0.3,0.5) .. controls +(80:0.9cm) and +(130:0.9cm) .. (4,0.5);

\draw[-] (3.65,0.5) -- (3.91,0.5);
\draw[-] (3.65,0.3) -- (3.65,0.5);
\draw[-] (3.91,0.3) -- (3.91,0.5);
\draw[-latex'] (3.8,0.5) .. controls +(80:0.9cm) and +(130:0.9cm) .. (6,0.5);
\end{scope}
\end{tikzpicture}
\caption{(left) \ATA $\mathcal{A}_{\varphi}$ with $\varphi=\qed ( a
  \Rightarrow \lozenge_{[1,2]} b)$. \ (right) The grouping of clocks. \label{ExGilles}\label{Graph1}}
\end{figure}

\begin{example}
 Let us illustrate the idea behind the approximation function
  $\appfunc{\varphi}$ by considering again the run prefixes on $\theta
  = (a,0.1) (a,0.2) (a,1.9) (b,2) (b,3) (b,4)\ldots$ in
  Fig.~\ref{reprun}. The two first positions (with
  $\sigma_1=\sigma_2=a$) of $\theta$ satisfy $\lozenge_{[1,2]}b$,
  \emph{thanks to the $b$ in position $4$} (with $\tau_4=2$), while
  position 3 (with $\sigma_3=a$) satisfies $\lozenge_{[1,2]}b$
  \emph{thanks to the $b$ in position $5$} (with $\tau_5=3$), see
  Fig.~\ref{Graph1} (right). Hence, $\appfunc{\varphi}$ groups the two
  clock copies created in $\ell_\lozenge$ when reading the two first
  $a$'s, but keeps the third one apart. This yields the
  $\appfunc{\varphi}$-run $\pi''$ in Fig.~\ref{reprun}. On the other
  hand, the strategy of grouping all the clock copies present in each
  location, which yields $\pi'$, is not a good solution. This prefix
  cannot be extended to an accepting run because of the copy in state
  $(\ell_\lozenge, [2.1, 2.9])$ in the rightmost configuration, that
  will never be able to visit an accepting location.

\begin{figure}[t]
\centering
\begin{tikzpicture}[xscale=.82,yscale=.9]

\everymath{\scriptstyle}

\draw (0,0) node {$\pi$};
\draw (1,0.7) node (Z) {$C_0$};
\draw (1, 0) node [rectangle,fill=black!10!white,inner sep=1pt,rounded corners=1mm] (A) {$\begin{array}{l}(\ell_\Box,0)\\ \end{array}$};
\draw[white] (1,0.5) -- (1,1) node [pos=0.35, sloped, left, black] {$=$};

\draw (3,1.4) node (Y) {$C_2$};
\draw[-latex'] (Z) -- (2.6,1.4) node [pos=.4,sloped, above] {$0.1, a$} node [pos=1, below] {$_{Id}$};
\draw (3,0.9) node [rectangle,fill=black!10!white,inner sep=1pt,rounded corners=2mm,below] {$\begin{array}{l} (\ell_\Box,0.1)\\ \\ \\ (\ell_\lozenge,0) \end{array}$};
\draw[-latex'] (A) -- (2.3,0.6);
\draw[-latex'] (A) -- (2.3,-0.6);
\draw[white] (3,1.2) -- (3,1.7) node [pos=0.3, sloped, left, black] {$=$};

\draw (5.5,1.4) node (X) {$C_4$};
\draw[-latex'] (Y) -- (X) node [pos=.5,above] {$0.1, a$} node [pos=1, below] {$_{Id}$};
\draw (5.5,0.9) node [rectangle,fill=black!10!white,inner sep=1pt,rounded corners=2mm,below] {$\begin{array}{l} (\ell_\Box,0.2)\\ \\ (\ell_\lozenge,0)\\ (\ell_\lozenge,0.1) \end{array}$};
\draw[-latex'] (3.7,0.6) -- (4.8,0.6);
\draw[-latex'] (3.7,0.5) -- (4.8,-0.2);
\draw[-latex'] (3.7,-0.6) -- (4.8,-0.6);
\draw[white] (5.5,1.2) -- (5.5,1.7) node [pos=0.3, sloped, left, black] {$=$};

\draw (8.1,1.4) node (W) {$C_6$};
\draw[-latex'] (X) -- (W) node [pos=.5,above] {$1.7, a$} node [pos=1, below] {$_{Id}$};
\draw (8.1,0.9) node [rectangle,fill=black!10!white,inner sep=1pt,rounded corners=2mm,below] {$\begin{array}{l} (\ell_\Box,1.9)\\ (\ell_\lozenge,0)\\ (\ell_\lozenge,1.7)\\ (\ell_\lozenge,1.8) \end{array}$};
\draw[-latex'] (6.2,0.6) -- (7.4,0.6);
\draw[-latex'] (6.2,0.5) -- (7.4,0.2);
\draw[-latex'] (6.2,-0.2) -- (7.4,-0.2);
\draw[-latex'] (6.2,-0.6) -- (7.4,-0.6);
\draw[white] (8.1,1.2) -- (8.1,1.7) node [pos=0.3, sloped, left, black] {$=$};

\draw (10.7,1.4) node (V) {$C_8$};
\draw[-latex'] (W) -- (V) node [pos=.5,above] {$0.1, b$} node [pos=1, below] {$_{Id}$};
\draw (10.7,0.9) node [rectangle,fill=black!10!white,inner sep=1pt,rounded corners=2mm,below] {$\begin{array}{l} (\ell_\Box,2)\\ (\ell_\lozenge,0.1)\\ \end{array}$};
\draw[-latex'] (8.8,0.6) -- (10,0.6);
\draw[-latex'] (8.8,0.2) -- (10,0.2);
\draw[white] (10.7,1.2) -- (10.7,1.7) node [pos=0.3, sloped, left, black] {$=$};

\draw (13.3,1.4) node (U) {$C_{10}$};
\draw[-latex'] (V) -- (U) node [pos=.5,above] {$1, b$} node [pos=1, below] {$_{Id}$};
\draw (13.3,0.9) node [rectangle,fill=black!10!white,inner sep=1pt,rounded corners=1mm,below] {$\begin{array}{l} (\ell_\Box,3)\\ \end{array}$};
\draw[-latex'] (11.4,0.6) -- (12.7,0.6);
\draw[white] (13.3,1.2) -- (13.3,1.7) node [pos=0.3, sloped, left, black] {$=$};

\begin{scope}[yshift=-1.6cm]

\draw (0,0) node {$\pi'$};
\draw (1, 0) node [rectangle,fill=black!10!white,inner sep=1pt,rounded corners=1mm] (A) {$\begin{array}{l}(\ell_\Box,0)\\ \end{array}$};

\draw (3,0.47) node [rectangle,fill=black!10!white,inner sep=1pt,rounded corners=2mm,below] {$\begin{array}{l} (\ell_\Box,0.1)\\ (\ell_\lozenge,0) \end{array}$};
\draw[-latex'] (A) -- (2.3,0.22);
\draw[-latex'] (A) -- (2.3,-0.22);

\draw (5.5,0.47) node [rectangle,fill=black!10!white,inner sep=1pt,rounded corners=2mm,below] {$\begin{array}{l} (\ell_\Box,0.2)\\ (\ell_\lozenge,[0,0.1]) \end{array}$};
\draw[-latex'] (3.7,0.22) -- (4.6,0.22);
\draw[-latex'] (3.7,0.18) -- (4.6,-0.22);
\draw[-latex'] (3.7,-0.22) -- (4.6,-0.22);

\draw (8.1,0.47) node [rectangle,fill=black!10!white,inner sep=1pt,rounded corners=2mm,below] {$\begin{array}{l} (\ell_\Box,1.9)\\ (\ell_\lozenge,[0,1.8]) \end{array}$};
\draw[-latex'] (6.4,0.22) -- (7.2,0.22);
\draw[-latex'] (6.4,0.18) -- (7.2,-0.22);
\draw[-latex'] (6.4,-0.22) -- (7.2,-0.22);

\draw (10.7,0.47) node [rectangle,fill=black!10!white,inner sep=1pt,rounded corners=2mm,below] {$\begin{array}{l} (\ell_\Box,2)\\ (\ell_\lozenge,[0.1,1.9]) \end{array}$};
\draw[-latex'] (9,0.22) -- (9.65,0.22);
\draw[-latex'] (9,-0.22) -- (9.65,-0.22);

\draw (13.3,0.47) node [rectangle,fill=black!10!white,inner sep=1pt,rounded corners=2mm,below] {$\begin{array}{l} (\ell_\Box,3)\\ (\ell_\lozenge,[2.1,2.9]) \end{array}$};
\draw[-latex'] (11.75,0.22) -- (12.25,0.22);
\draw[-latex'] (11.75,-0.22) -- (12.25,-0.22);
\end{scope}

\begin{scope}[yshift=-3cm]
\draw (0,0) node {$\pi''$};
\draw (1, 0) node [rectangle,fill=black!10!white,inner sep=1pt,rounded corners=1mm] (A) {$\begin{array}{l}(\ell_\Box,0)\\ \end{array}$};

\draw (3,0.7) node [rectangle,fill=black!10!white,inner sep=1pt,rounded corners=2mm,below] {$\begin{array}{l} (\ell_\Box,0.1)\\ \\ (\ell_\lozenge,0) \end{array}$};
\draw[-latex'] (A) -- (2.3,0.4);
\draw[-latex'] (A) -- (2.3,-0.4);

\draw (5.5,0.7) node [rectangle,fill=black!10!white,inner sep=1pt,rounded corners=2mm,below] {$\begin{array}{l} (\ell_\Box,0.2)\\ \\ (\ell_\lozenge,[0,0.1]) \end{array}$};
\draw[-latex'] (3.7,0.4) -- (4.6,0.4);
\draw[-latex'] (3.7,0.35) -- (4.6,-0.4);
\draw[-latex'] (3.7,-0.4) -- (4.6,-0.4);

\draw (8.1,0.7) node [rectangle,fill=black!10!white,inner sep=1pt,rounded corners=2mm,below] {$\begin{array}{l} (\ell_\Box,1.9)\\ (\ell_\lozenge,0)\\ (\ell_\lozenge,[1.7,1.8]) \end{array}$};
\draw[-latex'] (6.4,0.4) -- (7.1,0.4);
\draw[-latex'] (6.4,0.35) -- (7.1,0);
\draw[-latex'] (6.4,-0.4) -- (7.1,-0.4);

\draw (10.7,0.7) node [rectangle,fill=black!10!white,inner sep=1pt,rounded corners=2mm,below] {$\begin{array}{l} (\ell_\Box,2)\\ (\ell_\lozenge,0.1)\\ \end{array}$};
\draw[-latex'] (9.15,0.4) -- (10,0.4);
\draw[-latex'] (9.15,0) -- (10,0);

\draw (13.3,0.7) node [rectangle,fill=black!10!white,inner sep=1pt,rounded corners=1mm,below] {$\begin{array}{l} (\ell_\Box,3)\\ \end{array}$};
\draw[-latex'] (11.4,0.4) -- (12.7,0.4);
\end{scope}

\end{tikzpicture}
\caption{Several \ATA runs.}
\label{reprun}
\end{figure}
\end{example}

\paragraph{A Büchi transition system for each \ATA.} Thanks to the bound $M(\varphi)$ on the number of necessary clock copies and the approximation function $\appfunc{\varphi}$, given by Theorem \ref{ThmBorne}, we can use the method of Miyano Hayashi \cite{MH84} to construct a Büchi timed transition system from $A_\varphi$: $\mhts{\Aa_\varphi,\appfunc{\varphi}}$. Each state of a configuration of $\Aa_\varphi$ will be marked by $\top$ or $\bot$. Intuitively, a state is
marked by $\top$ iff all the branches it belongs to have visited a
final location of $\Aa_\varphi$ since the last accepting state of
$\mhts{\Aa_\varphi,\appfunc{\varphi}}$ (i.e., a state where all markers are $\top$).
Formally, for an $\ATA$ $\Aa$ and an approximation function $f$, we define the transition system $\mhts{\Aa,f} = ( \Sigma, S^{\sf MH}, s^{\sf MH}_{0}, \rightsquigarrow^{\sf MH}, \rightarrow^{\sf MH}, \alpha )$ where: $S^{\sf MH} = \lbrace (\ell_k, I_k, m_k)_{k \in K} \; \vert \; \lbrace(\ell_k, I_k)_{k \in K}\rbrace \text{ is a configuration of } \Aa \text{ and } \forall k \in K, m_k \in \lbrace \top,\bot \rbrace$, $s^{\sf MH}_{0} = \lbrace (\ell_0, [0,0], \bot) \rbrace$ (because $\ell_0 \notin F$), $\alpha = \lbrace (\ell_k, I_k, m_k)_{k \in K} \in S^{\sf MH} \; \vert \; \forall k \in K, m_k = \top \rbrace$. For $t \in \R$ and $s, s' \in S$, supposing $s = \lbrace (\ell_k, I_k, m_k)_{k \in K} \rbrace$, we have $s \overset{t}{\rightsquigarrow}^{\sf MH} s'$ iff  $s' = \lbrace (\ell_k, I_k + t, m_k)_{k \in K} \rbrace$. $\rightsquigarrow^{\sf MH} = \underset{t \in \R}{\bigcup} \overset{t}{\rightsquigarrow}^{\sf MH}$.\\
"$\rightarrow^{\sf MH}$", representing the discrete transitions, is defined in way that (1) a transition with $\rightarrow^{\sf MH}$ between two states of $S^{\sf MH}$ corresponds to a discrete transition between the configurations of $\Aa$ they represent and (2) the markers of the third component of states of $S^{\sf MH}$ are kept updated.  Actually, if $s \overset{\sigma}{\rightarrow}^{\sf MH} s'$, we want the last component of a trio of $s'$ to be $\bot$ iff its first component is not in $F$ and, either it comes from the grouping of trios emanating from trios such that at least one of them had $\bot$ as last component.  Moreover, if $s \in \alpha$, we want to start again the marking: we put all the third components of the trios of $s$ to $\bot$ before proceeding as previously. 
Formally:
\begin{enumerate}
\item[(i)] For $s \in S^{\sf MH} \setminus \alpha$ and $s' \in S^{\sf MH}$, supposing $s = \lbrace (\ell_k, I_k, m_k)_{k \in K} \rbrace$ and $s' = \lbrace (\ell_{k'}, I_{k'}, m_{k'})_{k' \in K'} \rbrace$, $s \overset{\sigma}{\rightarrow}^{\sf MH} s'$ iff
\begin{enumerate}
\item $\lbrace (\ell_k, I_k)_{k \in K} \rbrace \overset{\sigma}{\longrightarrow}_{f} \lbrace (\ell_{k'}, I_{k'})_{k' \in K'} \rbrace$, i.e 
$\break \lbrace (\ell_{k'}, I_{k'})_{k' \in K'} \rbrace \in f( \succ(\lbrace (\ell_k, I_k)_{k \in K} \rbrace, \sigma))$;
\item $\forall k' \in K'$: $\left( \ell_{k'} \in F \; \Rightarrow \; m_{k'} = \top \right)$ ;
\item $\forall k' \in K'$ with $\ell_{k'} \notin F$: if $\exists k \in K$ s.t. $(\ell_k,I_k,\bot)\in s$ and  $(\ell_{k'},I_{k'}) \in\dest(\conf{s},\conf{s'},(\ell_k,I_k,\bot))$, we have $m_{k'}=\bot$ ; otherwise, $m_{k'} = \top$.
\end{enumerate}
\item[(ii)] For $s \in \alpha$ and $s' \in S^{\sf MH}$, supposing $s = \lbrace (\ell_k, I_k, \top)_{k \in K} \rbrace$, $s \overset{\sigma}{\rightarrow}^{\sf MH} s'$ iff $\lbrace (\ell_k, I_k, \bot)_{k \in K} \rbrace \overset{\sigma}{\rightarrow}^{\sf MH} s'$ according to the rules in (i).
\end{enumerate}

\begin{proposition}
\label{Prop3}
Let $\Aa$ be an \ATA and $f$ be an approximation function: $L^\omega (\mhts{\Aa,f}) = L^\omega_{f}(\Aa)$.
\end{proposition}
\begin{proof}
$(\supseteq)$ Let $\theta = (\bar{\sigma}, \bar{\tau}) \in L^\omega_{f}(\Aa)$, with $\bar{\sigma}=\sigma_1\sigma_2\cdots\sigma_n \dots$ and
$\bar{\tau}=\tau_1\tau_2\cdots\tau_n \dots$.  We will prove that $\theta \in L^\omega (\mhts{\Aa,f})$.
Let us note $t_{i} = \tau_{i} - \tau_{i-1}$ for all $1 \leq i \leq \vert \theta \vert$, assuming $\tau_{0} = 0$.  We have an accepting $f$-run of $\Aa$ on $\theta$, say $\pi: C_{0} \overset{t_1}{\rightsquigarrow} C_{1} \overset{\sigma_{1}}{\longrightarrow}_{f} C_{2} \overset{t_2}{\rightsquigarrow} C_3 \overset{\sigma_{2}}{\longrightarrow}_{f} \dots \overset{t_i}{\rightsquigarrow} C_{2i-1} \overset{\sigma_{i}}{\longrightarrow}_{f} C_{2i} \dots$.  We must prove that there is an accepting run of $\mhts{\Aa,f}$ on $\theta$, say $\pi'$: $s_{0} \overset{t_{1}}{\rightsquigarrow}^{\sf MH} s_{1} \overset{\sigma_{1}}{\rightarrow}^{\sf MH} s_{2} \overset{t_{2}}{\rightsquigarrow}^{\sf MH} s_3 \overset{\sigma_{2}}{\rightarrow}^{\sf MH} \dots \overset{t_{i}}{\rightsquigarrow}^{\sf MH} s_{2i-1} \overset{\sigma_{i}}{\rightarrow}^{\sf MH} s_{i} \dots$.  We construct $\pi'$ by induction, proving additionally that the two following properties hold for $j \geq 0$:
\begin{enumerate}
\item[$(\star 2j)$] $C_{2j} = \lbrace (\ell_k, I_k)_{k \in K} \rbrace \;$ iff $\; s_{2j} = \lbrace (\ell_k, I_k, m_k)_{k \in K} \; \vert \; m_k \in \lbrace \top, \bot \rbrace \rbrace$ ;
\item[$(\ast 2j)$] if a location of $F$ occurs on all the branches of $\pi$ between the last configuration $C_{2j'}$ such that $s_{2j'} \in \alpha$ (or, failing that, between $C_{0}$) and $C_{2j}$, then, $s_{2j} \in \alpha$.
\end{enumerate}
\underline{Basis:} We know that $C_{0} = (\ell_{0}, [0,0])$ and $s_0 = \lbrace (\ell_0, [0,0], \bot) \rbrace$.  It is clear that $(\star 0)$ and $(\ast 0)$ are verified because only $\ell_0 \notin F$ occurs on the (unique) branch of $\pi$.\\
\underline{Induction:} Suppose that we constructed $\pi'$ until $s_{2i}$ and that $(\star 2j)$ and $(\ast 2j)$ are verified $\forall 0 \leq j \leq i$.  We will construct $\pi'$ until $s_{2(i+1)}$ in way $(\star 2(i+1))$ and $(\ast 2(i+1))$ will still be verified.\\
First, we must construct $s_{2i+1}$ such that $s_{2i} \overset{t_{i+1}}{\rightsquigarrow}^{\sf MH} s_{2i+1}$.  Suppose $s_{2i} = 
\break \lbrace (\ell_k, I_k, m_k)_{k \in K} \rbrace$, as $(\star 2i)$ is verified by hypothesis, it means that $C_{2i}$ can also be written as $\lbrace (\ell_k, I_k)_{k \in K} \rbrace$.  We must choose $s_{2i+1}$ to be $\lbrace (\ell_k, I_k + t_{i+1}, m_k)_{k \in K} \rbrace$.  As $C_{2i} \overset{t_{i+1}}{\rightsquigarrow} C_{2i+1}$, $C_{2i+1} = C_{2i} + t_{i+1} = \lbrace (\ell_k, I_k + t_{i+1})_{k \in K} \rbrace$.  Remark that the following property holds:
\begin{enumerate}
\item[$(\star 2i+1)$] $C_{2i+1} = \lbrace (\ell_k, I_k)_{k \in K} \rbrace \;$ iff $\; s_{2i+1} = \lbrace (\ell_k, I_k, m_k)_{k \in K} \; \vert \; m_k \in \lbrace \top, \bot \rbrace \rbrace$.
\end{enumerate}
Secondly, we must construct $s_{2(i+1)}$ such that $s_{2i+1} \overset{\sigma_{i+1}}{\rightarrow} s_{2(i+1)}^{\sf MH}$.  We know that $C_{2i+1} \overset{\sigma_{i+1}}{\longrightarrow}_{f} C_{2(i+1)}$. Suppose $s_{2i+1} = \lbrace (\ell_k, I_k, m_k)_{k \in K} \rbrace$, as $(\star 2i+1)$ is verified, $C_{2i+1}$ can be written as $\lbrace (\ell_k, I_k)_{k \in K} \rbrace$.  Let us further suppose that $C_{2(i+1)} = \lbrace (\ell_{k'}, I_{k'})_{k' \in K'} \rbrace$.  We construct $s_{2(i+1)} = \lbrace (\ell_{k'}, I_{k'}, m_{k'})_{k' \in K'} \rbrace$ to be the unique state\footnote{once the minimal models of the definition of $\overset{\sigma_{i+1}}{\rightarrow}$ are chosen, it is easy to see there exists a unique possible choice for the values of the $m_k$ of $s_{2(i+1)}$.} of $S^{\sf MH}$ such that $\lbrace (\ell_k, I_k)_{k \in K} \rbrace \overset{\sigma_{i+1}}{\longrightarrow}_{f} \lbrace (\ell_{k'}, I_{k'})_{k' \in K'} \rbrace$. $(\star 2(i+1))$ is trivially verified.  It stays to prove that $(\ast 2(i+1))$ is satisfied.  Suppose that a location of $F$ occurs on all the branches of $\pi$ between the last configuration $C_{2j'}$ such that $s_{2j'} \in \alpha$ (or, failing that, between $C_{0}$) and $C_{2(i+1)}$. We must prove that $s_{2(i+1)} \in \alpha$, i.e. all the trios of $s_{2(i+1)}$ has $\top$ as last component.  Let $\beta = \beta_0 \beta_1 \beta_2 \dots \beta_{2j'} \dots \beta_{2j} \dots \beta_{2(i+1)} \dots$ be a branch of $\pi$.  The hypothesis implies there exists a transition $\overset{\sigma_{j}}{\longrightarrow}_{f}$, for $j' \leq j \leq i+1$, such that $\beta_{2j} = (\ell, I)$ for a certain $\ell \in F$. By $(\ast 2j)$, $(\ell, I, m_k) \in s_{2j}$ for a certain $m_k$ in $\lbrace \top, \bot \rbrace$, but point (i) (b) of the definition of $\rightarrow$ obliges $m_k$ to be $\top$.  So, the third components associated in $\pi'$ to the different states of the branches of $\pi$ will gradually (between $s_{2j'}$ and $s_{2(i+1)}$) become $\top$.  We must still ensure they will eventually never become $\bot$ again.  In fact, a pair $(\ell, I)$ of a certain state $s_j$ (for $2j' \leq j \leq 2(i+1)$) of $\pi$, corresponding to $(\ell, I, \top)$ in $\pi'$, can have a successor $(\ell', I')$ in $\pi$ such that this corresponds to $(\ell', I',\bot)$ in $\pi'$ iff $s_j$ is accepting and $\ell' \notin F$ (which is not possible under the present hypothesis) or $(\ell', I', \bot)$ comes from the grouping of trios (thanks to $f$) emanating from trios such that at least one of them had $\bot$ as last component (case (i) (b) of the definition of $\rightarrow$).  It means that no location of $F$ occurs on one of the branches of $\pi$ leading to $(\ell', I')$, say $\beta' = \beta'_0 \beta'_1 \beta'_2 \dots \beta_{2j'} \dots \beta'_{2k'} \dots \beta'_{2(i+1)} \dots$, with $\beta'_{2k} = (\ell',I')$, since $\beta'_{2j'}$ (else, we contradict case (i) (c)). But, by hypothesis, a location of $F$ occurs on all the branches of $\pi$ between steps $2j'$ and $2(i+1)$, so there exists a transition $\overset{\sigma_{2\tilde{j}}}{\longrightarrow}_{f}$, for $k' < \tilde{j} \leq i+1$, such that $\beta'_{2\tilde{j}}$ has its location in $F$: once again, point (i) (b) of the definition of $\rightarrow$ obliges $m_{2\tilde{j}}$ to be $\top$.\\
As a location of $F$ occurs on all branches of $\pi$ between steps $2j'$ and $2(i+1)$ and there is only a finite number of branches leading to a state of $C_{2(i+1)}$, we conclude that we can only encounter this last case a finite number of time and so $s_{2(i+1)} \in \alpha$: $(\ast 2(i+1))$ is satisfied.\\
To end this part of the proof, we must show that $\pi'$ is accepting.  The previous induction proves that $(\ast 2j)$ is verified for all $j \geq 0$.  As $\pi$ is accepting, a location of $F$ occurs on all the branches of $\pi$ infinitely often, and so there is an infinite number of j and j' such that the antecedent of $(\ast 2j)$ is true.  So, in this same infinite number of times, we know that $s_{2j} \in \alpha$, what proves that $\pi'$ is accepting.\\
$(\subseteq)$ Let $\theta = (\bar{\sigma}, \bar{\tau}) \in L^\omega(\mhts{\Aa,f})$, with $\bar{\sigma}=\sigma_1\sigma_2\cdots\sigma_n \dots$ and $\bar{\tau}=\tau_1\tau_2\cdots\tau_n \dots$.  We will prove that $\theta \in L^\omega_{f}(\Aa)$.  Let us note $t_{i} = \tau_{i} - \tau_{i-1}$ for all $1 \leq i \leq \vert \theta \vert$, assuming $\tau_{0} = 0$.  We have an accepting run of $\mhts{\Aa,f}$ on $\theta$, say $\pi$: $s_{0} \overset{t_{1}}{\rightsquigarrow}^{\sf MH} s_{1} \overset{\sigma_{1}}{\rightarrow}^{\sf MH} s_{2} \overset{t_{2}}{\rightsquigarrow}^{\sf MH} s_3 \overset{\sigma_{2}}{\rightarrow}^{\sf MH} \dots \overset{t_{i}}{\rightsquigarrow}^{\sf MH} s_{2i-1} \overset{\sigma_{i}}{\rightarrow}^{\sf MH} s_{i} \dots$.  We must prove that there is an accepting $f$-run of $\Aa$ on $\theta$, say $\pi'$: $C_{0} \overset{t_1}{\rightsquigarrow} C_{1} \overset{\sigma_{1}}{\longrightarrow}_{f} C_{2} \overset{t_2}{\rightsquigarrow} C_3 \overset{\sigma_{2}}{\longrightarrow}_{f} \dots \overset{t_i}{\rightsquigarrow} C_{2i-1} \overset{\sigma_{i}}{\longrightarrow}_{f} C_{2i} \dots$.  We construct $\pi'$ by induction, proving additionally that the two following properties hold for $j \geq 0$:
\begin{enumerate}
\item[($\star 2j$)] $C_{2j} = \lbrace (\ell_k, I_k)_{k \in K} \rbrace \;$ iff $\; s_{2j} = \lbrace (\ell_k, I_k, m_k)_{k \in K} \; \vert \; m_k \in \lbrace \top, \bot\rbrace \rbrace$ ;
\item[($\ast 2j$)] if $s_{2j} \in \alpha$, then, a location of $F$ occurs on all the branches of $\pi'$ between the last configuration $C_{2j'}$ such that $s_{2j'} \in \alpha$ (or, failing that, between $C_{0}$) and $C_{2j}$.
\end{enumerate}
\underline{Basis:} We know that $s_0 = \lbrace (\ell_0, [0,0], \bot) \rbrace$ and $C_{0} = (\ell_{0}, [0,0])$. $(\star 0)$ and $(\ast 0)$ are trivially verified.\\
\underline{Induction:} Suppose that we constructed $\pi'$ until $C_{2i}$ and that $(\star 2j)$ and $(\ast 2j)$ are verified $\forall 0 \leq j \leq i$.  We will construct $\pi'$ until $C_{2(i+1)}$ in way $(\star 2(i+1))$ and $(\ast 2(i+1))$ will still hold.\\
First, we must construct $C_{2i+1}$ such that $C_{2i} \overset{t_{i+1}}{\rightsquigarrow} C_{2i+1}$.  We know that $s_{2i} \overset{t_{i+1}}{\rightsquigarrow}^{\sf MH} s_{2i+1}$. Suppose $s_{2i} = \lbrace (\ell_k, I_k, m_k)_{k \in K} \rbrace$, as $(\star 2i)$ is verified by hypothesis, $C_{2i} = \lbrace (\ell_k, I_k)_{k \in K} \rbrace$.  We must choose $C_{2i+1}$ to be $C_{2i+1} = \lbrace (\ell_k, I_k + t_{i+1})_{k \in K} \rbrace$.  As $s_{2i} \overset{t_{i+1}}{\rightsquigarrow}^{\sf MH} s_{2i+1}$, $s_{2i+1} = \lbrace (\ell_k, I_k + t_{i+1}, m_k)_{k \in K} \rbrace$.  Remark that the following property holds:
\begin{enumerate}
\item[$(\star 2i+1)$] $C_{2i+1} = \lbrace (\ell_k, I_k)_{k \in K} \rbrace \;$ iff $\; s_{2i+1} = \lbrace (\ell_k, I_k, m_k)_{k \in K} \; \vert \; m_k \in \lbrace \top,\bot \rbrace \rbrace$.
\end{enumerate}
Secondly, we must construct $C_{2(i+1)}$ such that $C_{2i+1} \overset{\sigma_{i+1}}{\longrightarrow}_{f} C_{2(i+1)}$.  We know that $s_{2i+1} \overset{\sigma_{i+1}}{\rightarrow}^{\sf MH} s_{2(i+1)}$. Suppose $s_{2i+1} = \lbrace (\ell_k, I_k, m_k)_{k \in K} \rbrace$, as $(\star 2i+1)$ is verified, $C_i = \lbrace (\ell_k, I_k)_{k \in K} \rbrace$.  Let us suppose further suppose that $s_{i+1} = \lbrace (\ell_{k'}, I_{k'}, m_{k'})_{k' \in K'} \rbrace$.  We construct $C_{2(i+1)} = \lbrace (\ell_{k'}, I_{k'})_{k' \in K'} \rbrace$ so that $(\star 2(i+1))$ holds.  Remark that, by definition of $\rightarrow^{\sf MH}$, as $s_{2i+1} \overset{\sigma_{i+1}}{\rightarrow}^{\sf MH} s_{2(i+1)}$, $C_{2i+1} \overset{\sigma_{i+1}}{\longrightarrow}_{f} C_{2(i+1)}$, what we wanted. It stays to prove that $(\ast 2(i+1))$ is satisfied. Suppose that $s_{2(i+1)} \in \alpha$ (i.e. $\forall k' \in K'$, $m_{k'} = 1$).  We must prove that, between the last configuration $C_{2j'}$ such that $s_{2j'} \in \alpha$ (or, failing that, between $C_{0}$) and $C_{2(i+1)}$, a location of $F$ occurs on all the branches of $\pi'$.  Let $\beta = \beta_0 \beta_1 \beta_2 \dots \beta_{2j'} \dots \beta_{2j} \dots \beta_{2(i+1)} \dots$ be a branch of $\pi'$ and suppose by contradiction that $\forall j' \leq j \leq i+1$, $\beta_{2j}$ has not its location in $F$.  As $s_{j2'} \in \alpha$, all the third components of $s_{2j'}$ are replaced by $\bot$ (likewise, by definition of $s_{0}$, its unique trio has $\bot$ as last component) before evolving reading $\sigma_{j'+1}, \sigma_{j'+2}, \dots, \sigma_{i+1}$ thanks to the rules in (i) in the definition of $\rightarrow^{\sf MH}$.  But, observing rules in (i), when a trio has $\bot$ as last component, it can only evolve to a trio with $\top$ as last component if case (2) it satisfied, what is impossible along $\beta$.  This contradicts the fact that $s_{2(i+1)} \in \alpha$.\\
To end the proof, we must show that $\pi'$ is accepting. The previous induction proves that $(\ast 2j)$ holds for all $j \geq 0$.  As $\pi$ is accepting, we know that $s_{2j} \in \alpha$ for infinitely many $j$ and so, between any two successive such $j$'s all the branches of $\pi'$ saw $F$.  We conclude that all the branches of $\pi'$ saw $F$ infinitely often, what proves that $\pi'$ is accepting.\qed
\end{proof}

\paragraph{Büchi timed automata for MITL formulas.} Building on this formalisation, one can define a timed automaton with Büchi acceptance condition $\Bb_\varphi 
\break =(\Sigma, B, b_0, X, F^\Bb,
\delta^\Bb)$ that simulates $\mhts{\Aa_\varphi,\appfunc{\varphi}}$,
and thus accepts $\sem{\varphi}$, for every MITL formula
$\varphi$.\\
A Büchi timed automaton (TA) is a tuple $\Bb =(\Sigma, B, b_0, X, F^\Bb,
\delta^\Bb)$, where $\Sigma$ is a finite \emph{alphabet}, $B$ is a finite
set of \emph{locations}, $b_0\in B$ is the \emph{initial location},
$X$ is a finite set of \emph{clocks}, $F^\Bb \subseteq B$ is the set of
\emph{accepting locations}, and $\delta^\Bb \subseteq
B \times \Sigma \times \mathcal{G}(X)\times 2^X\times B$ is a finite set
of \emph{transitions}, where $\mathcal{G}(X)$ denotes the set of
\emph{guards on $X$}, i.e. the set of all finite conjunctions of
\emph{clock constraints} on clocks from $X$. For a transition $(b,
\sigma, g, r, b')$, we say that $g$ is its \emph{guard}, and $r$
its \emph{reset}. A \emph{configuration} of a TA is a pair $(b,
v)$, where $v:X\mapsto\R^+$ is a \emph{valuation} of the clocks in
$X$. We denote by $\configs{\Bb}$ the set of all configurations of
$\Bb$. For all $t\in \R^+$, we have (time successor)
$(b,v)\timestep{t}(b',v')$ iff $b=b'$ and $v'=v+t$ where
$v+t$ is the valuation s.t. for all $x\in X$: $(v+t)(x)=v(x)+t$. For
all $\sigma\in \Sigma$, we have (discrete successor)
$(b,v)\xrightarrow{\sigma}(b',v')$ iff there is $(b, \sigma,
g, r, b')\in \delta$ s.t. $v\models g$, for all $x\in r$: $v'(x)=0$
and for all $x\in X\setminus r$: $v'(x)=v(x)$. We write
$(b,v)\xrightarrow{t,\sigma}(b',v')$ iff there is
$(b'',v'')\in \configs{\Bb}$
s.t. $(b,v)\timestep{t}(b'',v'')\xrightarrow{\sigma}(b',v')$.
Let $\theta=(\bar{\sigma},\bar{\tau})$ be a timed word with
$\bar{\sigma}=\sigma_1\sigma_2\cdots\sigma_n \dots$ and
$\bar{\tau}=\tau_1\tau_2\cdots\tau_n \dots$.  Let us note $t_{i} = \tau_{i} -
\tau_{i-1}$ for all $i \geq 1$, assuming
$\tau_{0} = 0$.  A \emph{run} of $\Bb$ on $\theta$ is an infinite
sequence of successor steps that
is labelled by $\theta$, i.e. a sequence of the form: $(b_0, v_0)
\xrightarrow{t_1, \sigma_1} (b_1, v_1)
\xrightarrow{t_2, \sigma_2} (b_2, v_2)
\xrightarrow{t_3, \sigma_3} (b_3, v_3)
\dots$, 
where $v_0$ assigns $0$ to all clocks.  Such a run of $\Bb$ is accepting iff there is infinitely many $(b_i,v_i)$ with an accepting $b_i$.  We say that a timed word $\theta$ is accepted by $\Bb$ iff there exists an accepting run of $\Bb$ on $\theta$. We denote by $L^{\omega}(\Bb)$ the \emph{language} of $\Bb$, i.e. the set of infinite timed words accepted by $\Bb$.\\
For an MITL formula $\varphi$, we construct $\Bb_\varphi =(\Sigma, B, b_0, X, F^\Bb, \delta^\Bb)$ as follows. 
Locations of $\Bb_\varphi$ associate with each location
$\ell$ of $\Aa_\varphi$ a sequence of triples $(x,y,m)$, where $x$ and
$y$ are clocks that store the infimum and supremum of an 
interval respectively, and $m$ is a Miyano-Hayashi marker. 
Formally, for a set of clocks $X$, we let
$\loc(X)$ be the set of functions $S$ that associate with each $\ell
\in L$ a finite sequence $(x_1,y_1,m_1),\ldots,(x_n,y_n,m_n)$ where,
for $1\leq i \leq n$, $m_i \in \{ \top, \bot \}$ and $(x_i,y_i)$ is a
pair of clocks from $X$ s.t. each clock only occurs once in
$S(L)$. Then:
\begin{itemize}
\item $X$ is the set of clocks of $\Bb_\varphi$ s.t. $|X|=M(\varphi)$;
\item $B = \lbrace S \in \loc(X) \rbrace$ is the set of locations of
  $\Bb_\varphi$. Thus, a configuration $(S, v)$ of $\Bb_\varphi$
  (where $S$ is the location and $v$ the valuation of the clocks $X$)
  encodes the labeled configuration $C=\{(\ell, [v(x), v(y)], m)\mid
  (x,y,m)\in S(\ell)\}$;
\item $b_0 = S_0$ is the initial location of $\Bb_\varphi$ and is
  s.t. $\forall \ell \in L \setminus \{ \ell_0 \}$, $S_0(\ell) =
  \emptyset$, and $S_0(\ell_0)=(x,y,\bot)$, where $x$ and $y$ are two
  clocks arbitrarily chosen from $X$;
\item $F^\Bb = \{ S \in B \: \vert \: (x,y,m) \in S(L) \Rightarrow m =
  \top \}$ is the set of final locations of $\Bb_\varphi$.
\end{itemize}
Finally, we must define the set of transitions $\delta^\Bb$ to let
$\Bb_\varphi$ simulate the executions of $\Aa_\varphi$. First, we observe
that, for each location $\ell\in L$, for each $\sigma\in\Sigma$,
all \emph{arcs} in $\delta$ are either of the form
$(\ell,\sigma,\mathit{true})$ or $(\ell,\sigma,\mathit{false})$ or of
the form $\big(\ell,\sigma,\ell\wedge
x.(\ell_1\wedge\cdots\wedge\ell_k)\wedge g\big)$ or of the form
$\big(\ell,\sigma,x.(\ell_1\wedge\cdots\wedge\ell_k)\wedge g\big)$,
where $g$ is \emph{guard} on $x$, i.e. a finite conjunction of clock
constraints on $x$. Let $S \in B$ be a location of
$\Bb_\varphi$, $\ell \in L$, $\sigma\in\Sigma$ be a
letter and $(x,y,m)$ be a 3-tuple occurring in $S(\ell)$. Let
us associate to this 3-tuple an arc $a$ of $\delta$ of the form
$(\ell,\sigma,\gamma)$. Then, we associate to $a$ a \emph{guard}
$\gu{a}$, and two sets $\re{a}$ and $\de{a}$, defined as follows:
\begin{itemize}
\item if $\gamma\in\{\mathit{true},\mathit{false}\}$, then, $\gu{a}= \gamma$
  and $\re{a}=\de{a}=\emptyset$.
\item if $\gamma$ is of the form
  $x.(\ell_1\wedge\cdots\wedge\ell_k)\wedge g$, then $\gu{a}=g$,
  $\re{a}=\{\ell_1,\ldots,\ell_k\}$ and $\de{a}=\emptyset$.
\item  if $\gamma$ is of the form
  $\ell\wedge x.(\ell_1\wedge\cdots\wedge\ell_k)\wedge g$, then $\gu{a}=g$,
  $\re{a}=\{\ell_1,\ldots,\ell_k\}$ and $\de{a}=\{(x,y)\}$.
\end{itemize}
Thanks to those definitions, we can now define $\delta^\Bb$. Let $S^\vartriangle$ be a
location in $B$.  We want that to encounter an accepting location of $\Bb_\varphi$ in a run corresponds to the occurrence of a location of $F$ on all branches of the corresponding run of $\Aa_\varphi$.  When such an accepting location of $\Bb_\varphi$ is encountered, we need to start again the marking: the following definition of $S$ reflects this need. If $S^\vartriangle \in F^\Bb$, we let $S$ to be such that: $\forall \ell \in L$, $S(\ell) = \{ (x,y,\bot) \: \vert \: (x,y,\top) \in S^\vartriangle(\ell) \}$.  If $S^\vartriangle \notin F^\Bb$, we let $S = S^\vartriangle$.  We assume:
\begin{align*}
  \{(\ell_1,x_1,y_1,m_1),\ldots, (\ell_k,x_k,y_k,m_k)\} &= \{(\ell,x,y,m)\mid
  (x,y,m)\in S(\ell) \cup \bar{S} (\ell) \}.
\end{align*}
\noindent Then, $(S^\vartriangle, \sigma, g, r, S') \in \delta^\Bb$ iff there is a set $A = \lbrace (a_i)_{i=1}^{k} \rbrace$:
\begin{itemize}
\item For all $i \in \{ 1, \dots, k \}$: $a_i$ is an arc of $\delta$ of the form $(\ell,\sigma,\gamma_i)$ associated to $(x_i, y_i, m_i) \in S(\ell_i)$.
\item For each $\ell\in L \setminus F$, we let $S^\star(\ell)=(x^{\star}_1,y^{\star}_1,m^{\star}_1) (x^{\star}_2,y^{\star}_2,m^{\star}_2) \cdots (x^{\star}_n,y^{\star}_n,m^{\star}_n)$ be obtained from $S(\ell)$ by deleting all the trios $(x,y,m)$ such that $(x,y)\notin \bigcup_{i=1}^k \de{a_i}$. 
For each $\ell\in L \cap F$, we let $S^\star(\ell) = (x^{\star}_1,y^{\star}_1,\top) (x^{\star}_2,y^{\star}_2,\top) \cdots (x^{\star}_n,y^{\star}_n,\top)$ be obtained from $S(\ell)$ by deleting all the trios $(x,y,m)$ such that $(x,y)\notin\bigcup_{i=1}^k \de{a_i}$.
Then, for all $\ell \in L$:\medskip
\begin{enumerate}
\item \textit{if $\ell\notin\bigcup_{i=1}^k \re{a_i}$:} \medskip\\
$S'(\ell) = S^{\star}(\ell)$\medskip
\item \textit{else, if $\ell \in F$ (in particular, for $1 \leq i \leq n$, $m^{\star}_i = \top$):}\medskip\\
$S'(\ell) \in \big\{ (x,y,\top) \cdot S^\star(\ell) \; , \; (x,y^{\star}_1,m^{\star}_1) (x^{\star}_2,y^{\star}_2,m^{\star}_2) \cdots (x^{\star}_{n},y^{\star}_n,m^{\star}_n) \big\}$\medskip
\item \textit{else, if $\ell \in \underset{\underset{m_i = 0}{i\in \{1, \dots, k\}}}{\bigcup} \re{a_i}$:}\medskip\\
$S'(\ell) \in \big\{ (x,y,\bot) \cdot S^\star(\ell) \; , \; (x,y^{\star}_1,\bot) (x^{\star}_2,y^{\star}_2,m^{\star}_2) \cdots (x^{\star}_{n},y^{\star}_n,m^{\star}_n) \big\}$\medskip
\item \textit{else} (i.e. $\ell \notin \underset{\underset{m_i = 0}{i\in \{1, \dots, k\}}}{\bigcup} \re{a_i}$ and $\ell \in \underset{\underset{m_i = 1}{i\in \{1, \dots, k\}}}{\bigcup} \re{a_i}$)\medskip\\
$S'(\ell) \in \big\{ (x,y,\top) \cdot S^\star(\ell) \; , \; (x,y^{\star}_1,m^{\star}_1) (x^{\star}_2,y^{\star}_2,m^{\star}_2) \cdots (x^{\star}_{n},y^{\star}_n,m^{\star}_n) \big\}$\medskip
\end{enumerate}
When $S'(\ell)=(x,y,\bot) \cdot S^\star(\ell)$ or $(x,y,\top)\cdot \bar{S}^\star(\ell)$, we let $R_\ell=\{x,y\}$;
when $S'(\ell)=(x,y^{\star}_1,\bot)(x^{\star}_2,y^{\star}_2,m^{\star}_2)\cdots(x^{\star}_n,y^{\star}_n,m^{\star}_n)$ or $(x,y^{\star}_1,m^{\star}_1)(x^{\star}_2,y^{\star}_2,m^{\star}_2)\cdots(x^{\star}_n,y^{\star}_n,m^{\star}_n)$, we let $R_\ell=\{x\}$; and we let $R_\ell=\emptyset$ otherwise.\medskip
\item $g=\bigwedge_{1\leq i\leq k}(\gu{a_i}[x/x_i]\wedge
  \gu{a_i}[x/y_i])$.
\item $r=\cup_{\ell\in L}R_\ell$.
\end{itemize}

\begin{theorem}
$L^\omega(\Bb_\varphi)=L^\omega_{\appfunc{\varphi}}(\Aa_\varphi)$.
\end{theorem}
\begin{proof}
To prove this, we will show that the transition system $S_{\Bb_\varphi} = (\configs{\Bb_\varphi}, 
\break \rightsquigarrow, \longrightarrow)$ induced by $\Bb_\varphi$ in the classical semantics is $\mhts{\Aa_\varphi, \appfunc{\varphi}}$ in which $(S,v) \in \configs{\Bb_\varphi}$ corresponds to $\{ (\ell, [v(x),v(y)],m) \: \vert \: (x,y,m) \in S(\ell) \}$.  It is easy to see that the initial configuration of $\Bb_\varphi$, $(S_0,v_0)$, where for all $x \in X$, $v_0(x) = 0$ corresponds to the initial state $s_0$ of $\mhts{\Aa_\varphi, \appfunc{\varphi}}$.  Now, suppose that we reached a configuration $(S,v) \in \configs{\Bb_\varphi}$ corresponding to the state $T$ of $\mhts{\Aa_\varphi, \appfunc{\varphi}}$, i.e. $T = \{ (\ell, [v(x),v(y)],m) \: \vert \: (x,y,m) \in S(\ell) \}$.\\
\textbf{\underline{Timed transition:}} let $t \in \R$, $(S,v) \overset{t}{\rightsquigarrow} (S,v+t)$ while $T = \bigcup_{\ell \in L} \{ (\ell, [v(x),v(y)],m) \: \vert \: \break
(x,y,m) \in S(\ell) \} \overset{t}{\rightsquigarrow}^{\sf MH} \bigcup_{\ell \in L} \{ (\ell, [v(x),v(y)]+t, m) \: \vert \: (x,y,m) \in S(\ell) \}$: these two images correspond.\\
\textbf{\underline{Discrete transition:}}\\
\textbf{\underline{From $S_{\Bb_\varphi}$ to $\Ss_{\Aa, \appfunc{\varphi}}$:}} Suppose that $(S^\vartriangle,v) \overset{\sigma}{\longrightarrow} (S',v')$ and that $T^\vartriangle$ corresponds to $(S^\vartriangle,v)$, i.e. $T^\vartriangle = \bigcup_{\ell \in L} \{ (\ell, [v(x),v(y)],m) \: \vert \: (x,y,m) \in S^\vartriangle(\ell) \}$.  We will show that there exists $T'$ such that $T^\vartriangle \overset{\sigma}{\longrightarrow} T'$ in $\Ss_{\Aa, \appfunc{\varphi}}$ and $(S',v')$ and $T'$ correspond.  As, $(S^\vartriangle,v) \overset{\sigma}{\longrightarrow} (S',v')$, if $S^\vartriangle \in F^\Bb$, $S^\vartriangle$ is turned to $S = \{ (x,y,\bot) \: \vert \: (x,y,\top) \in S(L) \}$.  In this case, $T^\vartriangle \in \alpha$ and it is turned to $T = \bigcup_{\ell \in L} \{ (\ell, [v(x),v(y)], \bot) \: \vert \: \break
(x,y,\top) \in S(\ell) \}$ which corresponds to $(S,v)$.  If $S^\vartriangle \notin F^\Bb$, $T^\vartriangle \notin \alpha$ and we let $S = S^\vartriangle$ and $T = T^\vartriangle$.  Then, for all $\ell \in L$, an arc $a_{(x,y)}$ must have been associated to each $(x,y,m) \in S(\ell)$.  $S^\star$ is then defined in way each $(x,y,m)$ associated to an arc that loop is still associated to the same location while the others disappear.  In the last 4 cases (1., 2., 3. and 4.) all the locations to which an arc goes without looping (characterized by the fact that there is a reset through this location in our automata for MITL formula $\Aa_\varphi$) are considered: either two new reset (!) clocks are associated to these locations, or a new reset (!) clock replace the clock $x^{\star}_1$ of the first pair of clocks associated to this location (and so, to the smallest represented interval).  Whatever is the case thanks to which $(S',v')$ has been formed, taking the corresponding arc of $\Bb_\varphi$, each of the intervals $(x,y)$ of $S$ must have satisfied the clock constraints $\gu{a_{(x,y)}}[x/x_i]\wedge \gu{a_{(x,y)}}[x/y_i]$, what corresponds to the fact that the whole interval $[v(x),v(y)]$ satisfies $\gu{a_{(x,y)}}[x/x_i]$ (because this guard is an interval and is so convex).  Moreover, $v'$ is obtained from $v$ by associating $\bot$ to each clock reset in the case 1., 2., 3. or 4. used to create $(S',v')$, and by associating $v(x)$ to all other clock x.  This treatment of $(S,v)$ to obtain $(S',v')$ exactly correspond to the fact that $\bigcup_{\ell \in L} \{ (\ell,[v(x),v(y)]) \: \vert \: (x,y,m) \in S(\ell) \} 
\overset{\sigma}{\longrightarrow}_{\appfunc{\varphi}} \bigcup_{\ell \in L} \{ (\ell,[v'(x),v'(y)]) \: \vert \: (x,y,m) \in S'(\ell) \}$ thanks to the minimal models obtained following arc $a_{(x,y)}$ from each $(\ell, [v(x),v(y)])$ ($\appfunc{\varphi}$ is applied to the result in way a merging is effectuated iff clock $x^{\star}_1$ is replaced by a new reset clock from $S$ to $S'$). We so take $T' = \bigcup_{\ell \in L} \{ (\ell,[v'(x),v'(y)], m) \: \vert \: (x,y,m) \in S'(\ell) \}$.  It is not difficult to see that, whatever is the case thanks to which $(S',v')$ has been formed, the marking of pairs of clock of $S'$ enables $T'$ to satisfy conditions (b) and (c) of the definition of the transition $\longrightarrow$ of  $\Ss_{\Aa, \appfunc{\varphi}}$.  We so have $T^\vartriangle \overset{\sigma}{\longrightarrow} T'$ in $\Ss_{\Aa, \appfunc{\varphi}}$ and $(S',v')$ and $T'$ correspond.\\
\textbf{\underline{From $\Ss_{\Aa, \appfunc{\varphi}}$ to $S_{\Bb_\varphi}$:}} Suppose that $T^\vartriangle \overset{\sigma}{\longrightarrow} T'$ in $\Ss_{\Aa, \appfunc{\varphi}}$ and that $(S^\vartriangle,v)$ and $T^\vartriangle$ correspond, i.e. $T^\vartriangle = \bigcup_{\ell \in L} \{ (\ell, [v(x),v(y)],m) \: \vert \: (x,y,m) \in S^\vartriangle(\ell) \}$. We will show that there exists $(S',v')$ such that $(S^\vartriangle,v) \overset{\sigma}{\longrightarrow} (S',v')$.  As $T^\vartriangle \overset{\sigma}{\longrightarrow} T'$, if $T^\vartriangle \in \alpha$, $T^\vartriangle$ is turned to $T = \bigcup_{\ell \in L} \{ (\ell, [v(x),v(y)], \bot) \: \vert \: (x,y,\top) \in S(\ell) \}$.  In this case, $S^\vartriangle \in F^\Bb$ and is turned to $S = \{ (x,y,\bot) \: \vert \: (x,y,\top) \in S(L) \}$, so that $(S,v)$ corresponds to $T$.  If $T^\vartriangle \notin \alpha$, $S^\vartriangle \notin F^\Bb$ and we let $T = T^\vartriangle$ and $S = S^\vartriangle$.  Then, $T \overset{\sigma}{\longrightarrow} T'$ and the fact that $T$ and $(S,v)$ correspond imply that $\bigcup_{\ell \in L} \{ (\ell,[v(x),v(y)]) \: \vert \: (x,y,m) \in S(\ell) \} 
\overset{\sigma}{\longrightarrow}_{\appfunc{\varphi}} \{ (\ell,I) \: \vert \: (\ell,I,m) \in T' \}$.  It means an arc $a_{(x,y)}$ has been chosen for each $(x,y)$ to create a minimal model of $\delta(\ell, \sigma)$ wrt $[v(x),v(y)]$.  We take $(S',v')$ to be the \textbf{unique} successor of $(S,v)$ in $S_{\Bb_\varphi}$ obtained associating to each $(x,y, m)$ the arc $a_{(x,y)}$ and such that $R_\ell$ is a singleton iff the application of $\appfunc{\varphi}$ merged the new interval created in location $\ell$ with the previous smallest one (it is not difficult to see that we well obtain a \textbf{unique} successor this way observing the definition of $\delta^\Bb$).  As detailed in the section "From $S_{\Bb_\varphi}$ to $\Ss_{\Aa, \appfunc{\varphi}}$" of this proof, $(S',v')$ will be such that $\bigcup_{\ell \in L} \{ (\ell,[v(x),v(y)]) \: \vert \: (x,y,m) \in S(\ell) \} 
\overset{\sigma}{\longrightarrow}_{\appfunc{\varphi}} \bigcup_{\ell \in L} \{ (\ell,[v'(x),v'(y)]) \: \vert \: (x,y,m) \in S'(\ell) \}$ thanks to the minimal models obtained following arc $a_{(x,y)}$ from each $(\ell, [v(x),v(y)])$ ($\appfunc{\varphi}$ is applied to the result in way a merging is effectuated iff clock $x^{\star}_1$ is replaced by a new reset clock from $S$ to $S'$).  In other words, $T'$ and $(S',v')$ correspond, thanks to the fact that the unique possible marking present on the third components of elements of $T'$ will be the same than the obtained marking of elements of $(S',v')$ (it is easy to see, observing all the cases thanks to which $T'$ has been constructed).\qed
\end{proof}

\section{MITL model-checking and satisfiability with
  TOCATA\label{sec:mitl-model-checking}}
In this section, we fix an MITL formula $\varphi$ and a TA $\Bb=
(\Sigma, B, b_0, X, \delta^\Bb, F^\Bb)$, and we consider the two
following problems:
\begin{inparaenum}[(i)]
\item the \emph{model-checking problem} asks whether $L(\Bb)\subseteq
  \sem{\varphi}$;
\item the \emph{satisfiability problem} asks whether
  $\sem{\varphi}\neq\emptyset$.
\end{inparaenum}
The construction of the TA $\Bb_{\neg\varphi}$ from $\varphi$ of the
previous section allows to solve those problems using classical
algorithms \cite{AD94}. Unfortunately, building $\Bb_{\neg\varphi}$
can be prohibitive in practice. To mitigate this difficulty, we
present an efficient \emph{on-the-fly} algorithm to perform MITL
model-checking, which has as input the TA $\Bb$ and the \TATA
$\Aa_{\neg\varphi}$ (whose size is linear in the size of
$\varphi$). It consists in exploring symbolically the state space of
the timed transition system $\Ss_{\Bb,\neg\varphi}$ which is obtained
by first taking the synchronous product of
$\tts{\Aa_{\neg\varphi},\appfunc{\neg\varphi}}$ and the transition
system\footnote{See appendix~\ref{sec:buchi-timed-automata} for a
  formal definition} $\tts{\Bb}$ of $\Bb$ \cite{AD94}, and then
associating Miyano-Hayashi markers with its states, by adapting the
construction of $\mhts{\Aa,f}$ given above to cope with the
configurations of $\Bb$. Namely, we associate a Miyano-Hayashi marker with the
configurations of $\Bb$ too, and a state of $\Ss_{\Bb,\neg\varphi}$ is
accepting iff all markers (including the one on the $\Bb$
configuration) are $\top$.  Obviously, $L(\Bb)\subseteq \sem{\varphi}$
iff $\Ss_{\Bb,\neg\varphi}$ has no accepting run (i.e., no run
visiting accepting states infinitely often). Symmetrically, we can
solve the \emph{satisfiability problem} by looking for accepting run
in $\mhts{\Aa_\varphi,\appfunc{\varphi}}$ (since the techniques are
similar for model-checking and satisfiability, we will only detail the
former in this section).\\
Formally, we define the transition system $\Ss_{\Bb,\neg\varphi} = ( \Sigma, S, s_{0}, \rightsquigarrow, \rightarrow, \alpha)$ where: $S$ is the set of elements of the form $\lbrace (\ell_k, I_k, m_k)_{k \in K} \rbrace \cup \lbrace (b, v, m_\Bb) \rbrace$ where $\lbrace(\ell_k, I_k)_{k \in K}\rbrace$ is a configuration of $\Aa_{\neg \varphi}$, $(b, v)$ is a state of $\Bb$, $m_\Bb \in \{\top,\bot\}$ and $\forall k \in K, m_k \in \lbrace \top,\bot\rbrace$, $s_0 = \lbrace (\ell_0, [0,0], \bot), (b_0, v_0, m) \rbrace$, where $m = \top$ iff $b_0 \in F^\Bb$ ; $\alpha$ contains all the elements of $S$ of the form $\lbrace (\ell_k, I_k, \top)_{k \in K} \rbrace \cup \lbrace (b, v, \top) \rbrace$.\\
For $t \in \R$ and $s, s' \in S$, supposing $s = \lbrace (\ell_k, I_k, m_k)_{k \in K} \rbrace \cup \lbrace (b, v, m_\Bb) \rbrace$, we have $s \overset{t}{\rightsquigarrow} s'$ iff  $s' = \lbrace (\ell_k, I_k + t, m_k)_{k \in K} \rbrace \cup \lbrace (b, v + t, m_\Bb) \rbrace$. $\rightsquigarrow = \underset{t \in \R}{\bigcup} \overset{t}{\rightsquigarrow}$.\\
"$\rightarrow$" is defined in way that (1) a transition between two states of $\Ss_{\Bb,\neg\varphi}$ corresponds to a transition between the configurations of $\Aa_{\neg \varphi}$ they contain, (2) a transition between the states of $\Bb$ they contain and (3) the markers of the third component of states of $S$ are kept updated.
Formally, $\rightarrow \; = \; \underset{\sigma \in \Sigma}{\bigcup} \overset{\sigma}{\rightarrow}$, where:
\begin{enumerate}
\item[(i)] For $s \in S \setminus \alpha$ and $s' \in S$, supposing $s = \lbrace (\ell_k, I_k, m_k)_{k \in K} \rbrace \cup \lbrace (b, v, m_\Bb) \rbrace$ and $s' = \lbrace (\ell_{k'}, I_{k'}, m_{k'})_{k' \in K'} \rbrace \cup \lbrace (b', v', m'_\Bb) \rbrace$, $s \overset{\sigma}{\rightarrow} s'$ iff
\begin{enumerate}
\item $\lbrace (\ell_k, I_k)_{k \in K} \rbrace \overset{\sigma}{\longrightarrow}_{\appfunc{\neg \varphi}} \lbrace (\ell_{k'}, I_{k'})_{k' \in K'} \rbrace$ in $\Aa_{\neg \varphi}$, i.e 
$\break \lbrace (\ell_{k'}, I_{k'})_{k' \in K'} \rbrace \in \appfunc{\neg\varphi}( \succ(\lbrace (\ell_k, I_k)_{k \in K} \rbrace, \sigma))$, and $(b, v) \overset{\sigma}{\longrightarrow} (b', v')$ in $\Bb$;
\item $\forall k' \in K'$: $\left( \ell_{k'} \in F \; \Rightarrow \; m_{k'} = \top \right)$;
\item $\forall k' \in K'$ with $\ell_{k'} \notin F$: if $\exists k \in K$ s.t. $(\ell_k,I_k,\bot)\in s$ and  $(\ell_{k'},I_{k'}) \in\dest(\conf{s},\conf{s'},(\ell_k,I_k,\bot))$, we have $m_{k'}=\bot$; otherwise, $m_{k'} = \top$;
\item $m'_{\Bb} = \top \;$ iff $\; m_\Bb = \top$ or $b' \in F^\Bb$.
\end{enumerate}
\item[(ii)] For $s \in \alpha$ and $s' \in S$, supposing $s = \lbrace (\ell_k, I_k, \top)_{k \in K} \rbrace \cup \lbrace (b, v, \top) \rbrace$, $s \overset{\sigma}{\rightarrow} s'$ iff $\lbrace (\ell_k, I_k, \bot)_{k \in K} \rbrace \cup \lbrace (b, v, \bot) \rbrace \overset{\sigma}{\rightarrow} s'$ according to the rules in (i).
\end{enumerate}

Defining $\appfunc{\neg \varphi}$ on configurations of $\Bb \times \Aa_{\neg \varphi}$ in way $\appfunc{\neg\varphi}( \lbrace (\ell_k, I_k)_{k \in K} \rbrace \cup \lbrace (b, v) \rbrace ) = F^M(\lbrace (\ell_k, I_k)_{k \in K} \rbrace) \cup \lbrace (b, v) \rbrace$, with $M =\max\{2\times |L|,M(\varphi)\}$, the following proposition can be proven similarly as Proposition \ref{Prop3}.

\begin{proposition}
\label{Prop3Bis}
For every MITL formula $\varphi$, the associated $\Aa_{\neg \varphi}$ and $\appfunc{\neg \varphi}$, and for every Büchi timed automaton $\Bb$: $L^\omega (\Ss_{\Bb,\neg\varphi}) = L^\omega_{\appfunc{\neg \varphi}}(\Bb \times \Aa_{\neg \varphi})$.
\end{proposition}

\paragraph{Region-based algorithms} Since $\Ss_{\Bb,\varphi}$, we adapt
the region abstraction of \cite{OW07} in order to cope with:
\begin{inparaenum}
\item the valuations of the clocks that are now \emph{intervals}; and
\item the Miyano-Hayashi markers.
\end{inparaenum}
Following the approach of \cite{OW07} we represent each region by a
unique word. In the sequel, we note $\cmax$ the
maximal constant of automata $\Bb$ and $\Aa_{\neg \varphi}$.

\begin{definition}
We define the equivalence relation $\sim$ on $\R^{+}$ by: $u \sim v \:$ iff $\:$either $u$ and $v > c_{\max}$, or $u$ and $v \leq c_{\max}$, $\lceil u \rceil = \lceil v \rceil$ and $\lfloor u \rfloor = \lfloor v \rfloor$.  The set of classes of $\sim$, called regions, is $REG = \lbrace \lbrace 0 \rbrace, (0,1), \lbrace 1 \rbrace, (1,2), \dots, (c_{\max}-1,c_{\max}), \lbrace c_{\max} \rbrace, (c_{\max}, + \infty) \rbrace$.  We will note $Reg(v)$ the class of $v \in \R^{+}$.
\end{definition}
In the sequel, $\forall v \in \R^{+}$, we note $frac(v)$ the fractionnal part of $v$. Moreover, we suppose that $\forall v > c_{\max}, frac(v) = 0$.\\
We can now define an equivalence relation $\equiv$ on $S$.

\begin{definition}
Let $s$ and $s'$ be two states of $S$ such that the configurations of $\Aa_{\neg \varphi}$ they contain have the same cardinality. We suppose that $s = \lbrace (\ell_k, I_k, m_k)_{k \in K} \rbrace \cup \lbrace (b, v, m_\Bb)\rbrace$ and (without loss of generality) $s' = \lbrace (\ell'_{k}, I'_{k}, m'_{k})_{k \in K} \rbrace \cup \lbrace (b', v', m'_\Bb)\rbrace$.  Then, we define $s \equiv s'$ iff:
\begin{enumerate}
\item $b = b'$ and $\forall k \in K : \ell_{k} = \ell'_{k}$ ; $m_\Bb = m'_\Bb$ and $\forall k \in K : m_{k} = m'_{k}$,
\item $\forall 1 \leq p \leq n : v(x_{p}) \sim v'(x_{p}) \text{ and } \forall k \in K : ( \inf(I_{k}) \sim \inf(I'_{k}) \wedge \sup(I_{k}) \sim \sup(I'_{k}) )$,
\item $\forall 1 \leq p, q \leq n : frac(v(x_{p})) \bowtie frac(v(x_{q}))$ iff $frac(v'(x_{p})) \bowtie frac(v'(x_{q}))$,
\item $\forall k, k' \in K : frac(\inf(I_{k})) \bowtie frac(\inf(I_{k'}))$ iff $frac(\inf(I'_{k'})) \bowtie frac(\inf(I'_{k'}))$,
\item $\forall k, k' \in K : frac(\sup(I_{k})) \bowtie frac(\sup(I_{k'}))$ iff $frac(\sup(I'_{k})) \bowtie frac(\sup(I'_{k'}))$,
\item $\forall k, k' \in K : frac(\inf(I_{k})) \bowtie frac(\sup(I_{k'}))$ iff $frac(\inf(I'_{k})) \bowtie frac(\sup(I'_{k'}))$,
\item $\forall k \in K, \forall 1 \leq p \leq n : frac(\inf(I_{k})) \bowtie frac(v(x_{p}))$ iff $frac(\inf(I'_{k})) \bowtie frac(v'(x_{p}))$,
\item $\forall k \in K, \forall 1 \leq p \leq n : frac(\sup(I_{k})) \bowtie frac(v(x_{p}))$ iff $frac(\sup(I'_{k})) \bowtie frac(v'(x_{p}))$,
\end{enumerate}
where ${\bowtie} \in \lbrace <, =, > \rbrace$.
\end{definition} 

In this definition, the same indices have been chosen for the two configurations of $\Aa_{\neg \varphi}$ in aim to make "correspond" $(\ell_{k},I_{k})$ with $(\ell'_{k},I'_{k})$.  Condition 1. forces corresponding $(\ell_{k},I_{k}, m_k)$ and $(\ell'_{k},I'_{k}, m'_k)$, as well as the two states of $\Bb$, to have (resp.) the same locations and markers.  Condition 2. forces intervals of corresponding $(\ell_{k},I_{k})$ and $(\ell'_{k},I'_{k})$ to have their infima and suprema in the same regions, and values of the same clocks of the two states of $\Bb$ to be in the same regions. The other conditions forces all the clock values present in $s$ (clock values, values of infimum and supremum of intervals) to present the same fractional part order than the corresponding clock values present in $s'$.

\begin{proposition}[Time-abstract bisimulation]
\label{Bisim}
Let $s_1, s_2 \in S$ such that $s_1 \equiv s_2$.  Then, for each transition $s_1 \overset{t}{\rightsquigarrow} z_1$ with $t \in \R^{+}$ (resp. $s_1 \overset{\sigma}{\longrightarrow}_{\appfunc{\neg \varphi}} z_1$, with $\sigma \in \Sigma$) and $z_1 \in S$, there exists $t' \in \R^{+}$ and $z_2 \in S$ such that $s_2 \overset{t'}{\rightsquigarrow} z_2$ (resp. there exists $z_2 \in S$ such that $z_1 \overset{\sigma}{\longrightarrow}_{\appfunc{\neg \varphi}} z_2$) and $z_1 \equiv z_2$.
\end{proposition}
\begin{proof}
Let us note $s_1 = \lbrace (\ell_{k},I_{k}, m_k)_{k \in K} \rbrace \cup \lbrace (b, v, m_b) \rbrace$, as $s_1 \equiv s_2$, we can suppose that $s_2 = \lbrace (\ell_{k},I'_{k}, m_k)_{k \in K} \rbrace \cup \lbrace(b, v', m_b)\rbrace$.  In the following, we will use the following notation $C = \lbrace (\ell_{k},I_{k})_{k \in K} \rbrace$, $C' = \lbrace (\ell_{k},I'_{k})_{k \in K} \rbrace$, $s = \lbrace (b, v) \rbrace$ and $s' = \lbrace(b, v')\rbrace$.\\
First, suppose that $s_1 \overset{\sigma}{\longrightarrow}_{\appfunc{\neg \varphi}} z_1$, with $\sigma \in \Sigma$ and $z_1 =  \lbrace (\ell_{q},J_{q}, m_q)_{q \in Q} \rbrace \cup \lbrace (r, v^{\star}, m_r) \rbrace$.  Let us note $D = \lbrace (\ell_{q},J_{q})_{q \in Q} \rbrace$ and $z = (r, v^{\star})$. On the first hand, we have that $C \overset{\sigma}{\longrightarrow}_{\appfunc{\neg \varphi}} D$ and so $D \in \appfunc{\neg \varphi}(E)$, where $E = \underset{k \in K}{\bigcup} E_k$ and each $E_k$ is a minimal model of $\delta(\ell_{k},\sigma)$ wrt $I_{k}$.
We recall that $E = \underset{k \in K}{\bigcup} A_{k}[I_{k}]$ where $A_{k}$ is the set of terms of the disjunctive normal form of $\delta(\ell_{k},\sigma)$.  Let $E' = \underset{k \in K}{\bigcup} A_{k}[I'_{k}]$, i.e., we use the same arcs from $C$ to $E$ than from $C'$ to $E'$, for each $D' \in \appfunc{\neg \varphi}(E')$, we have that $C' \overset{\sigma}{\longrightarrow}_{\appfunc{\neg \varphi}} D'$ (because as $s_1 \equiv s_2$, thanks to condition 2. of the definition of "$\equiv$", $I_{k}$ and $I'_{k}$ satisfy the same clock constraints). 
On the other hand, $s_1 \overset{\sigma}{\longrightarrow}_{\appfunc{\neg \varphi}} z_1$ means that $s \overset{\sigma}{\longrightarrow} z$: there exists an arc $(b, \sigma, r, c, R) \in \Delta$ such that $v \models c$ and $\forall x_{p} \in R : v^{\star}(x_{p}) = 0$, while $\forall x_{p} \notin R : v^{\star}(x_{p}) = v(x_{p})$.  
Let us note $z' = (r, v'^{\star})$ where $v'^{\star}$ is such that $\forall x_{p} \in R : v'^{\star}(x_{p}) = 0$, while $\forall x _{p} \notin R : v'^{\star}(x_{p}) = v'(x_{p})$.  We have that $s' \overset{\sigma}{\longrightarrow} z'$ following the arc $(b, \sigma, r, c, R)$: as $s_1 \equiv s_2$, the condition 2. of the definition of "$\equiv$" enables $v(x_{1}), \dots, v(x_{n})$ and $v'(x_{1}), \dots, v'(x_{n})$ to satisfy the same clock constraints (so that $v' \models c$). Let us note $z_3 = \lbrace (\ell,I, m) \vert (\ell,I) \in E \rbrace \cup \lbrace (r, v^{\star}, m_r) \rbrace$ the \footnote{once the minimal models and the arc of $\Bb$ of the definition of $\overset{\sigma}{\rightarrow}$ are chosen, it is easy to see there exists a unique possible choice for the values of the markers of $z_3$.} element of $S$ such that $s \overset{\sigma}{\longrightarrow} z_3$, and $z_4 = \lbrace (\ell,I, m) \vert (\ell,I) \in E' \rbrace \cup \lbrace (r, v'^{\star}, m'_r) \rbrace$ the unique element of $S$ such that $s \overset{\sigma}{\longrightarrow} z_3$.  We will prove that $z_3 \equiv z_4$.  Condition 1. of the definition of $\equiv$ is satisfied because $C$ and $C'$ owned the same markers, as well as $s$ and $s'$ (because $s_1 \equiv s_2$) and we chose the same minimal models to go from $C$ to $E$ than from $C'$ to $E'$ and the same arc of $\Bb$ to go from $s$ to $z$ than from $s'$ to $z'$ (these are the unique parameters for the choice of the markers of $z_3$ and $z_4$).  We still must prove that conditions 2. to 8. are satisfied.  The clock values observed for verifying conditions 2. to 8. are included in $\lbrace v(x_{p}) \vert 1 \leq p \leq n \rbrace \cup \lbrace v'(x_{p}) \vert 1 \leq p \leq n \rbrace \cup \lbrace inf(I_{k}) \vert k \in K \rbrace \cup \lbrace sup(I_{k}) \vert k \in K \rbrace \cup \lbrace inf(I'_{k}) \vert k \in K \rbrace \cup \lbrace sup(I'_{k}) \vert k \in K \rbrace \cup \lbrace 0 \rbrace$ (as the discrete transitions either let the clocks values unchanged or replace them by 0). However, conditions 2. to 8. were verified on $\lbrace v(x_{p}) \vert 1 \leq p \leq n \rbrace \cup \lbrace v'(x_{p}) \vert 1 \leq p \leq n \rbrace \cup \lbrace inf(I_{k}) \vert k \in K \rbrace \cup \lbrace sup(I_{k}) \vert k \in K \rbrace \cup \lbrace inf(I'_{k}) \vert k \in K \rbrace \cup \lbrace sup(I'_{k}) \vert k \in K \rbrace$ thanks to the fact that $s_1 \equiv s_2$. Moreover if a clock value is replaced by 0 in $z$ or $E$, the corresponding clock value in $z'$ or $E'$ is also replaced by 0 (which is the smallest possible clock value, so that its comparisons with all other clocks values will enable to verify 3. to 8.). So, $z_3 \equiv z_4$.\\
Now, as $D \in \appfunc{\neg \varphi}(E)$, we choose $D' \in \appfunc{\neg \varphi}(E')$ such that the intervals of $E'$ grouped to obtain $D'$ correspond to those grouped in $E$ to obtain $D$.  Let us recall that $D = \lbrace(\ell_{q},J_{q})_{q \in Q} \rbrace$ and let us note $D' = \lbrace (\ell_{q},J'_{q})_{q \in Q} \rbrace$
Formally, we want $D'$ to be such that: $\forall(\ell_{k},I_{k}),(\ell_{\tilde{k}},I_{\tilde{k}}) \in C : \big( (\ell_q,J_q) \in \dest(C,D,(\ell_k,I_k)) \wedge (\ell_q,J_q) \in \dest(C,D,(\ell_{\tilde{k}},I_{\tilde{k}})) \big) \Rightarrow \big( (\ell_q,J'_q) \in \dest(C',D',(\ell_k,I'_k)) \wedge (\ell_q,J'_q) \in \dest(C',D',(\ell_{\tilde{k}},I'_{\tilde{k}})) \big)$.
We conclude that $z_1 \equiv z_2$ (the argument is the same as that why $z_3 \equiv z_4$).\\
Secondly, suppose that $s_1 \overset{t}{\rightsquigarrow} z_1$ for a certain $t \in \R^{+}$.  We must prove there exists $t' \in \R^{+}$ and a configuration $z_2$ such that $s_2 \overset{t'}{\rightsquigarrow} z_2$ \textbf{(*)}.  To do that, we first define the "time successor" of an element $s_1$ of $S$, noted $next(s_1)$: it is an element of the first equivalence class of $\equiv$ reachable from the class of $s_1$ (and different from it) letting time elapsing.  Then, we prove that $next(s_1) \equiv next(s_2)$.  We finally deduce \textbf{(*)} from this last result.\\
Let V = $\lbrace v(x_{i}) \vert 1 \leq p \leq n \rbrace \cup \lbrace v  \vert (\ell_{k},I_{k}) \in C \wedge ( v = \inf(I_{k}) \vee v = sup(I_{k}) ) \rbrace$ be the set of clock values present in $s$ and $C$.  We note $\mu = \max \lbrace frac(v) \vert v \in V \rbrace$.  We define $d$ as $\frac{1-\mu}{2}$ if V contains an integer smaller or equal to $c_{max}$, and $1-\mu$ otherwise.  We define the time successor of $s_1$, noted $next(s_1)$, to be $\lbrace (\ell_{k},I_{k}+d, m_k)_{k \in K} \rbrace \cup \lbrace (b, v+d, m_b) \rbrace$.\\
We will prove that: as $s_1 \equiv s_2$, we have that $next(s_1) \equiv next(s_2)$ \textbf{(**)}.  We have that $next(s_1) = \lbrace (\ell_{k},I_{k}+d, m_k)_{k \in K} \rbrace \cup \lbrace (b, v+d, m_b) \rbrace$ for a certain $d > 0$, and $next(s_2) = \lbrace (\ell_{k},I'_{k}+d', m_k)_{k \in K} \rbrace \cup \lbrace (b, v'+d', m_b) \rbrace$ for a certain $d' > 0$.  The effect on $s_1$ is either, if an integer is present among the clocks values, to keep the order of their fractional parts unchanged, either, otherwise, to permute the order of the fractional parts of the clocks values with the largest fractional parts such that they now have a zero fractional part (and so are the clock values with the smallest fractional parts).  The effect on $s_2$ being the same, in the same cases, and the conditions of $s_1 \equiv s_2$ certifying that an integer is present among the clocks values of $s_1$ iff there is an integer between the clocks values of $s_2$, conditions 1. to 8. are still verified on $next(s_1)$ and $next(s_2)$: $next(s_1) \equiv next(s_2)$.\\
Now, as $s_1 \overset{t}{\longrightarrow} z_1$: it means there exists an $n \geq 0$ such that $z_1 \equiv next^{n}(s_1)$.  As $s_1 \equiv s_2$, we deduce from \textbf{(**)} that: $next^{n}(s_1) \equiv next^{n}(s_2)$ and so, taking $z_2 = next^{n}(s_2)$, we verify \textbf{(*)} ($d'$ is the sum of the $n$ $d$'s used to recursively compute $next(s_2), next^{2}(s_2), \dots, next^{n}(s_2)$). \qed
\end{proof}

Remark that the size of the configurations of $\Aa_{\neg \varphi}$ we can encounter in $\Ss_{\Bb,\neg\varphi}$ is bounded by $M(\neg \varphi)$, thanks to the use of $\appfunc{\neg \varphi}$: there is only a finite number of such configurations, and so the number of configurations of $\Aa_{\neg \varphi}$ we can encounter in $\Ss_{\Bb,\neg\varphi}$ is finite.  Therefore, as the number of regions is also finite, the quotient of $\Ss_{\Bb,\neg\varphi}$ by $\equiv$ is finite and we can elaborate a model-checking algorithm using it.\\
Following the approach of \cite{OW07} we will represent each of these regions by a unique word.  Once again, the definition of these words must be adapted.  On one hand, they must maintain an additional component corresponding to markers due to the Miyano Hayashi construction.  On the other hand, each interval must be split in two parts: one representing the infimum of the interval (its region and the relative value of its fractional part) and the other its supremum.  Yet, we must use natural numbers to associate each infimum of interval to the suitable supremum.  We encode
regions of $\Ss_{\Bb, \neg \varphi}$ by finite words whose letters are finite sets of tuples of the form $(\ell, r, m, k)$, where $\ell\in L\cup L^\Bb$, $r\in REG$, $m\in\{\top,\bot\}$ and $0\leq k\leq
M(\varphi)/2$. Here is the definition of the function $H$ that associate to each $s \in S$ the region it is in.\\
For $s = \lbrace (\ell_{k},I_{k}, m_k)_{k \in K} \rbrace \cup \lbrace (b, v, m) \rbrace$, $H(s) = H_1H_2\cdots H_m$ is defined as follows:
\begin{enumerate}
\item For each location $\ell$, let $C(\ell)=\{(\ell',I,m)\in C\mid
  \ell'=\ell\}$. Assume $C(\ell)=\{(\ell_1, I_1, m_1),\ldots,
  (\ell_k,I_k, m_k)\}$, with $I_1\leq\cdots\leq I_k$. Then, we first
  build $E_\ell=\{(\ell_i, \inf(I_i),m_i, i),(\ell_i, \sup(I_i),m_i,
  i)\mid 1\leq i\leq k\}$.
\item We treat $(\ell^\Bb, v, m)$ symmetrically,
  and let $E^\Bb=\lbrace(\ell^\Bb, v(x_{1}), m, 1), \dots,
  \break (\ell^\Bb,
  v(x_{n}), m, n)\rbrace$. We let $\Ee=E^\Bb\cup_{\ell\in L} E_\ell$. That
  is, all elements in $\Ee$ are tuples $(\ell, v,m,i)$, where $\ell$
  is a location (of $\Aa_\varphi$ or $\Bb$), $v$ is a real value
  (interval endpoint or clock value), $m$ is a Miyano-Hayashi marker
  and $i$ is bookkeeping information that links $v$ to an interval (if
  $\ell$ is a location of $\Aa_\varphi$), or to a clock ($\ell$ is a
  location of $\Bb$).
\item We partition $\Ee$ into $\Ee_1,\ldots, \Ee_m$ s.t. each $\Ee_i$
  contains all elements from $\Ee$ with the same fractional part to
  their second component (assuming $frac(u) = 0$ for all
  $u>\cmax$). We assume the ordering $\Ee_1$, $\Ee_2$,\ldots, $\Ee_m$
  reflects the increasing ordering of the fractional parts.
\item For all $1\leq i\leq m$, we obtain $H_i$ from $\Ee_i$ by
  replacing the second component of all elements in $\Ee_i$ by the
  region from $REG$ they belong to.
\end{enumerate}
Thus, noting $\max_{\ell}$ the maximal number of interval that can be present in location $\ell \in L$ of $\Aa_{\neg \varphi}$ (given by Theorem \ref{ThmBorne}), $H(s)$ will be a finite word over the alphabet $\Lambda = (B \cup L) \times REG \times \lbrace \top, \bot\rbrace \times \lbrace 1, 2, \dots, \max(\underset{\ell \in L}{\max}(\max_{\ell}), n) \rbrace$.  We will also view $H$ as a function $H : S \rightarrow \Lambda^{\star}$.

\begin{example}
  Consider a TA $\Bb$ with 1 clock, let $\cmax=2$, and let $s =\break
  \big\{ \{(\ell_{1}, [0,1.3], \bot), (\ell_{1}, [1.8,2.7], \top)\},
  (\ell^\Bb, 1.3, \bot)\big\}$. The first step of the construction
  yields the set $\Ee=\{(\ell_{1}, 0, \bot, 1), (\ell_{1}, 1.3, \bot,
  1), (\ell_{1}, 1.8, \top, 2), (\ell_{1}, 2.7, \top, 2),\break
  (\ell^\Bb,1.3,\bot,1) \}$. Then, we have $H(s)= \big\{ (\ell_{1},
  \{0\}, \bot, 1), (\ell_{1}, (2,+\infty), \top, 2) \big\}\break\big\{
  (\ell_{1}, (1,2), \bot, 1), (\ell^\Bb, (0,1), \bot, 1) \big\}\big\{
  (\ell_{1}, (1,2), \top, 2) \big\}$.
\end{example}

\begin{proposition}
\label{equivH}
Let $s, s' \in S$.  We have: $s \equiv s'$ iff $H(s) = H(s')$.
\end{proposition}
\begin{proof}
($\Rightarrow$)  Suppose that $s \equiv s'$, then the order of the fractional parts of all the clock values that $s$ contains is the same than those of all the corresponding clock values that $s'$ contains (see $\equiv$ conditions 3. to 8.).  Moreover, their corresponding states have their infima (resp. their suprema, resp. clocks values) in the same region (see $\equiv$ condition 2.) and their locations are the same.  The way $H(s)$ and $H(s')$ are constructed, their will be no difference between these two words.\\
($\Leftarrow$)  Suppose $H(s) = H(s')$, then we associate each interval of the configuration of $\Aa_{\neg \varphi}$ $s$ contains, clearly defined by two 4-tuples in $H(s)$, to the interval of $s'$ that is represented by the two corresponding 4-tuples of $H(s')$.  Moreover we associate each value $v_{i}$ of the clock $x_i$ of $\Bb$, clearly defined by a certain 4-tuple in $H(s)$, with the value of $v'_i$ of the clock $x_{i}$ of $\Bb$ represented in the corresponding 4-tuple of $H(s')$.  As $H(s) = H(s')$, conditions 1. and 2. of $\equiv$ are of course verified.  The other conditions are also respected thanks to the groupings executed on the elements of $H(s)$ and $H(s')$ to have an increasing order of the fractional parts of the second components (i.e. clock) values.\qed
\end{proof}

\begin{definition}
\label{DefH}
We define:
\begin{center}
$\mathcal{H} \; = \; \Ss_{\Bb,\neg\varphi} \big/ \equiv \; \; = \; \lbrace H(s) \; \vert \; s \in S \rbrace$.
\end{center}
For all $W^{1}, W^{2} \in \mathcal{H}$ and $\sigma \in \Sigma$ we define $\; W^{1} \overset{\sigma}{\longrightarrow} W^{2} \;$ iff $\; \forall s^{1} \in H^{-1}(W^{1}),$ $\exists s^{2} \in H^{-1}(W^{2}) \; : \; s^{1} \overset{\sigma}{\longrightarrow} s^{2}$.\\
For all $W^{1}, W^{2} \in \mathcal{H}$, we define $W^{1} \longrightarrow_{T} W^{2}$ iff $\; \forall s^{1} \in H^{-1}(W^{1}),$ $\exists t \in \R$ and $\exists s^{2} \in H^{-1}(W^{2}) \; : \; s^{1} \overset{t}{\rightsquigarrow} s^{2}$.
\end{definition}

\begin{proposition}
\label{configChoice}
Let $W^{1}, W^{2} \in \mathcal{H}$, $\sigma \in \Sigma$ and $t \in \R^{+}$.\\
$W^{1} \overset{\sigma}{\longrightarrow} W^{2}$ iff $\exists s^{1} \in H^{-1}(W^{1})$ and $s^{2} \in H^{-1}(W^{2}) : s^{1} \overset{\sigma}{\longrightarrow} s^{2}$.
\end{proposition}
\begin{proof}
($\Rightarrow$) Follows directly from Definition \ref{DefH}.\\
($\Leftarrow$) Suppose that $s^{1} \in H^{-1}(W^{1}), s^{2} \in H^{-1}(W^{2})$, and that $s^{1} \overset{\sigma}{\longrightarrow}$ $s^{2}$.  Let $s^{3} \in H^{-1}(W^{1})$, we must prove that $\exists s^{4} \in H^{-1}(W^{2}) : s^{3} \overset{\sigma}{\longrightarrow} s^{4}$.  As $s^{3} \in H^{-1}(W^{1})$ and $s^{1} \in H^{-1}(W^{1})$, by Proposition \ref{equivH}, $s^{3} \equiv s^{1}$. As $s^{1} \overset{\sigma}{\longrightarrow} s^{2}$, Proposition \ref{Bisim} ensures that $\exists s^{4} \in H^{-1}(W^{2}) : s^{3} \overset{\sigma}{\longrightarrow} s^{4}$.
\end{proof}

\begin{definition}
\label{DefPostH}
Let $\mathcal{W} \subseteq \mathcal{H}$, we define:
\begin{center}
$Post(\mathcal{W}) := \lbrace W' \in \mathcal{H} \; \vert \; \exists \sigma \in \Sigma$, $W \in \mathcal{W}$ and $W'' \in \mathcal{H} : W \longrightarrow_{T} W'' \overset{\sigma}{\longrightarrow} W' \rbrace$.
\end{center}
\end{definition}

We claim that, for any word $W \in \mathcal{H}$, $Post(W)$ is finite and effectively computable.
Let $W \in \mathcal{H}$.  The set of all $W''$ such that $W \longrightarrow_{T} W''$ is a finite set of words with the same number of 4-tuples than W.  We form this set accumulating the words computed recursively as follows, using at each time the last $W_{next}$ obtained (and starting from $W$):
\begin{itemize}
\item if the first letter of $W_{next}$ contains 4-tuples whose second component is $\lbrace 0 \rbrace$ or $\lbrace 1 \rbrace$ or ... or $\lbrace c_{\max} \rbrace$, 
\begin{enumerate}
\item the 4-tupples of this first letter whose second component is $\lbrace c_{\max} \rbrace$ are replaced by the same 4-tuples in which $\lbrace c_{\max} \rbrace$ is replaced by $(c_{\max},+\infty)$,
\item the other 4-tuples of this first letter are deleted from it (if it then becomes empty, it is omitted). A new set of 4-tuples is created as a new second letter: it will contain these same 4-tuples in which the second component is replaced by the immediately following region ($(0,1)$ instead of $\lbrace 0 \rbrace$, $(1,2)$ instead of $\lbrace 1 \rbrace$, ..., $(c_{\max} -1, c_{\max})$ instead of $\lbrace c_{\max} -1 \rbrace$).  The end of the word does not change. 
\end{enumerate}
The following $W_{next}$ is the word created like this ;
\item else, the last letter of $W_{next}$ is deleted.  Its 4-tuples are modified to create a new set that will contain these same 4-tuples in which the second component is replaced by the immediately following region ($\lbrace 1 \rbrace$ instead of $(0,1)$, $\lbrace 2 \rbrace$ instead of $(1,2), \dots , \lbrace c_{\max} \rbrace$ instead of $(c_{\max}-1,c_{\max})$).  This new set is either joined with the first letter of the modified $W_{next}$, if it contains 4-tuples having $(c_{\max}, + \infty)$ as second components, either added as a new first letter of the modified $W_{next}$ otherwise.  The rest of the word does not change. The following $W_{next}$ is the word created like this ;
\end{itemize}
We stop when we encounter a $W_{next}$ that has a unique letter whose 4-tuples have $(c_{\max}, + \infty)$ as second components.
\\
Then, for each possible $W''$ we easily find a $s \in S$ such that $H(s) = W''$ (note that the choice of $s$ does not matter thanks to Proposition \ref{configChoice}).  For all $\sigma \in \Sigma$, it is easy to compute the set of elements $s'$ of $S$ such that $s \overset{\sigma}{\longrightarrow} s'$.  This set is finite and, from each of its elements $s'$, we can get back $H(s')$.\\
Once we have examined each letter $\sigma \in \Sigma$, for each possible $W''$, the (finite !) set of all the $H(s')$ found form $Post(W)$.

\begin{definition}
\label{DefPostPlusH}
Let $\mathcal{W} \subseteq \mathcal{H}$ and $n \in \N_0$, we define: $Post^{n}(\mathcal{W}) = Post^{n-1}(\mathcal{W})$, with $Post^0(\mathcal{W}) = \mathcal{W}$ and $Post^1(\mathcal{W}) = Post(\mathcal{W})$.\\
We define $Post^*(\mathcal{W}) = \underset{n \in \N}{\bigcup} Post^n(\mathcal{W})$ and $Post^+(\mathcal{W}) = \underset{n \in \N_0}{\bigcup} Post^n(\mathcal{W})$.
\end{definition}
Remark that, as $\mathcal{H}$ is finite, $\exists m \in \N : Post^*(\mathcal{W}) = Post^m(\mathcal{W})$.\\

Let us note $H_0 := H(s_0)$ the word of $\mathcal{H}$ corresponding to the initial state of $\Ss_{\Bb,\neg\varphi}$.  We say that a word $W$ of $\mathcal{H}$ is accepting iff the third components of all the 4-tuples it contains are "1" (such words correspond to accepting states of $\Ss_{\Bb,\neg\varphi}$).  We note $\mathcal{F} \subseteq \mathcal{H}$ the set of accepting words of $\mathcal{H}$.\\
Here is an (classical) algorithm for the model-checking of MITL in which the reachable words of $\mathcal{H} = \Ss_{\Bb,\neg\varphi} \big/ \equiv$ are computed "on the fly".\\

\begin{algorithm}
\caption{ModelCheckingMITL $\qquad \qquad \qquad \qquad \qquad \qquad \qquad \qquad \qquad \qquad \qquad$
Input: A TA $\mathcal{B}$ and the ATA $\mathcal{A}_{\neg \varphi}$, for $\varphi \in MITL$. $\qquad \qquad \qquad \qquad \qquad$ \blan{.}
Output: "true" iff $\mathcal{B}$ $\models$ $\varphi$.}
\begin{algorithmic}[1]
\State $C \; \leftarrow \; \emptyset$
\State $D \; \leftarrow \; Post^\star(H_0) \cap \mathcal{F}$

\While{$C \neq D$}
\State $C \; \leftarrow D$
\State $D \; \leftarrow \; Post^{+}(D) \cap \mathcal{F}$
\EndWhile

\If{$D = \emptyset$}
\State \textbf{return} true
\Else
\State \textbf{return} false
\EndIf
\end{algorithmic}
\label{AlgoBase}
\end{algorithm}

We use the following theorem to prove that this algorithm is correct.

\begin{theorem}[\cite{M11} - Theorem 2.3.20]
Let $\Aa = < Q, \Sigma, I, \rightarrow, F >$ be a non-deterministic Büchi automaton.
\begin{center}
$L( \Aa ) = \emptyset \quad$ iff $\quad GFP(\lambda X. Post^{+}(X) \cap F \cap Post^\star(I) ) \; = \; \emptyset$.
\end{center}
\end{theorem}

Remark that, noting $E_0 = Post^\star(H_0) \cap \mathcal{F}$ and $E_i = Post^{+}(E_{i-1}) \cap \mathcal{F}$, at the end of the i-st passage in the while loop of Algorithm \ref{AlgoBase}, $C = E_i$. The following lemma proves the correctness of Algorithm \ref{AlgoBase}.

\begin{lemma}
\label{Lem1}
Let $i \geq 1$. $E_i = E_{i-1}$ iff $E_{i-1} = GFP(\lambda X. Post^{+}(X) \cap \mathcal{F} \cap Post^\star(H_0) )$.
\end{lemma}
\begin{proof}
$(\Leftarrow)$  Suppose that $E_{i-1} = GFP(\lambda X. Post^{+}(X) \cap \mathcal{F} \cap Post^\star(H_0) )$.  In particular, $E_{i-1} = Post^{+}(E_{i-1}) \cap \mathcal{F} \cap Post^\star(H_0)$, i.e.:
\begin{center}
$E_{i-1} = E_{i} \cap Post^\star(H_0)$ $\quad (\ast)$.
\end{center}  
We will show that, $\forall j \geq 1$, $E_j \subseteq E_{j-1}$: in particular it means that $E_i \subseteq E_0 = Post^\star(H_0) \cap \mathcal{F} \subseteq Post^\star(H_0)$.  It enables to conclude from $(\ast)$ that $E_{i-1} = E_{i}$.\\
\underline{Basis:} We prove that $E_1 \subseteq E_0$.  We know that $Post^\star(H_0) \cap \mathcal{F} \subseteq Post^\star(H_0)$.  As function $Post^{+}$ is monotonic: $Post^{+}(Post^\star(H_0) \cap \mathcal{F}) \subseteq Post^{+}(Post^\star(H_0)) \subseteq Post^\star(H_0)$.  So, $E_1 = Post^{+} \left( Post^\star(H_0) \cap \mathcal{F} \right) \cap \mathcal{F} \subseteq Post^\star(H_0) \cap \mathcal{F} = E_0$.\\
\underline{Induction:} Suppose that $\forall 0 < k < i$, $E_k \subseteq E_{k-1}$. We must prove that $E_i \subseteq E_{i-1}$.  By induction hypothesis, $E_{i-1} \subseteq E_{i-2}$, and as function $Post^{+}$ is monotonic: $Post^{+}(E_{i-1}) \subseteq Post^{+}(E_{i-2})$.  Hence, $E_i = Post^{+}(E_{i-1}) \cap \mathcal{F} \subseteq Post^{+}(E_{i-2}) \cap \mathcal{F} = E_{i-1}$.\\
$(\Rightarrow)$ Suppose that $E_i = E_{i-1}$.  Let us first prove that $E_{i-1}$ is a fixed point of $\lambda(X) = Post^{+}(X) \cap \mathcal{F} \cap Post^\star(H_0)$.  $\lambda(E_{i-1}) = Post^{+}(E_{i-1}) \cap \mathcal{F} \cap Post^\star(H_0) = E_i \cap Post^\star(H_0) = E_i = E_{i-1}$.\\
Now, let us prove that $E_{i-1}$ is a greatest fixed point of $\lambda$.  Let $Y \supseteq E_{i-1}$ such that $Y = \lambda(Y)$.  We must prove that $Y = E_{i-1}$.  As we already know that $E_{i-1} \subseteq Y$, it stays to prove that $Y \subseteq E_{i-1}$.  As $Y = \lambda(Y)$, we have $Y = Post^{+}(Y) \cap \mathcal{F} \cap Post^\star(H_0) \subseteq \mathcal{F} \cap Post^\star(H_0) = E_0$.  As $Y \subseteq E_0$ and function $Post^{+}$ is monotonic, $Post^{+}(Y) \subseteq Post^{+}(E_0)$.  So, $Y = Post^{+}(Y) \cap \mathcal{F} \cap Post^\star(H_0) \subseteq Post^{+}(E_0) \cap \mathcal{F} \cap Post^\star(H_0) = E_1 \cap Post^\star(H_0) = E_1$.  Applying inductively the same reasoning, we obtain that, $\forall k \geq 0$, $Y \subseteq E_k$.  Hence, $Y \subseteq E_{i-1}$.\qed
\end{proof}

Here is a theorem giving a rough size of the number of words that Algorithm \ref{AlgoBase} must explore in the worst case.

\begin{theorem}
\label{thmNbConfig}
For every MITL formula $\varphi$, the associated \ATA $\Aa_{\neg \varphi} = (\Sigma, L, \ell_{0}, F, \delta)$, and all timed automaton $\Bb =(\Sigma, B, b_0, X, F^\Bb,
\delta^\Bb)$ with $n$ clocks, $\mathcal{H}$ (constructed thanks to $\Ss_{\Bb,\neg\varphi}$) contains $O\left( 2^{m} \right)$ elements, where $m = \left( \vert B \vert + \vert L \vert \right) . \left( 2.c_{\max}+2 \right) 
\break . 2 . \max \left(\overset{n}{\underset{i=1}{\max}}(\frac{\max_{i}}{2}), n \right) . \left( M(\neg \varphi) + n\right)$.
\end{theorem}
\begin{proof}
Let us note $m' = \left( \vert B \vert + \vert L \vert \right) . \left( 2.c_{\max}+2 \right) . 2 . \max \left(\overset{n}{\underset{i=1}{\max}}(\frac{\max_{i}}{2}), n \right)$. There is $2^{m'}$ elements in $\Lambda = (B \cup L) \times REG \times \lbrace \top, \bot\rbrace \times \lbrace 1, 2, \dots, \max(\overset{n}{\underset{i=1}{\max}}(\frac{\max_{i}}{2}), n) \rbrace$.  $\mathcal{H}$ is a set of words of $\Lambda^\star$ having at most $M(\neg \varphi) + n$ letters.  There are:
\begin{itemize}
\item 1 word of 0 letters on $\Lambda^\star$ ;
\item $2^{m'}$ word of 1 letters on $\Lambda^\star$ ;
\item $2^{2.m'}$ word of 2 letters on $\Lambda^\star$ ;
\item $2^{3.m'}$ word of 3 letters on $\Lambda^\star$ ;
\item ... ;
\item $2^{ \left(M(\neg \varphi) + n \right) . m'}$ word of $M(\neg \varphi) + n$ letters on $\Lambda^\star$ ;
\end{itemize}
Globally, there are $\Sigma_{j = 0}^{M(\neg \varphi) + n} \; 2^{j.m'} = O\left(2^{ \left(M(\neg \varphi) + n \right) . m'}\right) = O\left( 2^{m} \right)$ words in $\mathcal{H}$.
\end{proof}

\paragraph{Zone-based algorithms} In the case of TAs, zones have been
advocated as a data structure which is more efficient in practice than
regions \cite{AD94}. Let us close this section by showing how zones
for \ATA \cite{ADOQW08} can be adapted to represent set of states of
$\Ss_{\Bb, \neg\varphi}$. Intuitively, a zone is a constraint on the
values of the clock copies, with additional information ($loc_\Aa$ and
$loc_\Bb$) to associate clock copies of $\Aa_{\neg\varphi}$ and clocks
of $\Bb$ respectively to locations and Miyano-Hayashi markers.

Let us note $x$ the unique clock of $\Aa_{\neg \varphi}$ and $x^\Bb_1, x^\Bb_2, \dots, x^\Bb_n$ the clocks of $\Bb$. We will note $Copies(x)$ the set of $M(\neg \varphi)$ copies of $x$ denoted $x_1, x_2, \dots, x_{M(\neg \varphi)/2}, y_1, \break
y_2, \dots, y_{M(\neg \varphi)/2}$.  Intuitively, each pair of clock copies $(x_i, y_i)$ will represent an interval.  We also note $Copies_{begin}(x) = \lbrace x_1, x_2, \dots, x_{M(\neg \varphi)/2} \rbrace$, $Copies_{end}(x) = \lbrace y_1, y_2, \dots, y_{M(\neg \varphi)/2} \rbrace$, $Copies^m(x) = \lbrace x_1, x_2, \dots, x_{m}, y_1, y_2, \dots, y_{m} \rbrace$,
$\break Copies^m_{begin}(x) = \lbrace x_1, x_2, \dots, x_{m} \rbrace$ for $1 \leq m \leq M(\neg \varphi)/2$, and $Copies^0(x) = Copies^0_{begin}(x) = \emptyset$.\\
Our definition of zone uses a supplementary clock $x_0$ whose value is always 0.
Zones are also defined thanks to extended clock constraints on a set of clocks $C$, which are of the form $c \bowtie k$ and $c_1 - c_2 \bowtie k$ for $c, c_1, c_2 \in C$, $k \in \N$ and $\bowtie\in \lbrace <, \leq, >, \geq \rbrace$.  When $k = 0$, we will shorten $c_1 - c_2 \bowtie k$ by $c_1 \bowtie c_2$, and we will also use $c_1 = c_2$ as shorthand for $c_1 - c_2 \geq 0 \wedge c_1 - c_2 \leq 0$.

\begin{definition}
A zone $\Zz_m$ is a 3-tuple $(loc_\Aa, loc_\Bb, Z)$ where (1) $loc_\Aa : 
\break Copies^m_{begin}(x) \rightarrow L \times \{ \top, \bot \}$, (2) $loc_\Bb$ is a pair composed of a location of $B$ and a marker in $\{ \top, \bot \}$ and (3) $Z$ is a set of extended clock constraints on $Copies^m(x) \cup \lbrace 
x_{0}, 
x^\Bb_1, x^\Bb_2, \dots, x^\Bb_n \rbrace$,
  interpreted as a conjunction (it is a `classical zone' on this set of clocks
  \cite{AD94}).
\end{definition}
A zone $\Zz_m = (loc_\Aa, loc_\Bb, Z)$, with $loc_\Aa(t) = (loc^t, mark^t)$ and $loc_\Bb = (loc, mark)$, represents the set of states of $\Ss_{\Bb,\neg\varphi}$ of type $\lbrace (loc, vx^\Bb_1, vx^\Bb_2, \dots, vx^\Bb_n, mark), 
\break (loc^{x_1}, [vx_1, vy_1], mark^{x_1}), \dots, (loc^{x_m}, [vx_m, vy_m], mark^{x_m}) \rbrace$, where $vx^\Bb_1, vx^\Bb_2, 
\break \dots, vx^\Bb_n, vx_1, vy_1, \dots, vx_m, vy_m$ are values respectively satisfying the extended clock constraints on $x^\Bb_1, x^\Bb_2, \dots, x^\Bb_n, x_1, y_1, \dots, x_m, y_m$ present in $Z$.  By abuse of terminology, we will sometimes note "$s \in \Zz_m$" to refer to a state of $\Ss_{\Bb,\neg\varphi}$ represented by $\Zz_m$.

\begin{definition}
The initial zone is $\Zz^{init}_1 = (loc_\Aa, loc_\Bb, Z)$ with $Z = \{x^\Bb_1 = 0, \dots x^\Bb_n = 0, x_1 = 0, y_1 = 0 \}$, $loc_\Aa(x_1) = (\ell_0, \bot)$ and $loc_\Bb = (b_0, mark)$ where $mark = \top$ iff $b_0$ is accepting.\\
A zone $Z$ is accepting iff $\forall 1 \leq i \leq m$, $\exists \ell_i \in L$ s.t. $loc_\Aa(x_i) = (\ell_i,\top)$ and $\exists b \in B$ s.t. $loc_\Bb = (b,\top)$.
\end{definition}

\begin{definition}
Let $\Zz_m$ be a zone. $Post_T(\Zz_m)$ denotes the set of configurations that are timed successors of a configuration in $\Zz_m$.
\end{definition}
We can easily compute $Post_T(\Zz_m)$ from $\Zz_m$ deleting all clock constraints of the form $z - x_{0} < k$ or $z - x_{0} \leq k$, for $z \in Copies^m(x) \cup \lbrace x^\Bb_1, x^\Bb_2, \dots, x^\Bb_n \rbrace$, of $\Zz_m$. 

\begin{definition}
Let $\Zz_m$ be a zone. $Post_D(\Zz_m)$ denotes the set of configurations that are discrete successors of a configuration in $\Zz_m$.
\end{definition}
We compute $Post_D(\Zz_m)$, where $\Zz_m = (loc_\Aa, loc_\Bb, Z)$, as follows:
\begin{itemize}
\item for each label $\sigma \in \Sigma$,
\item for each possible transition $t_{\Bb}$ labelled by $\sigma$ of $\Bb$, starting from location $loc_\Bb$,
\item for each possible combination of transitions of $\Aa_{\neg \varphi}$ $(t_{1}, \dots , t_{m})$, all labelled by $\sigma$, such that $\forall 1 \leq i \leq m$, $t_i$ starts from $loc_\Aa(x_i)$\footnote{they associate to each interval $(x_i, y_i)$ a transition $t_i$ to take},
\end{itemize}
we are looking for possible successors as follow.\\
The possible successor zones are found following the arc $t_{\Bb}$ from the location of $loc_\Bb$, $t_{1}$ from the location of $loc_\Aa(x_1)$, ..., and $t_{m}$ from the location of $loc_\Aa(x_{m})$.  First, to take $t_{\Bb}$ implies to satisfy the clock constraint it carries: these clock constraints must be added to those of $Z$ on $\lbrace x^\Bb_1, x^\Bb_2, \dots, x^\Bb_n \rbrace$ (it can be easily done on $Z$ using a well-known algorithm on classical zones \cite{BY03}\todo{réf ok ?}).  Location $loc_\Bb$ must also be modified in the location $t_{\Bb}$ goes to and, potentially, certain clocks among $\lbrace x^\Bb_1, x^\Bb_2, \dots, x^\Bb_n \rbrace$ need to be reset (once again, it can be easily done using a well-known algorithm on classical zones \cite{BY03}\todo{réf ok ?}). Second, to take transitions $t_{i}$ (for $1 \leq i \leq m$) implies that $(x_i, y_i)$ must satisfy the clock constraints it carries: these clock constraints must be added to those of $\Zz_m$ on $Copies(x)$. Moreover, to take transitions $t_{1}, \dots , t_{m}$ can create new clock copies of value 0 in certain locations or/and letting clock copies with de same zone constraints in the same locations.  The new copies with value 0 that have just been created in certain locations can either:
\begin{itemize}
\item be grouped with the previous smallest interval associated to this location

\begin{itemize}
\item[\oran{$\leadsto$}] let us call $\ell$ this location: it corresponds to reset the clock $x_i$ such that, in $\Zz_m$, $loc_\Aa(x_i) = (\ell, mark)$ for a certain $mark \in \{0,1\}$ and $t_i$ loops on $\ell$ for a clock copy $x_i$ with a minimum value.\\
(We can only do this if there were such an interval associated to this location !)
\end{itemize}

\item create a new interval [0,0] associated to this location in the zone.

\begin{itemize}
\item[\oran{$\leadsto$}] let us call $\ell$ this location: it corresponds to use two new unused clock copies of $x$, $x_{m+1}$ and $y_{m+1}$, and extend function $loc$ in way $loc_\Aa(x_{m+1}) = \ell$.  Clocks $x_{m+1}$ and $y_{m+1}$ must be reset.\\
(We can only do this if the new number of intervals associated to the considered location does not exceed the maximal number of intervals we are allowed to associate to this location.)
\end{itemize}

\end{itemize}
Moreover, the markers of $loc_\Aa$ and $loc_\Bb$ must be kept updated.\\

Formally, let us note:
\begin{itemize}
\item $\forall \sigma \in \Sigma : E(c,\sigma) = \lbrace \text{arcs } t \text{ labeled by } \sigma \text{ and s.t. } Start(t) \text{ is the location of } 
\break loc_\Aa(c) \rbrace$,
\item $\forall \sigma \in \Sigma : E(\Bb,\sigma) = \lbrace \text{arcs starting from the location of } loc_\Bb \text{ and labeled } \break
\text{by } \sigma \rbrace$,
\item $\forall \sigma \in \Sigma : Z \odot \sigma = \lbrace (t_{\Bb}, t_{1}, \dots , t_{m}) \: \vert \: t_{\Bb} \in E(\Bb, \sigma) \text{ and } \forall 1 \leq i \leq m : t_{i} \in E(x_{i},\sigma) \rbrace$,
\item for every arc t: $Constr(t) = \lbrace c \: \vert \: c \text{ constraint present on the arc } t \rbrace$,
\item for $t_\Bb = (b_{start}, \sigma, g, r, b_{arrival})$, $g_{t_{\Bb}, t_1, \dots, t_m} = g \quad \wedge \underset{c_i \in Constr(t_{i})}{\underset{1 \leq i \leq m}{\bigwedge}} \left( \left.c_i\right|_{x = x_i} \wedge \left.c_i\right|_{x = y_i} \right)$,
\item for $t_\Bb = (b_{start}, \sigma, g, r, b_{arrival})$, $LM(t_\Bb, t_1, \dots, t_m) = \{ (loc, mark) \; \vert \;
\break ( \exists 1 \leq i \leq m : loc \in Dest(t_i), loc \text{ is not equal to the location of } loc_\Aa(x_i) \text{ and } 
\break mark \text{ is the marker of } loc_\Aa(x_i) ) \text{ or } ( loc = b_{arrival} \text{ and } mark \text{ is the marker}
\break \text{of } loc_\Bb ) \}$,
\item $Loop(t_1, \dots, t_m) = \{ x_i \; \vert \; \text{the location of } loc_\Aa(x_i) \text{ is in } Dest(t_i) \} \cup \{ y_i \; \vert \; \text{the}
\break \text{location of } loc_\Aa(x_i) \text{ is in } Dest(t_i) \}$,
\item $minLoop(t_1, \dots, t_m) = \{ x_i \in Loop(t_1, \dots, t_m) \cap Copies_{begin}(x) \; \vert \; \forall x'_i \in Loop(t_1, \dots, t_m), x'_i \geq x_i \text{ results from the constraints of Z} \}$,
\item $Reset^\ell(t_1, \dots, t_m) = \{ x^\ell, y^\ell \}$ for a certain $x^\ell \in Copies_{begin}(x) \setminus Loop(t_1, \dots, t_m)$ and $y^\ell \in Copies_{end}(x) \setminus Loop(t_1, \dots, t_m) \}$ if these sets are not empty and $\exists mark \in \{0,1\} \text{ s.t. } (\ell, mark) \in LM(t_\Bb,t_1, \dots, t_m)$ ; otherwise, 
$\break Reset^\ell(t_1, \dots, t_m) = \emptyset$,
\item $R^\ell = \{ Reset^\ell(t_1, \dots, t_m) \} \cup \{ \{x_i\} \; \vert \; Reset^\ell(t_1, \dots, t_m) \neq \emptyset, 
\break x_i \in minLoop(t_1, \dots, t_m) \text{ and } loc_\Aa(x_i) = \ell \}$.
\end{itemize}

$\Zz'_{m'} = (loc'_\Aa, loc'_\Bb, Z') \in Post_D(\Zz_m) \;$ iff (1) $\Zz_m$ is not final, $g_{t_{\Bb}, t_1, \dots, t_m} \cap Z$ is satisfiable and $\; \exists \sigma \in \Sigma, (t_\Bb,t_1, \dots, t_m) \in Z \odot \sigma$ and $(r_{\ell_1}, \dots, r_{\ell_p})$, with $\forall 1 \leq i \leq p, r_{\ell_i} \in R_{\ell_i}$, such that :
\begin{itemize}
\item $Z' = [\left( r \cup \bigcup_{j = 1}^p r_{\ell_j} \right) := 0] \left( g_{t_{\Bb}, t_1, \dots, t_m} \cap Z \right)$,
\item $\forall x_i \in Loop(t_1, \dots, t_m)$, $loc'_\Aa(x_i)$ is composed of the location of $loc_\Aa(x_i)$ and of:
\begin{itemize}
\item[$\bullet$] marker "1" if the location of $loc_\Aa(x_i)$ is in $F$,
\item[$\bullet$] the marker of $loc_\Aa(x_i)$ otherwise ;
\end{itemize}
$\forall x_\ell \in Reset^\ell(t_1, \dots, t_m) \cap Copies_{begin}(x)$, $loc'_\Aa(x_\ell)$ is composed of the location $\ell$ and of marker:
\begin{itemize}
\item[$\bullet$] "1" if $\ell \in F$ or $(\ell,0) \notin LM(t_\Bb, t_1, \dots, t_m)$,
\item[$\bullet$] "0" otherwise ;
\end{itemize}
\item $loc'_\Bb$ is composed of $Dest(t_\Bb)$ and of marker:
\begin{itemize}
\item[$\bullet$] "1" if $Dest(t_\Bb) \in F^\Bb$ or $(Dest(t_\Bb),0) \notin LM(t_\Bb, t_1, \dots, t_m)$,
\item[$\bullet$] "0" otherwise ;
\end{itemize}
\end{itemize}
or (2) $\Zz_m$ is final and, defining $loc^\star_\Aa$ such that $\forall 1 \leq i \leq m$, if $loc_\Aa(x_i) = (\ell, 1)$, then $loc^\star_\Aa (x_i) = (\ell, 0)$, $\Zz'_{m'}$ satisfies the conditions of case (1) in which $loc^\star_\Aa$ plays the role of $loc_\Aa$.

\begin{proposition}
Let $\Zz_m$ be a zone.  $Post_D(\Zz_m) = \{ s' \; \vert \; \exists s \in \Zz_m \text{ s.t. } s \rightarrow s' \text{ in } \Ss_{\Bb,\neg\varphi} \}$.
\end{proposition}
\begin{proof}
$(\subseteq)$ Suppose that $\Zz_m = (loc_\Aa, loc_\Bb, Z)$, where $loc_\Bb = (b, mark_\Bb)$ and $\forall1 \leq k \leq m$, $loc_\Aa(x_k) = (\ell_k, mark_k)$. Let $s' \in Post_D(\Zz)$.  There exists a certain $s' \in \Zz'_{m'}$ for a certain $\Zz'_{m'} \in Post_D$ constructed thanks to $\sigma, t_{\Bb}, t_{1}, \dots , t_{m}, r_{\ell_1}, \dots,
\break r_{\ell_p}$, we suppose that $t_\Bb = (b_{start}, \sigma, g, r, b_{arrival})$.  Let us suppose that $s' = 
\break \lbrace (\ell'_{k'}, I'_{k'}, mark'_{k'})_{k' =1}^{m'} \rbrace \cup \lbrace (b', v', mark'_\Bb) \rbrace$, with ($\forall 1 \leq k' \leq m'$) $I'_{k'} = [v'x_{k'},v'y_{k'}]$.  We construct $s = \lbrace (\ell_{k}, I_{k}, mark_{k})_{k =1}^{m} \rbrace \cup \lbrace (b, v, mark_\Bb) \rbrace$, with ($\forall 1 \leq k \leq m$) $I_{k} = [vx_{k},vy_{k}]$, where:
\begin{itemize}
\item[1.] \underline{Arc of $\Bb$ without reset:} $\forall 1 \leq i \leq n$: if $x^\Bb_i \notin r$, we define $v(x_i^\Bb) = v'(x_i^\Bb)$,
\item[2.] \underline{New complete interval:} $\forall 1 \leq k \leq m$: if there exists $1 \leq j \leq p$ such that $x_k \in r_{\ell_j} \cap Reset^{\ell_j}(t_1, \dots, t_m)$, then $x^{\ell_k}$ and $y^{\ell_k}$ was not used in $\Zz$ and their values must not be defined,
\item[3.] \underline{Loop without merge:} $\forall 1 \leq k \leq m$: if $\forall 1 \leq j \leq p$, $x_k \notin r_{\ell_j}$ but that $loc'_\Aa$ is defined on $x_k$, we define $vx_k = v'x_k$ and $vy_k = v'y_k$,
\item[4.] \underline{Loop with merge:} $\forall 1 \leq k \leq m'$: if there exists $1 \leq j \leq p$ such that $x_k \in r_{\ell_j} \setminus Reset^{\ell_j}(t_1, \dots, t_m)$, then we define $vy_k = v'y_k$,
\item[5.] \underline{Arc going out or arc of $\Bb$ with reset:} the values $vx_k$, $vy_k$ and $v(x^\Bb_i)$ that we still must define are arbitrarily chosen in way they satisfy the extended clock constraints of $g_{t_\Bb, t_1, \dots, t_m} \cap Z$ (which is possible because $g_{t_\Bb, t_1, \dots, t_m} \cap Z$ is satisfiable and we only chose values of clock/clock copies that have the same in $\Zz'_{m'}$, so that they can not prevent $g_{t_\Bb, t_1, \dots, t_m} \cap Z$ from being satisfiable).
\end{itemize}
We suppose that $\Zz_m$ is not accepting, so that $s$ is not accepting, the proof is very similar in the other case.  We must prove that $s \overset{\sigma}{\rightarrow} s'$ in $\Ss_{\Bb,\neg\varphi}$ (see conditions (i)-(a),(b),(c) and (d) of the definition of $\rightarrow$ in $\Ss_{\Bb,\neg\varphi}$).  Condition (d) (resp. (b) and (c)) is (are) trivially verified by definition of the function $loc'_\Bb$ (resp. $loc'_\Aa$) of $\Zz'_{m'}$.  Its stays to prove (a), i.e.: $(b,v) \overset{\sigma}{\rightarrow} (b',v')$ in $\Bb$ and $\{ (\ell_{k}, I_{k})_{k =1}^{m} \} \overset{\sigma}{\rightarrow}_{\appfunc{\neg \varphi}} \{ (\ell'_{k'}, I'_{k'})_{k' =1}^{m'} \}$ in $\Aa_{\neg \varphi, \appfunc{\neg \varphi}}$.\\
We first show that $(b,v) \overset{\sigma}{\rightarrow} (b',v')$ in $\Bb$ thanks to $t_\Bb$.  Indeed, $v \models g$ because $g$ is contained in $g_{t_\Bb, t_1, \dots, t_m} \cap Z$ (satisfied) ; $\forall x \in r$, $v'(x) = 0$ because it is reset by construction of $\Zz'_{m'}$ and $\forall x \in \{ x_1^\Bb, \dots, x_2^\Bb \} \setminus r$, $v'(x) = v(x)$ by 1..\\
Now, let us show that $\{ (\ell_{k}, I_{k})_{k =1}^{m} \} \overset{\sigma}{\rightarrow}_{\appfunc{\neg \varphi}} \{ (\ell'_{k'}, I'_{k'})_{k' =1}^{m'} \}$ in $\Aa_{\neg \varphi, \appfunc{\neg \varphi}}$.  We must prove there exists minimal models $M_k$, for $1 \leq k \leq m$ of $\delta(\ell_k, \sigma)$ wrt $I_k$ such that $\{ (\ell'_{k'}, I'_{k'})_{k' =1}^{m'} \} \in \appfunc{\neg \varphi} \left( \bigcup_{k=1}^{m} M_k \right)$.  For each $1 \leq k \leq m$, we take $M_k$ to be the minimal model of $\delta(\ell_k, \sigma)$ wrt $I_k$ obtained following the arc $t_k$: it can be taken because its constraints are verified on $vx_k$ and $vy_k$ (present in $g_{t_\Bb, t_1, \dots, t_m} \cap Z$) and as this constraint can only be an interval (convex), it is also verified on each $i \in I_k$.  We then take the element $E$ of $\appfunc{\neg \varphi}\left( \bigcup_{k=1}^{m} M_k \right)$ that merges the two more little intervals present in $\ell_j$ (for $1 \leq j \leq p$) iff $r_{\ell_j}$ is a singleton.  It stays to prove that the obtained configuration is exactly $\{ (\ell'_{k'}, I'_{k'})_{k' =1}^{m'} \}$.  It is based on the following facts:
\begin{itemize}
\item[a.] each $(\ell_k, I_k)$ that does not loop disappear from $\{ (\ell_{k}, I_{k})_{k =1}^{m} \}$ to $E$ and is not present in $\{ (\ell'_{k'}, I'_{k'})_{k' =1}^{m'} \}$ : the clock copies representing such intervals disappear from $\Zz_m$ to $\Zz'_{m'}$ (i.e. $loc_\Aa$ is not defined on them anymore),
\item[b.] for each location $\ell_j$ ($1 \leq j \leq p$) destination of at least one $t_k$ ($1\leq k \leq m$), $(\ell_j, [0,0])$ is present in $\bigcup_{k=1}^{m} M_k$ and
\begin{itemize}
\item \underline{if $r_{\ell_j} = Reset^{\ell_j}(t_1, \dots, t_m)$:} two new clock copies, say $x_{\ell_j}$ and $y_{\ell_j}$ are used by $\Zz'_{m'}$.  They are defined and reset in way $(\ell_j, [0,0])$ is also present in $\{ (\ell'_{k'}, I'_{k'})_{k' =1}^{m'} \}$ ;
\item \underline{else:} by definition of $r_{\ell_j}$, $R^{\ell_j}$ and $minLoop(t_1, \dots, t_m)$, the clock copy $x_j$ representing the beginning of the more little interval in $\{ (\ell_{k}, I_{k})_{k =1}^{m} \}$ that loops on $\ell_j$, say $I_j = [vx_j, vy_j]$, is reset in $\Zz'_{m'}$: $\{ (\ell'_{k'}, I'_{k'})_{k' =1}^{m'} \}$ contains $(\ell_j, [0,vy_j])$ ;
\end{itemize}
\item[c.] each $(\ell_k, I_k)$ that loops is still present in $\bigcup_{k=1}^{m} M_k$ : the clock copies representing $I_k$ in $\Zz_m$ are still present in $\Zz'_{m'}$ but the clock copy representing its beginning could have been reset, so that $\{ (\ell'_{k'}, I'_{k'})_{k' =1}^{m'} \}$ contains either $(\ell_k, [vx_k,vy_k])$ or $(\ell_k, [0,vy_k])$,
\item[d.] when computing $E$ from $\bigcup_{k=1}^{m} M_k$, we know the two more little intervals present in $\ell_j$ (for $1 \leq j \leq p$) are merged iff $r_{\ell_j}$ is non empty and equal to $Reset^{\ell_j}(t_1, \dots, t_m)$.  So, $(\ell_j,[0,0])$ and $(\ell_j,I_j)$ are merged in $(\ell_j,[0,vy_j])$ iff $I_j$ is the more little interval that loops on $\ell_j$, $r_{\ell_j}$ is a singleton and so must contain the clock copy representing its beginning : $x_j$.  In this case and only in this case, $x_j$ is then reset constructing $\Zz'_{m'}$ so that $E$ contains $(\ell_j,[0,vy_j])$ iff $\{ (\ell'_{k'}, I'_{k'})_{k' =1}^{m'} \}$ also contains $(\ell_j, [0, vy_j])$.
\end{itemize}

$(\supseteq)$ Let $\Zz_m = (loc_\Aa, loc_\Bb, Z)$ be a zone and $s'$ be such that $\exists s \in \Zz_m$ s.t. $s \rightarrow s'$ in $\Ss_{\Bb,\neg\varphi}$.  Let us show that $s' \in Post_D(\Zz_m)$.  Let us suppose that $s = \lbrace (\ell_{k}, I_{k}, mark_{k})_{k =1}^{m} \rbrace \cup \lbrace (b, v, mark_\Bb) \rbrace$, with ($\forall 1 \leq k \leq m$) $I_{k} = [vx_{k},vy_{k}]$ and $s' = \lbrace (\ell'_{k'}, I'_{k'}, mark'_{k'})_{k' =1}^{m'} \rbrace \cup \lbrace (b', v', mark'_\Bb) \rbrace$, with ($\forall 1 \leq k' \leq m'$) $I'_{k'} = [v'x_{k'},v'y_{k'}]$.  As $s \rightarrow s'$, there exists $\sigma \in \Sigma$ such that :
\begin{itemize}
\item $(b,v) \overset{\sigma}{\rightarrow} (b',v')$ in $\Bb$, i.e.: there exists an arc $t_\Bb = (b, \sigma, g, r, b')$ s.t. $v \models g$, $\forall x \in r$, $v'(x) = 0$ and $\forall x \in \{ x_1^\Bb, \dots, x_n^\Bb \} \setminus r$, $v'(x) = v(x)$;
\item $\{ (\ell_{k}, I_{k})_{k =1}^{m} \} \overset{\sigma}{\rightarrow}_{\appfunc{\neg \varphi}} \{ (\ell'_{k'}, I'_{k'})_{k' =1}^{m'} \}$ in $\Aa_{\neg \varphi, \appfunc{\neg \varphi}}$, i.e.: $\{ (\ell'_{k'}, I'_{k'})_{k' =1}^{m'} \} = E \in \appfunc{\neg <Phi}(\bigcup_{k=1}^{m} M_k)$ for certain minimal models $M_k$ of $\delta(\ell_k,\sigma)$ wrt $I_k$, which are themselves obtained by taking certain arcs $t_k$ (for $1\leq k \leq m$) from $\ell_k$.
\end{itemize}
Let us define $r_{\ell_j}$, for $1 \leq j \leq p$ by:
\begin{itemize}
\item $r_{\ell_j} = Reset^{\ell_j}(t_1, \dots, t_m)$ contains two new clock copies \emph{iff} $\exists 1\leq k\leq m$ s.t. $t_k$ goes to $\ell_j$ with a reset and no merge is applied by $\appfunc{\neg \varphi}$ on $\ell_j$,
\item $r_{\ell_j}$ contains the clock representing the beginning of the more little interval present in $\ell_j$ that loops on $\ell_j$ taking one of the $t_k$ (for $1\leq k \leq m$) \emph{iff} $\exists 1\leq k\leq m$ s.t. $t_k$ goes to $\ell_j$ with a reset and a merge is applied by $\appfunc{\neg \varphi}$ on $\ell_j$,
\item $r_{\ell_j} = \emptyset$ in the other case, i.e. when none of the $t_k$ (for $1\leq k \leq m$) goes to $\ell_j$ with a reset.
\end{itemize}
It is easy to prove that the $\Zz'_{m'} \in Post_D(\Zz_m)$ induced by $\sigma, t_\Bb, t_1, \dots, t_m, r_{\ell_1}, \dots, r_{\ell_p}$ contains $s'$ thanks to the following facts:
\begin{itemize}
\item as $s \in \Zz_m$, the constraints of $Z$ are satisfied by its clock values and having take arcs $t_\Bb, t_1, \dots, t_m$ ensure, the bounds of the intervals and the clock values of $s$ satisfy $g_{t_\Bb, t_1, \dots, t_m} \cap Z$ (in particular $g_{t_\Bb, t_1, \dots, t_m} \cap Z$ is satisfiable and $\Zz'_{m'}$ really exists),
\item a clock $x_i^\Bb$ (for $1 \leq i \leq n$) is reset in the construction of $\Zz'_{m'}$ iff $v'(x_i^\Bb) = 0$,
\item the facts b., c. and d. of the proof of inclusion $\subseteq$ are still true here.
\item the markers are treated in a similar manner from $\Zz_m$ to $\Zz'_{m'}$ as from $s$ to $s'$, observing the definition of the $\Zz'_{m'} \in Post_D(\Zz_m)$ induced by $\sigma, t_\Bb, t_1, \dots, t_m,
\break r_{\ell_1}, \dots, r_{\ell_p}$ and conditions (i)-(a),(b),(c),(d) and (ii) of the definition of $\rightarrow$ in $\Ss_{\Bb,\neg\varphi}$.
\end{itemize}\qed
\end{proof}


Algorithm \ref{AlgoBase} is still correct replacing $H_0$ by $\Zz^{init}_1$, the function $Post(\mathcal{W})$ by $Approx_\beta( Post(\mathcal{Z}) )$\footnote{see \cite{B02} for details on the definition of $Approx_\beta$} and if $\mathcal{F}$ designates the set of accepting zones.  The proof relies on a very simple adaptation of Theorem 2 of \cite{B02} to our definition of zone.

\paragraph{Eliminating useless clock copies} In many practical
examples, MITL formulas contain modalities of the form
$U_{[0,+\infty)}$ or $\tilde{U}_{[0,+\infty)}$ that do not impose any
real-time constraints (in some sense, they are LTL modalities). For
instance, consider the $\qed$ modality in $\varphi=\qed ( a
\Rightarrow \lozenge_{[1,2]} b)$. When this occurs in a formula
$\varphi$, we can simplify the representation of configurations of
$\Aa_\varphi$, by dropping the values of the clocks associated to
those modalities (these clocks can be regarded as \emph{inactive} in
the sense of \cite{DY96}).  We call those configurations \emph{reduced
  configurations}. In the example of Fig.~\ref{ExGilles}, this amounts
to skipping the clocks associated to $\ell_{\Box}$, and the
configuration $\{(\ell_{\Box},0.1)(\ell_\lozenge,0)\}$ of $\Aa_\varphi$
in Fig.~\ref{ExGilles} can be represented by a pair
$(\{\ell_{\Box}\},\{(\ell_\lozenge,0)\})$. As we will see in the next
section, maintaining reduced configurations, when possible, usually
improves the performance of the algorithms.\\
One can also verify that the values of the clock copies present in the initial location $(\neg \varphi)_{init}$ of $\Aa_{\neg \varphi}$ are not relevant. Let us note $\SubLTL{\neg \varphi} = \lbrace \varphi \in Sub(\neg \varphi) \; \vert \; \varphi = \varphi_{init} \text{ or the outermost operator of } \varphi$ is $U_{[0,+\infty)}$ or $\tilde{U}_{[0,+\infty)}\rbrace$.  Then, the states of $\Aa$ can be a couple $(\ell,I)$ where $\ell \in L \setminus \SubLTL{\neg \varphi}$ and $I \in \mathcal{I}(\R^{+})$ as well as a singleton $\ell \in \SubLTL{\neg \varphi}$.  A reduced configuration is then a pair $(S,C)$ where $S \subseteq \SubLTL{\neg \varphi}$ and $C$ is a configuration of $\mathcal{A}_{\neg \varphi}$ such that $\forall (\ell,v) \in C$, $\ell \notin \SubLTL{\neg \varphi}$.  It is clear that all the previous results still hold on reduced configurations and we also implemented Algorithm \ref{AlgoBase} with them : the region or zone abstraction is only used on $L \setminus \SubLTL{\neg \varphi\varphi}$.

\section{Experimental results\label{sec:experimental-results}}

\begin{table}[t]
  \caption{Benchmark for satisfiability (top) and model-checking (bottom). Reported values are execution time in ms / number of visited states.}
\begin{center}
{\scriptsize
 \begin{tabular}{|c|c|c|c|c|c|c|}
   \hline
   Sat ? & Formula & Size & Regions & Reduced regions & Zones &  Reduced zones \\
  \hline
  Sat & $E(5,[0,+\infty))$ & 5 & 74 / 61  & 16 / 31 & 58 / 36 & 39 / 31\\
  Sat & $E(10,[0,+\infty))$ & 10 & 3296 / 2045  & 369 / 1023 & 1374 / 1033 & 2515 / 1023\\
  Sat & $E(5,[5,8))$ & 5 & 382 / 228 & 394 / 228 & 83 / 33 & 86 / 33\\
  Sat & $E(10,[5,8))$ & 10 & 70129 / 7172 & 79889 / 7172 & 1982 / 1025 & 2490 / 1025\\\hline
  Sat & $A(10,[0,+\infty))$ & 10 & 1 / 1  & 1 / 1 & 4 / 1 & 5 / 1\\
  Sat & $A(10,[5,8))$ & 10 & 1926 / 7 & 2036 / 5 & 3036 / 2 & 3153 / 2\\
  Sat & $U(10,[0,+\infty))$ & 9 & 231 / 7  & 5 / 4 & 16 / 1 & 6 / 1\\
  Unsat & $U(2,[5,8])$ & 2 & 13 / 6 & 15 / 8 & 4 / 2 & 4 / 2\\
  Unsat & $U(3,[5,8])$ & 3 & OOM & OOM & OOM & OOM\\\hline
  Sat & $T(10,[0,+\infty[)$ & 9 & $>5$min  & 3 / 2 & 33 / 3 & 7 / 2\\
  Sat & $T(10,[5,8))$ & 9 & 52 / 2  & 40 / 2 & 11 / 2 & 11 / 2\\\hline
  Sat & $R(5,[0,+\infty))$ & 20 & $>5$min  & 301 / 270 & 4307 / 1321 & 145 / 81 \\
  Sat & $R(10,[0,+\infty))$ & 40 &  $>5$min  & OOM & OOM & $>5$min \\
  Sat & $R(5,[5,8))$ & 20 & OOM  & 6996 / 117 & 1299 / 36  & 1518 / 36\\
  Sat & $R(10,[5,8))$ & 40 &  $>5$min  & $>5$min & $>5$min & $>5$min\\\hline
  Sat & $Q(5,[0,+\infty))$ & 10 & 44 / 39  & 11 / 20 & 43 / 29 & 22 / 20\\
  Sat & $Q(10,[0,+\infty))$ & 20 &  1209 / 1041  & 286 / 521 & 841 / 540 & 933 / 521\\
  Sat & $Q(5,[5,8))$ & 10 & 497 / 98  & 378 / 57 & 167 / 32  & 181 / 32 \\
  Sat & $Q(10,[5,8))$ & 20 & 35776 / 2646  & 20324 / 2912 & 81774 / 782 & 86228 / 782 \\
  \hline
\end{tabular}

\medskip
\begin{tabular}{|c|c|c|c|c|c|c|c|}
  
  \hline
  Floors    &  Formula & Form./ TA size  & OK ? & Regions & Zones & Red. zones\\
  \hline
  2 & $\Box \bigwedge_{i = 1,2} \left( o_i \Rightarrow \lozenge_{]1,2]} c_i \right)$ & 3 / 10 & $\times$ & 166 / 89 & 57 / 35 & 56 / 32 \\
  2 & $\Box \bigwedge_{i = 1,2} \left( b_i \Rightarrow \lozenge_{[0,4]} o_i \right)$ & 3 / 10 & $\checkmark$ & 302 / 225 & 31 / 31 & 26 / 25 \\
  2 & $\Box \bigwedge_{i = 1,2} \left( l_i \Rightarrow \lozenge_{[0,6]} o_i \right)$ & 3 / 10 & $\checkmark$ & 820 / 690 & 68 / 51 & 60 / 40 \\
  
  3 & $\Box \bigwedge_{i = 1,\dots,3} \left( o_i \Rightarrow \lozenge_{]1,2]} c_i \right)$ & 4 / 37 & $\times$ & 681 / 480 & 463 / 154 & 337 / 140 \\
  3 & $\Box \bigwedge_{i = 1,\dots,3} \left( b_i \Rightarrow \lozenge_{[0,12]} o_i \right)$ & 4 / 37 & $\checkmark$ & $>$5min & 1148 / 541 & 1008 / 376 \\
  3 & $\Box \bigwedge_{i = 1,\dots,3} \left( l_i \Rightarrow \lozenge_{[0,14]} o_i \right)$ & 4 / 37 & $\checkmark$ & $>$5min & 1321 / 774 & 1387 / 540 \\
  
  4 & $\Box \bigwedge_{i = 1,\dots,4} \left( o_i \Rightarrow \lozenge_{]1,2]} c_i \right)$ & 5 / 114 & $\times$ & 5570 / 1638 & 1381 / 498 & 1565 / 461 \\
  4 & $\Box \bigwedge_{i = 1,\dots,4} \left( b_i \Rightarrow \lozenge_{[0,20]} o_i \right)$ & 5 / 114 & $\checkmark$ & $>$5min & 26146 / 5757 & 22776 / 3156 \\
  4 & $\Box \bigwedge_{i = 1,\dots,4} \left( l_i \Rightarrow \lozenge_{[0,22]} o_i \right)$ & 5 / 114 & $\checkmark$ & $>$5min & 52167 / 7577 & 48754 / 4337 \\
  
  5 & $\Box \bigwedge_{i = 1,\dots,5} \left( o_i \Rightarrow \lozenge_{]1,2]} c_i \right)$ & 6 / 311 & $\times$ & 61937 / 4692 & 3216 / 1402 & 3838 / 1310 \\
  5 & $\Box \bigwedge_{i = 1,\dots,5} \left( b_i \Rightarrow \lozenge_{[0,28]} o_i \right)$ & 6 / 311 & $\checkmark$ & $>$5min & $>$5min & OOM \\
  5 & $\Box \bigwedge_{i = 1,\dots,5} \left( l_i \Rightarrow \lozenge_{[0,30]} o_i \right)$ & 6 / 311 & $\checkmark$ & OOM & $>$5min & $>$5min \\
  \hline
\end{tabular}
}
\end{center}
\label{Table1}\label{Table2}
\vspace*{-.6cm}
\end{table}

To evaluate the practical feasibility of our approach, we have
implemented the region and zone-based algorithms for model-checking
and satisfiability in a prototype tool. To the best of our knowledge,
this is the first implementation to perform MITL model-checking and
satisfiability using an automata-based approach.  We first consider a
benchmark for the \emph{satisfiability problem}, adapted from the
literature on LTL \cite{GKLMR10} and consisting of six parametric
formulas (with $k \in \N$ and $I \in \mathcal{I}(\ninf)$):
%
$$
{\small
  \begin{array}{ll}
  E(k,I) = \bigwedge_{i = 1, \dots, k} \lozenge_{I} \; p_i & U(k,I) = ( \dots (p_1 U_{I} p_2) U_{I} \dots ) U_{I}
  p_k\\ 
  A(k,I) = \bigwedge_{i = 1, \dots, k} \Box_{I} \; p_i & T(k,I) = p_1 \tilde{U}_{I} (p_2 \tilde{U}_{I} ( p_3
  \dots p_{k-1} \tilde{U}_{I} p_k) \dots )\\ 
Q(k,I) = \bigwedge_{i =
    1, \dots, k} \left( \lozenge_{I} p_i \vee \; \Box_{I} p_{i+1}
  \right) &
 R(k,I) = \bigwedge_{i =
    1, \dots, k} \left( \Box_{I} \;( \lozenge_{I} p_i) \vee \;
  \lozenge_{I} ( \Box_{I} p_{i+1} ) \right)\\
  \end{array}
}
$$
Table~\ref{Table1} (top) reports on the running time, number of
visited regions and returned answer (column `Sat ?')  of the prototype
on several instances of those formulas, for the different data
structures. A time out was set after 5 minutes, and OOM stands for `out
of memory'.
Our second benchmark evaluates the performance of our
\emph{model-checking tool}. We consider a family of timed automata
$\Bb^{lift}_k$ that model a \emph{lift}, parametrised by the number
$k$ of floors. A button can be pushed at each floor to call it. 
The alphabet contains a letter $l_i$ for each floor $i$ saying that lift has been
called at this floor.  A button to send the lift at each floor is
present in it: the alphabet contains a letter $b_i$ for each floor $i$
saying that button $i$ has just been pushed. The lift takes 1 second
to go from a floor to the next one and to open/close its doors : it
stays 1 second open between these opening/closure.  Letter $o_i$
(resp. $c_i$) signifies the lift open (resp. close) its doors at floor
$i$ ; letter $p_i$ means the lift pass floor $i$ without stopping. The
lift goes up (resp. down) as long as it is called upper (resp. lower).
When it is not called anywhere, it goes to the medium floor and stays
opened there.\\
Fig.~\ref{FigLift} gives a representation of
$\Bb^{lift}_2$ as example.  We represent two similar edges (same starting and
ending locations, same reset) carrying two different letters by a
unique edge carrying these two letters.  A location of this automaton
is a 4-tuple $(n, direction, go, open?)$ where $n$ is the number of
the floor the cabin is present in ($0 \leq n < 2$, 0 representing the
ground floor), $direction \in \{u,d,h\}$ is $u$ if the lift is going
up, $d$ if it is going down and $h$ if it stays at the floor it is
present in, $go$ is the set of floors to which the cabin must go
(because the lift has been called at this floor or because the button
present in the cabin has been pushed to send the lift at this floor),
and $open? \in \{ \top, \bot\}$ is $\top$ iff the doors of the cabin
are opened.

\begin{figure}[t]
\centering
\begin{tikzpicture}[scale=.9]

\begin{scope}
\draw (0,0) node [rectangle,rounded corners, double, draw, inner sep=3pt] (A)
      {$\left(0,h,\emptyset,\top\right)$};

\draw (4,0) node [rectangle,rounded corners, double, draw, inner sep=3pt] (B)
      {$\left(0,u,\{1\},\top\right)$};

\draw (-2,2) node [rectangle,rounded corners, double, draw, inner sep=3pt] (C)
      {$\left(0,u,\{1\},\bot\right)$};

\draw (6,2) node [rectangle,rounded corners, double, draw, inner sep=3pt] (D)
      {$\left(0,u,\{0,1\},\bot\right)$};

\draw (0,-2) node [rectangle,rounded corners, double, draw, inner sep=3pt] (E)
      {$\left(1,d,\emptyset,\bot\right)$};

\draw (4,-2) node [rectangle,rounded corners, double, draw, inner sep=3pt] (F)
      {$\left(1,d,\{0\},\bot\right)$};

\draw (-.5,-4) node [rectangle,rounded corners, double, draw, inner sep=3pt] (G)
      {$\left(1,d,\{1\},\bot\right)$};

\draw (4.5,-4) node [rectangle,rounded corners, double, draw, inner sep=3pt] (H)
      {$\left(1,d,\{0,1\},\bot\right)$};

\draw (-1,-6) node [rectangle,rounded corners, double, draw, inner sep=3pt] (I)
      {$\left(1,d,\emptyset,\top\right)$};

\draw (5,-6) node [rectangle,rounded corners, double, draw, inner sep=3pt] (J)
      {$\left(1,d,\{0\},\top\right)$};

\draw [-latex'] (-2,0)--(A);

\draw [-latex'] (A) -- (B) node [pos=0.5,above] {$l_1; \ b_1$} node
      [pos=0.5,below] {$x:=0$};

\draw [-latex'] (B) -- (C) node [sloped, pos=0.5,above] {$c_0,x=2$} node
      [sloped,pos=0.5,below] {$x:=0$};

\draw [-latex'] (C) -- (D) node [pos=0.5,above] {$l_0$};

\draw [-latex'] (C) .. controls +(180:3cm) and +(180:3cm) .. (I) node
      [pos=0.2,left] {$o_1,x=2$} node [pos=0.3,left] {$x:=0$};

\draw [-latex'] (D) .. controls +(0:3cm) and +(0:3cm) .. (J) node
      [pos=0.2,right] {$o_1,x=2$} node [pos=0.3,right] {$x:=0$};

\draw [-latex'] (E) -- (A) node [pos=0.65,left] {$o_0,x=2$} node
      [pos=0.35,left] {$x:=0$};

\draw [-latex'] (E) -- (F) node [pos=0.5,above] {$l_0; \, b_0$};

\draw [-latex'] (E) -- (G) node [pos=0.5,left] {$l_1$};

\draw [-latex'] (F) -- (H) node [pos=0.5,right] {$l_1$};

\draw [-latex'] (F) -- (A) node [pos=0.4,above,sloped] {$o_0,x=2$} node
      [pos=0.4,below,sloped] {$x:=0$};

\draw [-latex'] (G) -- (H) node [pos=0.5,above] {$l_0 ; \, b_0$};

\draw [-latex'] (G) .. controls +(180:4cm) and +(220:3cm) .. (C) node
      [pos=0.4,right] {$p_0,x=1$} node [pos=0.3,right] {$x:=0$};

\draw [-latex'] (H) .. controls +(0:4cm) and +(330:3cm) .. (B) node
      [pos=0.7,right] {$o_0,x=1$} node [pos=0.6,right] {$x:=0$};

\draw [-latex'] (I) -- (J) node [pos=0.5,above] {$l_0 ; \, b_0$};

\draw [-latex'] (I) .. controls +(10:2.75cm)  .. (E) node
      [pos=0.6,left] {$c_1,x=2$} node [pos=0.5,left] {$x:=0$};

\draw [-latex'] (J) .. controls +(170:2.75cm)  .. (F) node
      [pos=0.6,right] {$c_1,x=2$} node [pos=0.5,right] {$x:=0$};

\end{scope}
\end{tikzpicture}
\caption{$\Bb^{lift}_2$}
\label{FigLift}
\end{figure}

\end{document}